\documentclass[a4paper,11pt]{article}
\usepackage[square,authoryear]{natbib}
\pdfoutput=1
\usepackage{marsden_article}

\newtheorem{assumption}[theorem]{Assumption}
\theoremstyle{remark}
\newtheorem{remark}[theorem]{Remark}

\usepackage[all]{xy}
\begin{document}
\title{An extension of the Dirac and Gotay-Nester theories of constraints for Dirac dynamical systems}
\author{Hern\'{a}n Cendra$^{a,1}$, Mar\'{\i}a Etchechoury$^{b}$\\
and Sebasti\'{a}n J.\ Ferraro$^{c}$}
\footnotetext[1]{Corresponding author. Fax: +54-221-424-5875}
\date{$a$
Departamento de Matem\'{a}tica\\
Universidad Nacional del Sur, Av.\ Alem 1253\\
8000 Bah\'{\i}a Blanca and CONICET, Argentina.\\
uscendra@criba.edu.ar\\
$b$ Laboratorio de Electr\'onica Industrial, Control e
Instrumentaci\'on,\\
Facultad de Ingenier\'{\i}a, Universidad Nacional de La Plata
and\\
Departamento de Matem\'atica, Facultad de Ciencias Exactas,
Universidad Nacional de La Plata.\\
CC 172, 1900 La Plata, Argentina.\\
marila@mate.unlp.edu.ar\\
$c$
Departamento de Matem\'{a}tica and Instituto de Matem\'{a}tica Bah{\'\i}a
Blanca\\
Universidad Nacional del Sur, Av.\ Alem 1253\\
8000 Bah\'{\i}a Blanca,\\
and CONICET, Argentina\\
sferraro@uns.edu.ar}

\maketitle

\begin{abstract}
This paper extends the Gotay-Nester and the Dirac theories of constrained systems in order to deal
with Dirac dynamical systems in the integrable case. Integrable Dirac dynamical systems are viewed as constrained systems where the constraint submanifolds are foliated, the case considered in  Gotay-Nester theory being the particular case where the foliation has only one leaf. A Constraint Algorithm for Dirac dynamical systems (CAD), which extends the Gotay-Nester algorithm, is developed. Evolution equations are written using a Dirac bracket adapted to the foliations and an abridged total energy which coincides with the total Hamiltonian in the particular case considered by Dirac. 
The interesting example of LC circuits is developed in detail.
The paper emphasizes the point of view that Dirac and Gotay-Nester theories are dual and that using a combination of results 
from both theories may have advantages in dealing with a given example, rather than using systematically one or the other. 
\end{abstract}

\section{Introduction}\label{sectionintroduction}

The main purpose of this paper is to study an extension of the Dirac theory of constraints 
(\cite{MR0043724,MR0094205,Dirac1964})
and the Gotay-Nester theory (\cite{MR506712,MR0535369})
for the case in which the equation of motion is a Dirac dynamical system on a manifold $M$, 
\begin{equation}\label{11drracds}
(x, \dot{x})\oplus d\mathcal{E}(x) \in D_x,
\end{equation}
where $D \subseteq TM \oplus T^\ast M$ 
is a given \emph{integrable} Dirac structure on
$M$, which gives a foliation of $M$ by presymplectic leaves, and $\mathcal{E}$ 
is a given function on
$M$, called the \emph{energy}.
One of the features of systems like (\ref{11drracds})
is that they can represent some constrained systems where the constraints appear even if the Lagrangian is nonsingular, as it happens in the case of integrable nonholonomic systems. We will work under explicit hypotheses about regularity, locality and integrability, although some of our results can be applied in more general cases, as indicated in the paragraph \textit{Hypotheses},  below in this section.

In the  Gotay-Nester theory 
the starting point is an equation of the form
\begin{equation}\label{11GN}
\omega(x) (\dot{x}, \,  ) =  d\mathcal{E}(x)
\end{equation}
where $\omega$ is a presymplectic form on 
$M$. 
In fact, in the  Gotay-Nester theory the more general case where $d\mathcal{E}(x)$ is replaced by an arbitrary closed 1-form $\alpha$ is considered, but this will not be relevant here.
\paragraph{Example: Euler-Lagrange and Hamilton's equations.}
Let ${L}\colon TQ \rightarrow \mathbb{R}$ be a Lagrangian,
degenerate or not. Let $\mathcal{E}(q,v,p)=pv-{L}(q,v)$ and
let $\omega \in \Omega^2(TQ \oplus T^*Q)$ be the presymplectic form
$\omega=dq^i \wedge dp_i$ on the Pontryagin bundle $M = TQ\oplus
T^\ast Q$. Then the Euler-Lagrange equations are written equivalently in
the form of equation (\ref{11GN}) with $x = (q,v,p)$. In
fact, we have, in local coordinates,
\begin{align}\label{E-L}
i_{(\dot{q},\dot{v},\dot{p})} dq^i \wedge dp_i
&=  \dot{q}^i
dp_i-\dot{p}_i dq^i \\
\label{E-L1}
d\mathcal{E}
&=  \frac{\partial \mathcal{E} }{\partial q^i}dq^i+\frac{\partial
\mathcal{E} }{\partial p_i}dp_i+\frac{\partial \mathcal{E} }{\partial v^i}dv^i\\
\label{E-L2} &=  -\frac{\partial {L}}{\partial
q^i}(q,v)dq^i+v^idp_i+ \left( p_i-\frac{\partial
{L}}{\partial v^i} \right) dv^i
\end{align}
Using equations (\ref{E-L})--(\ref{E-L2}) we can easily see that
(\ref{11GN}) becomes
\begin{align*}
\dot{q}^i & =  v^i \\
\dot{p}_i & =  \frac{\partial {L}}{\partial q^i}(q,v) \\
0 & =  p_i-\frac{\partial {L}}{\partial v^i}(q,v),
\end{align*}
which is clearly equivalent to the Euler-Lagrange equations. 
The case $\mathcal{E}(q,v,p)= H(q,p)$, where $H$ is a given Hamiltonian on $T^*Q$ of course gives Hamilton's equations.

The idea of using
the Pontryagin bundle to write such fundamental equations of physics appears in \cite{%
Livens1919,
MR0720129,
MR0720130,
MR0720131,
MR2078563,
MR2265464,
MR2265469%
}.\\

Equation (\ref{11GN}) 
is equivalent to equation (\ref{11drracds}) in the case in which the Dirac structure on 
$M$ is the one naturally associated to 
$\omega$, denoted 
$D_{\omega}$, in which case the foliation of $M$ has obviously only one presymplectic leaf.
In general, equation (\ref{11drracds}) may be considered as a collection of equations of the type (\ref{11GN}), one for each presymplectic leaf of the Dirac structure. However, in order to study the initial condition problem this approach might not be appropriate, because different initial conditions might belong to different presymplectic leaves and therefore correspond to different equations of the type (\ref{11GN}), which is not the usual situation.

The algorithm of the Gotay-Nester theory generates a finite sequence of secondary constraint submanifolds
$M \supseteq M_1\supseteq \dots\supseteq M_c$. The final constraint submanifold $M_c$ has the property that every solution curve of (\ref{11GN}) is contained in $M_c$ (in fact, it is the smallest submanifold having that property). Equations of motion are given by restricting the variable $(x, \dot{x})$ in equation (\ref{11GN}) to $TM_c$, and existence of solution curves $x(t)$ for a given initial condition $x_0 \in M_c$ is guaranteed, under the condition that the kernel of
\begin{equation*}
\omega^\flat(x) | T_x M_c \colon  T_x M_c\rightarrow T^\ast_x M
\end{equation*}
has locally constant dimension. 
An important point in the Gotay-Nester approach is that the equations defining the secondary constraint submanifolds are written in terms of $\omega$
(see formula \eqref{equationalgorithGN}),
which makes  the whole algorithm invariant under a group of transformations preserving $\omega$.

In order to solve \eqref{11drracds}, we will develop a constraint algorithm that we call 
CAD (\emph{constraint algorithm for Dirac dynamical systems}).
The CAD extends the Gotay-Nester algorithm and gives a sequence of secondary foliated constraint submanifolds.
The secondary foliated constraint submanifolds can be written in terms of the Dirac structure using formula \eqref{constraintsubmf3}, which generalizes \eqref{equationalgorithGN}.

Let $M \subseteq N$ be an embedding of $M$ into a symplectic manifold $(N, \Omega)$, in such a way that the presymplectic leaves of $M$ are presymplectic submanifolds of $N$ and let, by slight abuse of notation,
$\mathcal{E} \colon N \rightarrow \mathbb{R}$ denote an arbitrary extension of the given energy on $M$.
This kind of framework can be naturally constructed
in many examples. Then one can consider, at least locally, that $M$ is a primary foliated constraint submanifold defined by equating certain constraints to $0$, which generalizes the situation of the Dirac theory of constraints.
One of the results of this paper establishes that there exists a Dirac bracket whose symplectic leaves are \textit{adapted} to the foliations of the primary and final foliated constraint submanifolds.
We also prove that the equations of motion on the final foliated constraint submanifold can be nicely written in Hamiltonian form with respect to this Dirac bracket and the \emph{abridged total energy}, which we introduce in section \ref{subsectionequationsof motion}.

Our extension of the Gotay-Nester and Dirac theories has especially simple and interesting features when applied to LC circuits. We show that the algorithm stops either at $M_1$ or $M_3$, and concrete physically meaningful formulas for the first three constraint submanifolds and for the evolution equations can be given for this case. Other geometric treatments of LC circuits can be found, for instance, in \cite{Bloch1997,MR2265464,MR2265469,CeMaRaYo2011}.

Systems of the type (\ref{11drracds}) are important also in the not necessarily integrable case since they represent, for instance, the Lagrange-d'Alembert equations of nonholonomic mechanics. More references will be given in section \ref{subsectiondiracdynamicalsystems}.
We remark that even though the main focus of this paper is the case where the Dirac structure is integrable, many results, most notably the constraint algorithm CAD and the treatment of nonholonomic systems in section \ref{sectionconstalgforrdirac}, are proven for a not necessarily integrable Dirac structure.\\

Gotay-Nester's theory is more geometric than Dirac's and provides a certain notion of duality or correspondence between the two, in the spirit of the dualities between submanifolds and the equations defining them, manifolds and their ring of functions, Hamiltonian vector fields and their Hamiltonian or Poisson brackets and the collection of their symplectic leaves. Emphasizing this point of view in a common geometric language
and showing that a combination of results 
from both theories may have advantages in a given example, rather than using one or the other,
is an aspect of the paper.

From this point of view, the notion of second class constraints is related to the notion of second class constraint submanifold, 
the latter being determined in this paper as a submanifold that is tangent to a second class subbundle of a certain tangent bundle. They are symplectic submanifolds (\cite{MR0358860,BursztynIHES}). The first class constraint submanifolds are coisotropic submanifolds (\cite{BursztynIHES}).

The presence of primary and secondary first class constraints, which in Dirac's work is connected to the principle of determinacy and to the notion of physical variables, is explicitly related here to certain manifolds 
$\widetilde{M}$ and $\bar{M}$
in a commutative diagram, in Theorem
\ref{lemmmma}. The diagonal arrow in this diagram, which is a submersion, is dual to the diagonal arrow in the diagram in 
Lemma \ref{lemalema}, which is a surjective Poisson algebra homomorphism. 
In some examples, instead of applying the Dirac method to write equations of motion, it may be simpler to proceed in two stages. First apply the Gotay-Nester algorithm to find the final constraint submanifold, which, as proven in the Gotay-Nester theory, coincides with the submanifold defined by equating the final constraints obtained in the Dirac algorithm to $0$. Then switch to the Dirac viewpoint and write equations of motion using the Dirac bracket.  Also, in order to calculate the Dirac bracket at points of the final constraint submanifold, using the symplectic leaf that contains that submanifold may be easier than using Dirac's expression.

\paragraph{The point of view of IDEs.}
In order to study the problem of finding solution curves of general Dirac dynamical systems (\ref{11drracds}) the more general theory of \textit{Implicit Differential Equations} may be useful. 
An Implicit Differential Equation (IDE) on a manifold $M$ is written as
\begin{equation}\label{defsolIDE}
\varphi(x, \dot{x}) = 0.
\end{equation}
A solution of (\ref{defsolIDE}) \textit{at a point} $x$ is a
vector $(x, v) \in T_xM$ satisfying (\ref{defsolIDE}). A \emph{solution
curve}, say $x(t)$, $t \in (a,b)$, must satisfy, by definition, that $(x(t),
\dot{x}(t))$ is a solution at $x(t)$ for all $t \in (a,b)$. 

Basic questions such as existence, uniqueness or
extension of solutions for a given initial condition are not completely answered yet, although many partial results have
been established for certain classes of IDE.
One approach is to use
a \emph{constraint algorithm}, which consists of a decreasing sequence of constraint
submanifolds $M \supseteq M_1 \supseteq \ldots \supseteq M_c$ defined as follows,
\[
 M_{k+1} : = \{x \in M_k \,|\, \mbox{there exists } (x,v) \in TM_k
\mbox{ such that}\,\,\, \varphi (x,v) = 0\},
\]
with $M_0 := M$.
This algorithm, which obviously uses only the differentiable structure regardless of any other structure which may be present, like presymplectic, Poisson or Dirac structures, represents a differential geometric aspect underlying the algorithms  of Gotay-Nester, Dirac or CAD.
To ensure that
each $M_k$ is a submanifold and that
the algorithm stops after a finite number of steps one may choose to assume certain conditions of the type ``locally constant rank'' conditions. Then the original IDE
is reduced to an equivalent ODE depending on parameters on the 
\emph{final constraint submanifold} 
$M_c$. In fact, 
by construction, $M_c$ is characterized by the property that it is the smallest submanifold that contains all solutions curves of the given IDE. Therefore, if $M_c$ is empty, there are no solution curves.

Dirac's original work has a wide perspective from the physical point of view, with connections to classical, quantum and relativistic mechanics. 
However, from the point of view of abstract IDE theory, and in a very concise way, 
we may say that a combination of the Dirac and the Gotay-Nester methods shows how to
transform a given Gotay-Nester equation (\ref{11GN}) into an equivalent ODE depending on parameters (\ref{eq_dirac_evolution_thm}) on a final constraint submanifold, while what we show in this paper is how to transform a given Dirac dynamical system (\ref{11drracds}) into an equivalent ODE (\ref{eq:vector_field_foliated}) depending on parameters on a final foliated constraint submanifold.

Some more comments on the connection between Dirac's and Gotay-Nester ideas and IDE are in order.
We can compare \cite{%
MR0043724,
MR506712,
MR0535369,
MR1272402,
MR2006216,
MR2277353%
},
to see how the idea of applying a constraint algorithm works in different contexts.
In \cite{MR2277353}, one works in the realm of subanalytic sets;
in \cite{MR506712} and \cite{MR0535369}
one works with presymplectic manifolds; in \cite{MR2006216}
one works with complex algebraic manifolds; \cite{MR0043724} uses Poisson
brackets; in \cite{MR1272402} some degree of differentiability of the basic data
is assumed, and, besides, some constant rank hypotheses are added, essentially to
ensure
applicability of a certain constant rank theorem. 
Some relevant references for general IDEs connected to physics or control theory, which show a diversity of geometric or analytic methods or a combination of both are
 \cite{%
MR0751527,
MR1102803,
MR1245325,
MR1325842,
MR1401231,
MR1618157,
MR1447117,
MR2018120%
}.

\paragraph{Hypotheses.}
As for regularity,
we will work in the $C^\infty$ category. 
Throughout the paper we assume that the subsets 
$M_k$ appearing in 
sequences of the type
$M \supseteq M_1 \supseteq \dots \supseteq M_c$
that are generated by some constraint algorithm, are submanifolds regularly defined by equating some functions (constraints) to $0$, and that the sequence stops. More regularity conditions like Assumptions 
\ref{K1}, \ref{K2},
\ref{Lambda},
\ref{2s_constant},
\ref{assumptionFOLIATION},
and some others will be introduced when needed along the paper. 

Our results will be of a local character, but for some of them, like the notion of second class submanifolds, it is indicated how to define them globally.
The usage of local coordinates is almost entirely avoided and basic facts in symplectic geometry or Poisson algebra arguments are used instead.

The condition of integrability of the Dirac structure appearing in 
(\ref{11drracds})
has its own interest.
However, the CAD does not assumes integrability  and can be applied for instance to general nonholonomic systems. On the other hand, in the non-integrable case certain brackets that we study would not satisfy the Jacobi identity, and will not be studied in this paper.

\paragraph{Structure of the paper.}
The first part of the paper, which includes sections \ref{section7} and
\ref{mainresultsofdiracandgotaynester},
contains a review of the Dirac and Gotay-Nester methods.

In sections \ref{sectiondiracstructures} to \ref{sectionanextensionofetc} we develop our main results, extending the Gotay-Nester and Dirac theories. 
Section  \ref{sectiondiracstructures} is devoted to a review of basic facts on Dirac structures and Dirac dynamical systems (\cite{CeMaRaYo2011}).
The notion of Dirac structure (\cite{MR951168,MR998124,MR2534987}) gives a new possibility of understanding and extending the theory of constraints, which is the main purpose of the present paper. 
 In section \ref{sectionconstalgforrdirac} we develop our CAD algorithm. We do it for the general case of a not necessarily integrable Dirac structure so that one can apply CAD to general noholonomic systems. 
In section \ref{sectionexamples} we study the example of integrable nonholonomic systems, and LC circuits are viewed as a particular case. 
In section 
\ref{sectionanextensionofetc} we show how to extend the Dirac theory for the case of a Dirac dynamical system (\ref{11drracds}).

\section{A brief review of the Dirac and Gotay-Nester theories}\label{section7}

In this section we review, without proof,  basic facts belonging to the Dirac and the Gotay-Nester theories.

The Dirac method starts 
with a given submanifold $M$, called the primary constraint submanifold, of a symplectic manifold 
$(N, \Omega)$
defined by equating the primary constraints to $0$, and a given energy
$\mathcal{E} \colon N \rightarrow \mathbb{R}$. In the original work of Dirac $N = T^\ast Q$, $\Omega$ is the canonical symplectic form, $M$ is the image of the Legendre transformation and 
$\mathcal{E}$ is the Hamiltonian defined by Dirac.
However, locally, any symplectic manifold is an open subset of a cotangent bundle, therefore the original Dirac formalism can be applied to the, seemingly more general, case described above, at least for the local theory.\\

An interesting variant of the original situation considered by Dirac is the following.
Consider the canonical symplectic manifold $N = T^\ast TQ$ with the canonical
symplectic form
$\Omega$,
and let the primary constraint be
$M = TQ \oplus T^\ast Q$, canonically embedded in $N$ via the map given in local
coordinates
$(q,v,p,\nu)$ of $N$ by $\varphi(q,v,p)=(q,v,p,0)$. In particular, $M$ is
defined regularly by the equation
$\nu = 0$. If $\omega$ is the presymplectic form on $M$ obtained by pulling back the canonical symplectic form of $T^*Q$ using the natural projection, then one can show that $\varphi^*\Omega=\omega$. 
The embedding 
$\varphi$
is globally defined (see Appendix \ref{Pontryagin_embedding} for details). 
The number $pv$ is a well-defined function on $M$ and it can be naturally extended
to a function on a chart with coordinates $(q,v,p,\nu)$, but this does not
define a global function on $N$ consistently.
On the other hand, it can be extended to a smooth function on $N$ using partitions of unity and any such
extension will give consistent equations of motion. 
In this paper we will not consider global aspects. In any case, for a given Lagrangian
$L \colon TQ \rightarrow \mathbb{R}$ we can take 
$\mathcal{E} = pv - L(q,v)$.

\begin{remark}
For a given presymplectic manifold $(M,
\omega)$ one can always find an embedding $\varphi$ into a symplectic manifold 
$(P, \Omega)$ such that $\varphi^*\Omega=\omega$. Moreover, this embedding can also be chosen such that it is coisotropic, meaning that $\varphi (M)$ is a coisotropic submanifold of $P$ (see \cite{MR641768}). However, we should mention that the embedding given above is not coisotropic.
\end{remark}

The Dirac and the Gotay-Nester algorithms can be studied independently. On the other hand, they are related as follows.
For a given system  (\ref{11GN}) choose
a symplectic manifold
$(N, \Omega)$ in such a way that $(M, \omega)$ is a presymplectic submanifold
and $\mathcal{E}$ (using a slight abuse of notation) is an arbitrary extension of
$\mathcal{E}\colon M \rightarrow \mathbb{R}$ to $N$.
Moreover, assume that
$M$ is defined regularly by a finite set of equations 
$\phi^{(0)}_i = 0$, $i = 1, \dots,a_0$, where each  $\phi^{(0)}_i$  is a \textit{primary constraint}.
The Dirac algorithm gives a sequence of secondary constraints
$\phi^{(k)}_i$, $i = 1, \dots,a_k$, which defines regularly a sequence of secondary constraint submanifolds  $M_k$, $k = 1, \dots,c$,
by equations $\phi^{(k)}_i = 0$, $i = 1, \dots,a_k$, which coincide with the ones
given in the Gotay-Nester algorithm.

\subsection{A brief review of Dirac's theory}\label{briefreview}
Dirac's theory of constraints has been extensively studied from many different
points of view and extended in several directions. Part of those
developments in the spirit of geometric mechanics is contained in the following references, but the list is far from being complete, \cite{%
MR0434047,
MR506712,
MR0535369,
MR0569299,
MR0563632,MR0935624,
MR0720129,
MR0720130,
MR0720131,
MR0751527,
MR1212006,
MR1326085,
MR0796226,
MR0894631,
MR0914056,
MR0941034,
MR1082418,
MR1191617,
MR1303063,
MR1353155,
MR1450895,
MR1447117,
MR1692290,
MR2589829,
MR2817606,
MR2794294%
}.

As we have explained above we will work in a general context of a given symplectic manifold
$(N, \Omega)$, $M \subseteq N$ and
$\mathcal{E}\colon N \rightarrow \mathbb{R}$, where the primary constraint submanifold $M$ is regularly defined by equations
$\phi^{(0)}_i = 0$, $i = 1,\dots,a_0$ on $N$.

The Dirac constraint algorithm goes as follows.
One defines the \textit{total energy}
$\mathcal{E}_T = \mathcal{E} + \lambda_{(0)}^i \phi^{(0)}_i$.
The preservation of the primary constraints is written
$\{\phi^{(0)}_i, \mathcal{E}_T\}(x) = 0$, $i = 1,\dots,a_0$, $x \in M_0$,
or
\begin{equation*}%\label{1equationequationforlambda}
\{\phi^{(0)}_i, \mathcal{E}\}(x) + \lambda_{(0)}^j\{\phi^{(0)}_i, \phi^{(0)}_j\}(x)= 0, \,
i, j = 1,\dots,a_0,\, x \in M_0.
\end{equation*}
Then $M_1$ is defined by the condition that $x \in M_1$ if and only if
there exists $\lambda_{(0)} = (\lambda_{(0)}^1,\dots,\lambda_{(0)}^{a_0})$
such that the system of equations
$\phi^{(0)}_i (x)= 0$,  
$\{\phi^{(0)}_i, \mathcal{E}_T\}(x) = 0$, $i = 1,\dots,a_0$, 
is satisfied. The submanifold
$M_1$  is defined by equations
$\phi^{(1)}_i = 0$, $i = 1,\dots,a_1$, 
where each $\phi^{(1)}_i $ is a \textit{secondary constraint}, by
definition.
By proceeding iteratively one obtains a sequence
$M_0 \supseteq M_1 \supseteq\dots\supseteq M_c$. Then there are
\textit{final constraints}, say
$\phi^{(c)}_i$, $i = 1,\dots,a_c$, defining a 
submanifold $M_c$ by equations $\phi^{(c)}_i = 0$, $i = 1,\dots,a_c$, 
called the \textit{final constraint submanifold},
and the following condition is satisfied:
for each
$x \in M_c$ there exists
$(\lambda_{(0)}^1,\dots,\lambda_{(0)}^{a_0})$ such that
\begin{equation}\label{equationequationforlambda}
\{\phi^{(c)}_i, \mathcal{E}\}(x) + \lambda_{(0)}^j\{\phi^{(c)}_i, \phi^{(0)}_j\}(x)= 0, \,i
= 1,\dots,a_c,\,   j = 1,\dots,a_0.
\end{equation}

For each $x \in M_c$ the space of solutions of the linear system of equations
(\ref{equationequationforlambda}) in the unknowns 
$ \lambda_{(0)}^j$ is an affine subspace of 
$ \mathbb{R}^{a_0}$, called $\Lambda_x^{(c)}$,
whose dimension is a locally constant function
$d^{(c)}(x) = a_0 - \operatorname{rank}(\{\phi^{(c)}_i,
\phi^{(0)}_j\}(x))$. One can locally choose $d^{(c)}(x)$ unknowns as being free
parameters and the rest will depend affinely on them. Then the solutions of
(\ref{equationequationforlambda}) form an affine bundle 
$\Lambda^{(c)}$ over $M_c$.
After replacing  $\lambda_{(0)} \in \Lambda^{(c)}$ in the expression of the total
energy,
the corresponding Hamiltonian vector field, 
\begin{equation}\label{X_Dirac_classical}
X_{\mathcal{E}_T}(x) = X_{\mathcal{E}}(x) + \lambda_{(0)}^jX_{\phi^{(0)}_j}(x),
\end{equation}
$x \in M_c$, which will depend on the free unknowns, will be tangent to $M_c$.
Its integral curves, for an arbitrary choice of a time dependence of the free
unknowns, will be solutions of the equations of motion, which is the main
property of the final constraint submanifold $M_c$ from the point of view of
classical mechanics. 
The lack of uniqueness of solution for a given initial condition in $M_c$, given
by the presence of free parameters, indicates, according to Dirac, the
nonphysical character of some of the variables. In our context the physical variables will be given a geometric meaning.
\begin{remark}
Dirac introduces the notion of weak equality for functions on $T^\ast Q$. Two
such functions are \textit{weakly equal}, denoted $f \approx g$, if
$f|M_c = g|M_c$. Then, for instance 
$\phi^{(k)}_j \approx 0$.
If $f \approx 0$ then $f = \nu^i \phi^{(c)}_i$, for some functions $\nu^i$ on
$T^\ast Q$ and conversely.
Since we use the notion of a constraint submanifold, in particular
the final constraint submanifold, we prefer not to use the notation $\approx$.
\end{remark}

Now let us make some comments on the notions of first class and second class
constraints.
The rank of the skew-symmetric matrix 
$(\{\phi^{(c)}_i, \phi^{(c)}_j\}(x)), i,j = 1,\dots,a_c,  $
 is necessarily even, say,
$2s$, and it is assumed to be constant on $M_c$, as part of our regularity conditions; for our results we will assume later a stronger condition (Assumption \ref{2s_constant}).
One can choose, among the $\phi^{(c)}_i$, $i = 1,\dots,a_c$, $2s$ functions $\chi_j$, $j = 1,\dots,2s$ , and replace the rest of the $\phi^{(c)}_i$ by $\phi^{(c)}_i + \alpha^j_i \chi_j $ with appropriate functions $\alpha^j_i$ in such a way that
$\psi_i = 0$, $\chi_j = 0$, 
define 
$M_c$ 
regularly and, besides,
$\{\psi_i, \psi_{i^\prime}\}(x) = 0$,  
$\{\psi_i, \chi_j\}(x) = 0$,
$\det (\{\chi_j, \chi_{j^\prime}\}(x) ) \neq 0$, for $i, i^{\prime} = 1,\dots,a_c - 2s$, $j,j^\prime =
1,\dots,2s$
and
$x \in M_c$. 
The $\phi^{(c)}_j$ are linear combinations with smooth coefficients of the 
$\chi_j$  and $\psi_i$, and conversely.
The functions $\chi_j$, $j = 1,\dots,2s$,
are called \textit{second class constraints}
and the functions
$\psi_i$, $i = 1,\dots,a_c - 2s$, 
are called \textit{first class constraints}.

More generally, any function $\rho$ on $T^\ast Q$
satisfying 
$\rho |M_c = 0$,  
$\{\rho, \psi_i\} |M_c = 0$, 
$\{\rho, \chi_j\} |M_c = 0$,
is a first class constraint with respect to the submanifold 
$M_c$,
by definition. Any function $g$ on $T^\ast Q$ satisfying
$\{g, \psi_i\} |M_c = 0$,
$\{g, \chi_j\} |M_c = 0$,
is a \textit{first class function}, by definition. For instance, the
total energy $\mathcal{E}_T$ is a first class 
function.

Now define the energy $\mathcal{E}_c$ in terms of $\psi_i$, $\chi_j$,
$i = 1,\dots,a_c-2s$, $j = 1,\dots,2s$, as
\[
\mathcal{E}_c = \mathcal{E} + \lambda^i \psi_i + \mu^j \chi_j.
\]
The preservation of the constraints for the evolution generated by $\mathcal{E}_c$ can be
rewritten as
$\{\psi_i,\mathcal{E}_c\}(x) = 0$, which is equivalent to $\{\psi_i, \mathcal{E}\}(x) = 0$ for all
$x \in M_c$, and $\{\chi_j,\mathcal{E}_c\}(x) = 0$,  
for all $x \in M_c$. The latter is equivalent to
\begin{equation*}%\label{3equationequationforlambda}
\{\chi_i, \mathcal{E}\}(x) + \mu^j\{\chi_i, \chi_j\}(x)= 0,\  i, j = 1,\dots,2s, 
\end{equation*}
for all $x \in M_c$,
which determines the $\mu^j$ as well-defined functions on $M_c$.
Then the solutions 
$(\mu(x), \lambda)$
form an affine bundle with base
$M_c$
and whose fiber, parametrized by the free parameters
$\lambda$,
has dimension
$a_c - 2s$.

Any section $(\mu(x), \lambda(x))$ of this bundle determines 
$\mathcal{E}_c$ as a first class function. 
This means that 
$X_{\mathcal{E}_c}(x) \in T_x M_c$, for each $x \in M_c$, and therefore a solution curve of
$X_{\mathcal{E}_c}$ is contained in $M_c$ provided that the initial condition belongs to
$M_c$. The function $\mathcal{E}_c$ is essentially the \textit{extended Hamiltonian} defined by Dirac.

Dirac defines an interesting bracket, now called the \textit{Dirac 
bracket},
\[
\{F,G\}^*= \{F,G\} - \{F,\chi_i\}c^{ij}\{\chi_j, G\},
\]
which is defined on an open set in $T^\ast Q$ containing
$M_c$, 
where
$c^{ij}$,
which by definition is the inverse matrix of
$\{\chi_i, \chi_j\}$,
is defined.
The Dirac bracket is a Poisson bracket and has the important property
that for \textit{any} function $F$ on $T^\ast Q$,
the condition 
$\{F,\chi_j\}^*  = 0$, $j = 1,\dots,2s$,
is satisfied on a neighborhood of
$M_c$, which implies that 
$\dot{F} = \{F, \mathcal{E}_c\} =  \{F, \mathcal{E}_c\}^*$,
for any function $F$.
Besides,
$\{\psi_j,\psi_i\}^* = 0$, $i,j = 1,\dots,a_c-2s$,
on
$M_c$. 
Because of this, one may say that, with respect to the Dirac bracket, all the
constraints
$\chi_j$, $j = 1,\dots,2s$ 
and 
$\psi_i$, $i = 1,\dots,a_c-2s$, 
behave like first class constraints, with respect to $M_c$. 
\subsection{A brief review of the Gotay-Nester algorithm}\label{The Gotay-Nester Algorithm}
We will need just a few basic facts from the Gotay-Nester theory.

In order to find solution curves to (\ref{11GN}) we can apply the general algorithm for IDEs described in the introduction, and the final IDE can be written
\begin{equation*}\label{GNGNfinal}
\begin{split}
\omega(x)(\dot{x},\,\,)&=d\mathcal{E}(x)\\
(x, \dot x)&\in TM_c.
\end{split}  
\end{equation*}

Let $\omega_c$ be the pullback of $\omega$ to $M_c$. If $\omega_c$ is symplectic, one obtains the simpler equivalent equation
\begin{equation}\label{eq:Ham_form_sympl_omegac}
\omega_c(x)(\dot{x},\,)=d(\mathcal{E}| {M_c})
\end{equation}
on $M_c$, which is in Hamiltonian form. However, one must be aware that if $\omega_c$ is degenerate then
(\ref{eq:Ham_form_sympl_omegac}) is not equivalent in general to (\ref{GNGNfinal}).

The equation on $M_c$,
\begin{equation}\label{eckgotay}
\omega(x)(X,\,)= d\mathcal{E}(x),
\end{equation}
where $X \in TM_c$, defines an affine distribution on $M_c$,
more precisely, one has an affine bundle
$S^{(c)}$ with base $M_c$
whose fiber
$S^{(c)}_x$
at a given point
$x \in M_c$
is, by definition,
\begin{equation*}%\label{Scx}
S^{(c)}_x = \{X \in T_x M_c \mid \mbox{(\ref{eckgotay}) is satisfied}\}.
\end{equation*}

The equivalence between the Dirac and the Gotay-Nester algorithms can be made explicit as an isomorphism of affine bundles, as follows.
The affine bundles $S^{(c)} \rightarrow M_c$ and $\Lambda^{(c)} \rightarrow M_c$ are isomorphic over the identity on $M_c$, more precisely, the isomorphism $\Lambda^{(c)} \rightarrow S^{(c)}$ is given by
$\lambda_{(0)} \rightarrow X_{\mathcal{E}_T}$, using equation (\ref{X_Dirac_classical}).
In particular, the rank of $S^{(c)}$ is $d^{(c)}$.

\paragraph{Describing the secondary constraints using $\omega$.}
The constraint manifolds $M_k$ defined by the
algorithm can be described by \emph{equations written in terms of the presymplectic form} $\omega$, which is a simple but important idea, because those equations will obviously be
invariant under changes of coordinates preserving $\omega$.
Depending on the nature of $\omega$ one may obtain analytic, smooth, linear, etc., equations, which may
simplify matters in given examples. The context of
reflexive Banach manifolds is used in \cite{MR506712} and \cite{MR0535369}.
 
The condition defining the subsets
$M_{k+1}$, $k = 0,1,\dots$ namely,
\[
\operatorname{i}_{v}\omega (x) = d\mathcal{E}(x) \mbox{ has at least
one solution } v \in T_x M_k,
\]
is equivalent to 
\[
d \mathcal{E}(x) \in \left(T_x M_k
\right)^{\flat}
\]
or, since $\left(T_x M_k
\right)^{\flat} = \left(\left(T_x M_k
\right)^{\omega}\right)^\circ$,
\begin{equation}\label{equationalgorithGN}
\langle d\mathcal{E}(x), \left(T_x M_k
\right)^{\omega}\rangle = \{0\}.
\end{equation}
\section[First and second class constraints, constraint submanifolds...]{First and second class constraints, constraint submanifolds and equations of motion}\label{mainresultsofdiracandgotaynester}

In this section
we will describe some results on first class and second class constraints and constraint submanifolds and also equations of motion which will be useful for our extension of those notions, to be developed in section~\ref{sectionanextensionofetc}. 

\subsection{Linear symplectic geometry}\label{linearsymplecticgeometry}
The following results about linear symplectic geometry are an essential part of many of the arguments that we use in the paper, since under our strong regularity assumptions many of them are essentially of a linear character.

\begin{lemma}\label{lemma6.1}
Let
$(E,\Omega)$
be a symplectic vector space of dimension
$2n$,
$V\subseteq E$
a given subspace. For a given basis
$\alpha_i$, $i=1,\ldots, r$
of
$V^\circ$,
let
$X_i=\alpha_i^{\sharp},\,i=1,\ldots r$. Then the rank of the
matrix $[\alpha_i(X_j)]$ is even, say $2s$, and $X_i,\,i=1,\ldots,
r$ form a basis of $V^{\Omega}$. Moreover, the basis
$\alpha_i,\,i=1,\ldots, r$ can be chosen such that for all $j=1,\dots,r$ 
\begin{align*}
\alpha_i(X_j) & = \delta_{i,j-s},\,\,1 \leq i \leq s \\
\alpha_i(X_j) & = -\delta_{i-s,j},\,\,s+1 \leq i \leq 2s \\
\alpha_i(X_j) & = 0,\,\,2s < i \leq r.
\end{align*}
\end{lemma}
\begin{proof} Consider the subspace $V^\Omega=(V^\circ)^\sharp$. By
a well-known result there is a basis $X_i,\,i=1,\ldots,r$ of
$V^\Omega$ such that for all $j=1,\dots,r$
\begin{align*}
\Omega(X_i,X_j) & = \delta_{i,j-s},\,\,1 \leq i \leq s \\
\Omega(X_i,X_j) & = -\delta_{i-s,j},\,\,s+1 \leq i \leq 2s \\
\Omega(X_i,X_j) & = 0,\,2s < i \leq r
\end{align*}
then take $\alpha_i=X_i^{\flat}$. The first part of the lemma is easy to prove using this.
\end{proof}
\begin{lemma}\label{lemma6.2}
Let $\alpha_i$, $i=1, \dots, r$ be a basis of $V^\circ$ having the
properties stated in Lemma \ref{lemma6.1}. Then $X_i$, $i=2s+1,\dots,
r$ form a basis of $V \cap V^\Omega$.
\end{lemma}
\begin{proof} Let $X=\lambda^i X_i$ be an arbitrary vector in $V^\Omega$.
Now $\lambda^iX_i \in V \cap
V^\Omega$ iff $\alpha_j(X)=\lambda^i(\alpha_j(X_i))=0,\,j=1,\ldots, r$. Since the first $2s$
columns of the matrix $[\alpha_i(X_j)]$ are linearly independent and the rest are zero,
we must have $\lambda^i=0$, for $1 \leq i \leq 2s$, and $\lambda^i$, $i=2s+1,\ldots, r$ are arbitrary. This means that $V \cap V^\Omega$ is generated by
$X_i$, $i=2s+1,\ldots, r$.
\end{proof}
\begin{corollary}\label{cordim}
$\dim V \cap V^\Omega=r-2s$.
\end{corollary}
\begin{proof}
Immediate from Lemma \ref{lemma6.2}.
\end{proof}
 Let $\omega$ be the pullback of $\Omega$ to $V$ via the inclusion. Then $(V,\omega)$ is a presymplectic
space. In what follows, the ${}^\flat$ and ${}^\sharp$ operators are taken with respect to $\Omega$ unless specified otherwise.
\begin{lemma}\label{lemmaperp}
$V^\omega=V \cap V^\Omega$.
\end{lemma}
\begin{proof}
$X \in V^\omega$
iff
$\omega(X,Y)=0,\,\forall\,Y \in
V$
iff
$\Omega(X,Y)=0,\,\forall\,Y \in V$.
This is equivalent to
$X \in V \cap
V^\Omega$.
\end{proof}
\begin{lemma}\label{lemma6.5}
Let
$\gamma_i,\,i=1,\ldots r$
be a given basis of
$V^\circ$
and
let
$Y_i=\gamma_i^\sharp,\,i=1, \ldots r$.
Let
$\beta \in E^*$
be
given. Then the following conditions are equivalent.
\begin{itemize}
  \item[(i)] $\beta(V^\omega)=0$.
\item[(ii)] The linear system
\begin{equation}
\label{linsyst} \beta(Y_i)+\lambda^j\gamma_j(Y_i)=0
\end{equation}
has solution
$\lambda = (\lambda^1,\dots,\lambda^r)$.
\end{itemize}
\end{lemma}
\begin{proof}
Let us show that (\ref{linsyst}) has solution
$(\lambda^1,\ldots,\lambda^r)$
iff the system
\begin{equation}
\label{linsyst2} \beta(X_k)+\mu^l\alpha_l(X_k)=0
\end{equation}
has solution
$(\mu^1,\ldots,\mu^r)$,
where
$k,l=1,\ldots r$
and
$\alpha_l$
is a basis satisfying the conditions of Lemma
\ref{lemma6.1}.
Since
$Y_i,\,i=1,\ldots r$
and
$X_k,\,k=1,\ldots r$
are both bases
of
$V^\Omega$ there is an invertible matrix $[A_k^i]$ such that
$X_k=A_k^iY_i$. Let $[B_i^l]$ be the inverse of $[A_k^i]$, so
$Y_i=B_i^lX_l$. Assume that (\ref{linsyst}) has solution
$\lambda^j,\,j=1,\ldots r$. We can write (\ref{linsyst}) as
\[
\beta(Y_i)+\lambda^j\Omega(Y_j,Y_i)=0,\,i=1,\ldots, r.
\]
Using this we have that for $k=1,\ldots, r$ 
\begin{equation*}
\begin{split}
0&=\beta(A_k^iY_i)+\lambda^j\Omega(Y_j,A_k^iY_i)=
\beta(X_k)+\lambda^j\Omega(Y_j,X_k)\\
&=\beta(X_k)+\lambda^j\Omega(B_j^lX_l,X_k)=\beta(X_k)+\mu^l\Omega(X_l,X_k)
\end{split}
\end{equation*}
where $\mu^l=\lambda^jB_j^l$. This means that the system
(\ref{linsyst2}) has solution. The converse is analogous.
Using this, Lemmas \ref{lemma6.2} and \ref{lemmaperp}, and the form of the coefficient matrix $[\alpha_l(X_k)]$ in Lemma \ref{lemma6.1},
the proof that (\ref{linsyst2}) has solution
$\mu = (\mu^1,\dots,\mu^r)$
iff $\beta(V^{\omega})=0$ is easy and is left to the reader.
\end{proof}
\begin{lemma}\label{lemma6.6}
Consider the hypotheses in Lemma \ref{lemma6.5}. Then the solutions to
\begin{equation}\label{presymplecticsharp}
i_{X}\omega=\beta|V  
\end{equation}
(if any) are precisely $X=\beta^\sharp+\lambda^jY_j$, where
$(\lambda^1,\ldots,\lambda^r)$ is a solution to (\ref{linsyst}). A solution to (\ref{presymplecticsharp}) exists if and only if $\beta(V^{\omega})=0$.
If $\omega$ is symplectic then (\ref{linsyst}) and (\ref{presymplecticsharp}) have a unique solution and if, in
addition, 
$\beta^\sharp \in V$, then 
$\lambda^1 = 0$, \dots, $\lambda^r = 0$
and $\beta^\sharp$ coincides with $X  = (\beta | V)^{\sharp_{\omega}}$ defined by (\ref{presymplecticsharp}).
\end{lemma}
\begin{proof} Since $Y_j$, $j=1,\dots, r$ form a basis of $V^\Omega$
we have that $(\lambda^1,\dots,\lambda^r)$ is a solution to (\ref{linsyst}) iff
$(\beta + \lambda^j\gamma_j)(V^{\Omega})=0$ iff $\beta +
\lambda^j\gamma_j \in V^\flat$ iff $\beta^\sharp+\lambda^jY_j \in
V$.
Now, let $X=\beta^\sharp + \lambda^jY_j$, where
$(\lambda^1,\ldots,\lambda^r)$ satisfies (\ref{linsyst}). Then we
have $X \in V$ as we have just seen and we also have
\[
i_{X}\omega = (i_{X}\Omega) | V = X^\flat | V =
(\beta + \lambda^j \gamma_j)| V = \beta | V,
\]
since $\gamma_j,\,j=1,\dots, r$ generate $V^{\circ}$.
We have proven that $X$ is a solution to (\ref{presymplecticsharp}). To prove
that every solution $X$ to (\ref{presymplecticsharp}) can be written as
before, we can reverse the previous argument. 
Using this, it is clear that if $\omega$ is symplectic then (\ref{linsyst}) has
unique solution, in particular, we have that 
$\det \left(\gamma_j(Y_i)\right) \neq 0$.
If, in addition,
$\beta^\sharp \in V$ then $\lambda^jY_j=X-\beta^\sharp\in V$. Since $Y_j$, $j = 1,\dots,r$ is a basis of $V^\Omega$, using Lemma \ref{lemmaperp} and the fact that $V^\omega=\{0\}$ we get that $\lambda^j=0$ for $j = 1,\dots,r$.
\end{proof}
\begin{corollary}\label{corollary6.7}
Let $\Lambda=\{\lambda\mid \lambda \mbox{ satisfies
(\ref{linsyst})}\}$. Then $\dim \Lambda=r-2s=\dim
\ker{\omega}$.
\end{corollary}
\begin{proof} $\ker{\omega}=V^\omega$, which has dimension $r-2s$ from Corollary \ref{cordim} and Lemma \ref{lemmaperp}.
On the other hand the
dimension of the subspace of $\lambda$ satisfying (\ref{linsyst})
is clearly also $r-2s$, since the coefficient matrix has rank $2s$.
\end{proof}

\subsection[Poisson-algebraic and geometric study of constraints...]{Poisson-algebraic and geometric study of constraints and constraint submanifolds.}\label{poisson-dirac-gotay-nester}

We shall start with the constrained Hamiltonian system 
$(N, \Omega, \mathcal{E}, M)$, where $(N, \Omega)$ is a symplectic manifold, 
$\mathcal{E} \colon N \rightarrow \mathbb{R}$ is the energy and $M \subseteq N$ is the primary constraint submanifold.
The equation to be solved, according to the Gotay-Nester algorithm, is
\begin{equation}\label{diracagainagain}
\omega(x)(X,\,\,) = d\mathcal{E}(x) | T_x M,
\end{equation}
where $X \in T_x M_c$ and $x \in M_c$, $M_c$ being {\it the final
constraint.}
Let $\omega_c$ be the pullback of $\Omega$ via the inclusion
of $M_c$ in $N$. Since $\omega_c$ is presymplectic, $\ker \omega_c$
is an involutive distribution.
{}From now on we will assume the following.

\begin{assumption}\label{K1}
 The distribution $\ker \omega_c$ has
constant rank and defines a regular foliation $K_c$, that is, the
natural map $p_{K_c} \colon  M_c \rightarrow \bar{M}_c$, where
$\bar{M}_c = M_c / K_c$ is a submersion.
\end{assumption}

\begin{lemma}\label{lemma7.1} The following assertions hold:

\textbf{(a)} There is a uniquely defined symplectic form
$\bar\omega_c$ on $\bar{M}_c$ such that $p_{K_c}^\ast
\bar\omega_c = \omega_c$.

\textbf{(b)} Let $\bar{X}$ be a given vector field on $\bar{M}_c$.
Then there is a vector field $X$ on $M_c$ that is
$p_{K_c}$-related to $\bar{X}$.

\textbf{(c)} Let $\bar{f} \in \mathcal{F}(\bar{M_c})$. Then there
exists a vector field $X$ on $M_c$ such that $X$ is
$p_{K_c}$-related to $X_{\bar{f}}$, and for any such vector field
$X$ the equality $\omega_c(x)(X,\, \,) =
d(p_{K_c}^\ast \bar{f})(x)$ holds for all $x\in M_c$.

\textbf{(d)} Let $X_{x_0}\in T_{x_0}M_c$. Then one can choose the
function $\bar{f}\in \mathcal{F}(\bar{M}_c)$ and the vector field
$X$ in \textbf{(c)} in such a way that $X(x_0) = X_{x_0}$.

\end{lemma}
\begin{proof}
\textbf{(a)}
By definition, the leaves of the foliation $K_c$
are connected submanifolds of $M_c$, that is, each
$p^{-1}_{K_c}(z)$, $z \in \bar{M}_c$, is connected.
For $z \in \bar{M}_c$, let $x \in M_c$ such that
$p_{K_c}(x)=z$. For $\bar{A},\,\bar{B} \in
T_{z}{\bar{M}_c}$, as $p_{K_c}$ is a submersion, there are $A,\,B
\in M_c$ such that
$T_{x}{p_{K_c}}{A}=\bar{A},\,T_{x}{p_{K_c}}{B}=\bar{B}$.
We define
$\bar{\omega}_{c}(z)(\bar{A},\bar{B})=\omega_{c}(x)(A,B)$.
To prove
that this is a good definition observe first that it is a consistent definition
for fixed $x$, which is easy to prove, using the fact that $\ker{\omega_c}(x) =
\ker{T_x{p_{K_c}}}$. Now
choose a Darboux chart centered at $x$, say $U\times V$, such that, in this
chart,
$p_{K_c}\colon U \times V \rightarrow U$ and $\omega_c(x^1,x^2)=\bar{\omega}_c(x^1)$,
where $\omega_c(x^1,x^2)$ and $\bar{\omega}_c(x^1)$ are independent of $(x_1,
x_2)$.
This shows that
$\bar{\omega}_c$ is well defined on the chart. Using this and the fact that one
can cover the connected submanifold $p^{-1}_{K_c}(z)$
with charts as explained above, one can deduce by a simple argument that
$\bar{\omega}_{c}(z)$
is well defined.

\textbf{(b)} Let $g$ be a Riemannian metric on $M_c$. Then
for each $x \in M_c$ there is a uniquely determined $X(x) \in T_x
M_c$ such that $X(x)$ is orthogonal to $\ker T_x p_{K_c}$ and
$T_{x} p_{K_c} X(x) =\bar{X}(x)$, for all $x \in M_c$. This
defines a vector field $X$ on $M_c$ which is $p_{K_c}$-related to
$\bar{X}$.

\textbf{(c)} Given $\bar{f}$ and using the result of \textbf{(b)} we
see that there is a vector field $X$ on $M_c$ that is
$p_{K_c}$-related to $X_{\bar{f}}$. Then,
for every $x \in M_c$ and every $Y_x \in T_x M_c$,
\begin{equation*}
\begin{split}
\omega_c(x) (X(x), Y_x) &=
\bar{\omega}_c(p_{K_c}(x))\left(X_{\bar{f}}(p_{K_c}(x)), T_x
p_{K_c} Y_x\right) = d \bar{f}(p_{K_c}(x)) (T_x p_{K_c} Y_x) \\
&=d(p_{K_c}^\ast \bar{f})(x)(Y_x).
\end{split}
\end{equation*}

\textbf{(d)} One can proceed as in \textbf{(b)} and \textbf{(c)},
choosing $\bar{f}$ such that
$(d\bar{f}\left(p_{K_c}(x_0))\right)^\sharp = T_{x_0} p_{K_c}
X_{x_0}$ and, besides, the metric $g$ such that $X_{x_0}$ is
perpendicular to $\ker T_{x_0}p_{K_c}$.
\end{proof}
\begin{definition}\label{definitionfirstclassconstr}
\textbf{(a)} For any subspace $A \subseteq \mathcal{F}(N)$ define
the distribution $\Delta_A \subseteq TN$ by $\Delta_A(x) = \{X_f
(x) \mid f\in A \}$.

\textbf{(b)} The space of \textbf{\textit{first class
functions}} is defined as
\begin{equation*}
R^{(c)} = \{f \in \mathcal{F}(N) \mid X_{f}(x) \in T_x M_c,
\mbox{ for all } x \in M_c\}.
\end{equation*}
In other words, $R^{(c)}$ is the largest subset of
$\mathcal{F}(N)$ satisfying
\[
\Delta_{R^{(c)}}(x) \subseteq T_x M_c,
\]
$x \in M_c$.
\end{definition}
\begin{remark}
Dirac was
interested in classical mechanics, where states are represented by points in
phase space, as well as in quantum mechanics where this is not the case.
{}From the point of view of classical mechanics, among the constraint submanifolds and constraints the only ones that seem to play an important role are  $M$, $M_c$ and 
the constraints $\phi^{(0)}_i$, $\phi^{(c)}_i$
defining them
by equations $\phi^{(0)}_i = 0$, $\phi^{(c)}_i = 0$, respectively.
\end{remark}
\begin{lemma}\label{lemmaI6}
\textbf{(a)} $R^{(c)}$ is a Poisson subalgebra of
$(\mathcal{F}(N), \{\, , \,\})$.

\textbf{(b)} $M_c$ is an integral submanifold of
$\Delta_{R^{(c)}}$. Moreover, for any vector field $X$ on $M_c$
that is $p_{K_c}$-related to a vector field $X_{\bar{f}}$ on
$\bar{M_c}$ there exists a function $f \in R^{(c)}$ such that $f|
M_c = p_{K_c}^\ast \bar{f}$ and $X = X_{f}|M_c$. In particular, any
vector field $X$ on $M_c$ satisfying $X(x) \in \ker \omega_c(x)$
for all $x \in M_c$ is $p_{K_c}$-related to the vector field $0$
on the symplectic manifold $\bar{M_c}$, which is associated to the
function $\bar{f} = 0$, therefore there exists a function $f \in
R^{(c)}$, which satisfies $f | M_c = 0$, such that $X(x) =
X_{f}(x)$, $x \in M_c$.
\end{lemma}
\begin{proof}
\textbf{(a)} Let $f, g \in R^{(c)}$. Then $X_{f}(x)$ and
$X_{g}(x)$ are both tangent to $M_c$ at points $x$ of $M_c$ which
implies that $-X_{\{f,g\}}(x) = [X_{f}, X_{g}](x)$ is also tangent
to $M_c$ at points of $x$ of $M_c$. This shows that $\{f,g\}\in
R^{(c)}$. It is easy to see that any linear combination of $f, g$
and also $fg$ belong to $R^{(c)}$.

 \textbf{(b)} By definition $\Delta_{R^{(c)}} \subseteq
TM_c$. We need to show the converse inclusion. Let $X_{x_0} \in
T_{x_0}M_c$, we need to find $f\in R^{(c)}$ such that $X_{f}(x_0)
= X_{x_0}$. Choose the function $\bar{f}$ and the vector field $X$
on $M_c$ as in Lemma \ref{lemma7.1}, \textbf{(d)}. 
Choose any
extension of $p_{K_c}^\ast \bar{f}$ to a function $g$ on $N$. For
each $x\in M_c$, we can apply Lemmas \ref{lemma6.5} and
\ref{lemma6.6} with $E := T_x N$, $V := T_x M_c$, $\beta :=
dg(x)$, $\gamma_i(x) : = d\phi^{(c)}_i(x)$, $i = 1,\dots,r_c$. We obtain
that in a neighborhood $U \subseteq N$ of each point $x_0$ of $M_c$
we can choose $\mathcal{C}^{\infty}$ functions $\lambda^i_{(c)}(x)$,
$i= 1,\dots,r_c$, such that
\[
X(x)= \left(dg(x) + \lambda^i_{(c)}(x)
d\phi^{(c)}_i(x)\right)^\sharp,
\]
for all $x \in M_c \cap U$. Let $f_U(x) = g(x) +
\lambda^i_{(c)}(x) \phi^{(c)}_i(x)$, for all $x\in U$. Then we
have that $X_{f_U}(x) = X(x)$, for all $x \in M_c \cap U$.

Now, consider a partition of unity $\rho_i$, $i \in I$, on $N$,
where each $\rho_i$ is defined on an open set $U_i$, $i \in I$.
Let $J \subseteq I$ be defined by the condition $i \in J$ if and
only if $U_i \cap M_c \neq \emptyset$. Using standard techniques
of partitions of unity and the above result one can assume without
loss of generality that for each $i \in J$ there is a function
$f_{U_i}$ defined on $U_i$ such that $X_{f_{U_i}}(x) = X(x)$, for
all $x \in M_c \cap U_i$. Let $f = \sum_{i \in J} \rho_i f_{U_i}$,
which can be naturally extended by $0$ on $N$. Then it is easy to
see, using the fact that $f_{U_i}(x) = p_{K_c}^\ast \bar{f} (x) =
g(x)$, for each $x \in M_c$, that $X(x)= X_{f}(x)$, for each $x
\in M_c$, and in particular $X_0 = X(x_0)= X_{f}(x_0)$.
\end{proof}
\begin{lemma}\label{lemmaII6}
\textbf{(a)} Each function $f \in R^{(c)}$ is locally constant on
the leaves of $K_c$ therefore, since they are connected, for
each $f \in R^{(c)}$ there is a uniquely determined $\bar{f} \in
\mathcal{F}(\bar{M}_c)$, called $(p_{K_{c}})_{\ast}f$, such that
$f| M_c = p^\ast_{K_c} \bar{f}$. Moreover, the vector fields
$X_f(x)$, $x \in M_c$, on $M_c$ and $X_{\bar{f}}$ on $\bar{M}_c$
are $p_{K_c}$-related.

\textbf{(b)} For each $\bar{f} \in \mathcal{F}(\bar{M}_c)$ there
exists $f \in R^{(c)}$ such that $f| M_c = p^\ast_{K_c}
\bar{f}$ and the vector fields $X_f(x)$, $x \in M_c$, on $M_c$ and
$X_{\bar{f}}$ on $\bar{M}_c$ are $p_{K_c}$-related.
\end{lemma}

\begin{proof}
\textbf{(a)} Let $f \in R^{(c)}$, we only need to show that $f$ is constant on
the leaves
of $K_c$, which is equivalent to showing that $df(x) | \ker
\omega_c (x) = 0$, for all $x \in\ M_c$. For a given $x\in M_c$,
let $v_x \in \ker \omega_c(x)$; then using Lemma \ref{lemmaI6}
\textbf{(b)} one sees that there is a function $g \in R^{(c)}$
such that $v_x= X_g(x)$. Then we have
\[
0 = \omega_c(x)\left(X_f(x), X_g(x)\right) =
\Omega(x)\left(X_f(x), X_g(x)\right) =df(x)X_g(x).
\]
Now we shall prove that $X_f$ and $X_{\bar{f}}$ are
$p_{K_c}$-related. For each $x \in M_c$ and each $Y_x \in T_x
M_c$, we have
\[
\omega_c (X_f(x), Y_x) = \Omega (X_f(x), Y_x) = df(x)(Y_x) = d
(p_{K_c}^\ast \bar{f})(x)(Y_x).
\]
Using this we obtain
\[
\omega_c(x) (X_f(x), Y_x) =
\bar{\omega}_c(x)(p_{K_c}(x))\left(T_x{p_{K_c}} X_f(x), T_x
p_{K_c} Y_x\right) = d \bar{f}(p_{K_c}(x)) (T_x p_{K_c} Y_x),
\]
which shows that $X_{\bar{f}}(p_{K_c}(x)) = T_x{p_{K_c}} X_f(x)$,
because $\bar{\omega}_c$ is symplectic and $T_x p_{K_c} Y_x$
represents an arbitrary element of $T_{p_{K_c}(x)}\bar{M}$.

\textbf{(b)} To find $f$ we choose a vector field $X$ that is $p_{K_c}$-related
to $X_{\bar{f}}$ according to Lemma \ref{lemma7.1} and then use Lemma
\ref{lemmaI6} \textbf{(b)}.
\end{proof}
\begin{definition}\label{definition7.5}
\begin{align*}
I^{(c)}
&= \{f \in R^{(c)} \mid f | M_c = 0\},\\
Z_{I^{(c)}}R^{(c)} &= \{f \in R^{(c)} \mid  \{f, h\} \in
I^{(c)},\, \mbox{for all}\,\, h\in R^{(c)}\}.
\end{align*}
Elements of $I^{(c)}$ are called \textbf{first class constraints}.
\end{definition}
\begin{lemma}\label{IIIlemma6}
\textbf{(a)} $I^{(c)}$ is a Poisson ideal of $R^{(c)}$, that is,
it is an ideal of the ring $R^{(c)}$ such that if $f \in I^{(c)}$,
then $\{f, h\}\in I^{(c)}$, for all $h \in R^{(c)}$.

\textbf{(b)} $Z_{I^{(c)}}R^{(c)}$ is a Poisson subalgebra of
$R^{(c)}$.
\end{lemma}
\begin{proof}
\textbf{(a)} Let $f,g \in I^{(c)}$ and $h \in R^{(c)}$. Then it is
immediate that $f + g$ and $hg$ belong to $I^{(c)}$. For any $h
\in R^{(c)}$, we have $\{f, h\}| M_c =
X_h(f)|M_c = 0$.

\textbf{(b)} Follows from \textbf{(a)}, using basic Poisson algebra
arguments.
\end{proof}
\begin{lemma}\label{lemmaIIII6I}
The following conditions are equivalent for a function $f\in R^{(c)}$. 

\textbf{(i)} $f\in Z_{I^{(c)}}R^{(c)}$.

\textbf{(ii)} $f | M_c$ is locally constant.

\textbf{(iii)} $X_f(x) \in \ker \omega_c(x)$ for $x \in M_c$.
\end{lemma}
\begin{proof}
Assume \textbf{(i)}. Then $\{f, h\} | M_c = 0$ for all $h\in
R^{(c)}$, that is, $df(x) X_h(x)| M_c = 0$. By Lemma
\ref{lemmaI6}, \textbf{(b)}, we know that $X_h(x)$ represents any
vector in $T_x M_c$. We can conclude that $f| M_c$ is locally
constant, so  \textbf{(ii)} holds. Now we will prove that
\textbf{(ii)} implies \textbf{(iii)}. Let $f| M_c$ be locally
constant. Then for all $g \in
R^{(c)}$ and all $x\in M_c$,
\[
0 = X_g (f)(x) = \Omega(x) (X_f, X_g)(x)= \omega_c(x) (X_f(x),
X_g(x)).
\]
Since, again by Lemma \ref{lemmaI6}, $X_g(x)$ represents any
element of $T_x M_c$, we can conclude that $X_f(x) \in \ker
\omega_c(x)$, so \textbf{(iii)} holds true. Now we will prove that
\textbf{(iii)} implies \textbf{(i)}. Assume that $X_{f}(x) \in
\ker \omega_c (x)$, $x\in M_c$. Then for all $g \in R^{(c)}$ and all $x \in M_c$,
\[
\{g, f\}(x) = \Omega(x) (X_g, X_f)(x) = \omega_c(x) (X_g, X_f)(x)
= 0,
\]
that is, $\{g, f\} \in I^{(c)}$. Using this and the definitions,
we see that $f \in Z_{I^{(c)}}R^{(c)}$.
\end{proof}
\begin{lemma}\label{lemalema}
The map $(p_{K_{c}})_\ast \colon  R^{(c)}\rightarrow
\mathcal{F}(\bar{M}_c)$ defined in Lemma \ref{lemmaII6} is a
surjective Poisson map and its kernel is $I^{(c)}$, therefore
there is a natural isomorphism of Poisson algebras
$(p_{K_{c}})_{\ast I^{(c)}} \colon  R^{(c)}/I^{(c)} \rightarrow
\mathcal{F}(\bar{M}_c)$.
\end{lemma}
\begin{proof}
Surjectivity of $(p_{K_c})_*$ and the fact that its kernel is $I^{(c)}$ follows
immediately from Lemma \ref{lemmaII6}
and the definitions. This implies that $(p_{K_{c}})_{\ast I^{(c)}}$ is an
algebra isomorphism.
Also, using the definitions, for $f, g \in R^{(c)}$ and any $x \in M_c$ we can
prove easily that
\begin{align*}
\{f,g\}(x) &= \Omega(x)(X_f(x), X_g(x)) =
\omega_c(x)(X_f(x), X_g(x))\\
&=
\bar{\omega}_c\left(p_{K_c}(x)\right)\left(X_{\bar{f}}\left(p_{K_c}(x)\right),
X_{\bar{g}}\left(p_{K_c}(x)\right)\right) = \{\bar{f},
\bar{g}\}(p_{K_c}(x)),
\end{align*}
where $\bar f=(p_{K_c})_*f$,  $\bar g=(p_{K_c})_*g$.
Denote by
$\pi_{I^{(c)}} \colon  R^{(c)} \rightarrow R^{(c)}/I^{(c)}$ the natural homomorphism
of Poisson algebras. Then from the previous equalities we obtain
$(p_{K_c})_{\ast I^{(c)}} \left\{\pi_{I^{(c)}}(f), \pi_{I^{(c)}}(g)\right\} =
\{\bar{f}, \bar{g}\}$, which shows that $(p_{K_c})_{\ast I^{(c)}}$ is a Poisson
isomorphism. In other words, we have the commutative diagram
\begin{center}\leavevmode
\xymatrix{
R^{(c)} \ar[r]^{\pi_{I^{(c)}}} \ar [dr]_{(p_{K_{c}})_\ast} & R^{(c)}/I^{(c)}
\ar[d]^{(p_{K_{c}})_{\ast I^{(c)}}}\\
&         \mathcal{F}(\bar{M}_c)}
\end{center}
All the arrows are defined in a natural way and they are
surjective Poisson algebra homomorphisms.
\end{proof}
\paragraph{Physical variables.}

It is immediate to see from the definitions that for all $x\in M_c$,
\begin{equation}\label{equationkernelss}
\ker \omega (x)\cap T_x M_c \subseteq \ker \omega_c (x).
\end{equation}

{}From now on we will assume the following.

\begin{assumption}\label{K2}
 \textbf{(a)} $\ker \omega (x)$ is a
regular distribution, that is, it determines a regular foliation
$K$ and the natural projection $p_K \colon  M \rightarrow \bar{M}$,
where $\bar{M}=M/K$, is a submersion.

\textbf{(b)} $\ker \omega (x)\cap T_x M_c$ is a
distribution of constant rank.  
\end{assumption}
\begin{theorem}\label{lemmmma}
The distribution $\ker \omega (x)\cap T_x M_c$ is
regular and has rank
$d^{(c)}(x)$. Its integral manifolds are $S \cap M_c$, where $S$
is an integral manifold of $\ker \omega$. Moreover, these integral
manifolds give a foliation $\widetilde{K}_c$ of $M_c$ which is
regular, that is, the natural map $p_{\widetilde{K}_c} \colon  M_c
\rightarrow \widetilde{M}_c$, where $\widetilde{M}_c = M_c / \widetilde{K}_c$
is a submersion. Besides, each leaf of the foliation
$K_c$ is foliated by leaves of $\widetilde{K}_c$, which gives a
naturally defined submersion $p_{K_c\widetilde{K}_c} \colon  \widetilde{M}_c
\rightarrow \bar{M}_c$. In other words, we obtain the commutative
diagram
\begin{center}\leavevmode
\xymatrix{
M_c \ar[r]^{p_{\widetilde{K}_c}} \ar [dr]_{p_{K_c}} & \widetilde{M_c}
\ar[d]^{p_{K_c\widetilde{K}_c}}\\
&          \bar{M}_c}
\end{center}
where each arrow is a naturally defined submersion.
\end{theorem}
\begin{proof}
The first assertion, about the rank of the distribution $\ker \omega (x)\cap T_x M_c$, is easy to prove.
Let $x_0\in M_c$. Then there exists a uniquely
determined integral manifold $S$ of the distribution $\ker \omega$
such that $x_0\in S$. Using that, by assumption, $\ker \omega
(x)\cap T_x M_c$ is a distribution of constant dimension and that
$\dim \left (\ker \omega (x)\cap T_x M_c \right) = \dim \left(T_x
S \cap T_x M_c\right)$ we can conclude that the intersection
$S \cap M_c$ coincides with the integral leaf of the integrable distribution 
of $\ker \omega
\cap TM_c$ containing $x_0$. So we obtain the foliation
$\widetilde{K}_c$ of $M_c$. Using (\ref{equationkernelss}) we can
deduce that each leaf of the foliation $K_c$ is foliated by leaves
of $\widetilde{K}_c$. The rest of the proof follows by
standard arguments.
\end{proof}
\begin{lemma}\label{lemma_big_diagram}
\textbf{(a)} The following diagram is commutative
\begin{center}\leavevmode
\xymatrix{ M \ar[r]^{p_{K}}& \bar{M}&
T\bar{M}\ar[l]_{\tau_{\bar{M}}}
\\
M_c \ar[r]^{p_{\widetilde{K}_c}} \ar[dr]_{p_{K_c}} \ar[u]^{f_c}&
\widetilde{M_c} \ar[d]^{p_{K_c\widetilde{K}_c}} \ar[u]_{\widetilde{f}_c}&
T\bar{M}|\widetilde{M}_c \ar[l]^{\tau_{\bar{M}}} \ar[r]^{\epsilon_c}
\ar[u]^{\widetilde{F}_c \circ \epsilon_c}& \widetilde{f}^\ast_c T\bar {M}
\ar[ul]_{\widetilde{F}_c}
\\
&\bar{M}_c }
\end{center}
where the arrows are defined as follows. The maps $p_K$,
$p_{K_c}$, $p_{\widetilde{K}_c}$ and $p_{K_c\widetilde{K}_c}$ are defined
in Assumption \ref{K1}, Assumption \ref{K2} and Theorem \ref{lemmmma}.
By definition, the map $f_c$ is the inclusion. The map
$\widetilde{f}_c$ is an embedding defined by $\widetilde{f}_c (S \cap M_c)
= S$, where $S$ is a leaf of the foliation $K$. We will think of
$\widetilde{f}_c$ as being an inclusion. The vector bundle $T\bar{M}|
\widetilde{M}_c$ is the tangent bundle $T\bar{M}$ restricted to
$\widetilde{M}_c$. In other words, since $\widetilde{f}_c$ is an
inclusion, $T\bar{M}| \widetilde{M}_c$ is identified via some
isomorphism, called $\epsilon_c$, with the pullback of $T\bar{M}$
by $\widetilde{f}_c$. We call $\widetilde{F}_c$ the natural map associated
to the pullback.

\textbf{(b)} The presymplectic form $\omega$ on $M$ passes to the
quotient via $p_K$ giving a uniquely defined symplectic form
$\bar{\omega}$ on $\bar{M}$, satisfying $p_K^\ast \bar{\omega} =
\omega$.
The presymplectic form $\omega_c$, which, by definition is $f_c^\ast \omega$,
defines uniquely a presymplectic form
$\widetilde{\omega}_c$ on $\widetilde{M}_c$ via $p_{\widetilde{K}_c}$ satisfying
$p_{\widetilde{K}_c}^{\ast} \widetilde{\omega}_c = \omega_c$, $\widetilde\omega_c = \widetilde{f}_c^\ast
\bar{\omega}$.
The energy
 $\mathcal{E}$ on $M$ satisfies $d \mathcal{E}(x) | \ker
\omega (x) = 0$, for all $x \in M_c$, therefore it defines
uniquely a 1-form on $T\bar{M}| \widetilde{M}_c$, called
$(\widetilde{F}_c \circ \epsilon)^{\ast} d \mathcal{E} \in \Gamma \left((T\bar{M}|
\widetilde{M}_c)^\ast\right)$. Since $\mathcal{E}$ is constant on each leaf of $\widetilde{K}_c$, it also defines a function
$\widetilde{\mathcal{E}}_c$ on $\widetilde{M}_c$. Since $T\widetilde{M}_c
\subseteq T\bar{M}_c$ via the inclusion $T \widetilde{f}_c$ we have
$(\widetilde{F}_c \circ \epsilon)^{\ast} d \mathcal{E} | T_{\widetilde{x}} \widetilde{M}_c =
d
\widetilde{\mathcal{E}}_c(\widetilde{x})$, for all $\widetilde{x} \in
\widetilde{M}_c$.

\textbf{(c)} Equation of motion (\ref{diracagainagain}) on $M_c$
passes to the quotient $\widetilde{M}_c$ as
\begin{equation}\label{tildeequation}
\bar{\omega}(\widetilde{x})(\widetilde{X}(\widetilde{x}),\,\,) =
(\widetilde{F}_c \circ \epsilon)^{\ast} d \mathcal{E}(\widetilde{x}),
\end{equation}
where $\widetilde{X}(\widetilde{x}) \in T_{\widetilde{x}} \widetilde{M}_c$. This
means that if $X(x) \in T_x M_c$ is a solution of
(\ref{diracagainagain}) then $\widetilde{X}(\widetilde{x}) := T_x
p_{\widetilde{K}_c} X(x)$, where $\widetilde{x} = p_{\widetilde{K}_c}(x)$, is
a solution of (\ref{tildeequation}). Therefore, a solution curve
$x(t)$ of (\ref{diracagainagain}) projects to a solution curve
$\widetilde{x}(t) = p_{\widetilde{K}_c} \left(x(t)\right)$ of
(\ref{tildeequation}) on $\widetilde{M}_c$. Equation
(\ref{tildeequation}) has unique solution $\widetilde{X}(\widetilde{x})$
for each $\widetilde{x} \in \widetilde{M}_c$. This solution also satisfies
the equation
\begin{equation}\label{tildeequation2}
\widetilde{\omega}_c(\widetilde{x})(\widetilde{X}(\widetilde{x}),\,\,) =
d\widetilde{\mathcal{E}}(\widetilde{x}).
\end{equation}
However solutions to equation (\ref{tildeequation2}) are not
necessarily unique, since $\ker \widetilde{\omega}_c(\widetilde{x})$ is not
necessarily $0$.

\textbf{(d)} The restriction of the energy $\mathcal{E} |
M_c$ satisfies
\[
d (\mathcal{E} | M_c) (x) | \ker \omega_c(x) = 0,
\]
for all $x \in M_c$, therefore there is a uniquely defined
function $\bar{\mathcal{E}}_c$ on $\bar{M_c}$ such that
$p_{K_c}^\ast \bar{\mathcal{E}}_c = \mathcal{E} | M_c$. The
equation
\begin{equation}\label{barequation}
\bar{\omega}_c(\bar{x})(\bar{X}(\bar{x}), \, \,) = d
\bar{\mathcal{E}}_c(\bar{x})
\end{equation}
has unique solution $\bar{X}(\bar{x})$ for $\bar{x} \in
\bar{M_c}$. If $\widetilde{X}(\widetilde{x})$ is a solution of
(\ref{tildeequation}) then $\bar{X}(\bar{x}) =
T_{\widetilde{x}}p_{K_c\widetilde{K}_c} \widetilde X({\widetilde{x}})$ is a solution of
(\ref{barequation}). Therefore, a solution curve $\widetilde{x}(t)$ of
(\ref{tildeequation}) projects to a solution curve $\bar{x}(t) =
p_{K_c\widetilde{K}_c} \left(\widetilde{x}(t)\right)$ of
(\ref{barequation}) on $\bar{M}_c$.
\end{lemma}
\begin{proof}
\textbf{(a)} The equality $p_{K_c} = p_{K_c\widetilde{K}_c} \circ
p_{\widetilde{K}_c}$ was proven in Theorem \ref{lemmmma}. The equality
$\tau_{\bar{M}} \circ \widetilde{F}_c \circ \epsilon_c = \widetilde{f}_c
\circ \tau_{\bar{M}}$ results immediately from the definitions.
The equality $p_{K} \circ f_c = \widetilde{f}_c \circ p_{\widetilde{K}_c}$
results by applying the definitions and showing that, for given $x
\in M_c$, $p_{K} \circ f_c (x)= S_x = \widetilde{f}_c \circ
p_{\widetilde{K_c}}(x)$, where $S_x$ is the only leaf of $K$
containing $x$.

\textbf{(b)} Existence and uniqueness of $\bar{\omega}$ and
$\widetilde{\omega}$ is a direct consequence of the definitions and standard
arguments on
passing to quotients. For
any $x\in M_c$ we know that there exists a solution $X$ of
equation (\ref{diracagainagain}), from which it follows
immediately that $d\mathcal{E}(x)| \ker \omega(x) = 0$.
The rest of the proof consists of standard arguments on
passing to quotients.

\textbf{(c)} We shall omit the proof of this item which is a
direct consequence of the definitions and standard arguments on
passing to quotients.

\textbf{(d)} If $X(x)$ is a solution of $\omega(x)(X(x),\,\,) =
d\mathcal{E}(x)$ then it is clear that it also satisfies
$\omega_c(x)(X(x),\,\,) = d(\mathcal{E}| M_c)(x)$. It follows that
\[
d (\mathcal{E} | M_c) (x) | \ker \omega_c(x) = 0,
\]
for all $x \in M_c$. The rest of the proof is a consequence of
standard arguments on passing to quotients.
\end{proof}
\begin{remark} Recall that the locally constant 
function $d^{(c)}(x)$ on $M_c$ is the dimension of the
distribution $\ker \omega(x) \cap T_x M_c$ on $M_c$ and also the
dimension of the fiber of the bundle $S^{(c)}$. If $d^{(c)}(x)$ is nonzero then there is no
uniqueness of solution to equation of motion (\ref{diracagainagain}),
since solution curves to that equation satisfy, by definition,
\begin{align*}%\label{diracagainagainagain}
\omega(x)(\dot{x},\,\,) 
&= 
d\mathcal{E}(x) | T_x M\\
(x, \dot{x})
& \in T_x M_c.
\end{align*}

Passing to the quotient manifold
$\widetilde{M}_c$ eliminates this indeterminacy and uniqueness of
solution is recovered. This is related to the notions of \emph{physical
state} and \emph{determinism}, mentioned by \cite[p.\ 20]{Dirac1964}.

The reduced equations of motion on the symplectic manifold $\bar M_c$ are Hamilton's equations. The reduced equations of motion on the manifold $\widetilde M_c$ are not necessarily Hamilton's equations but they are ODEs.
The natural fibration $p_{K_c\widetilde K_c}\colon \widetilde M_c \to \bar M_c$ has the following property. For $\bar m_0\in \bar M_c$ and $\widetilde m_0\in (p_{K_c\widetilde K_c})^{-1}(\bar m_0)\subseteq \widetilde M_c$, if $\bar m (t)$ and $\widetilde m (t)$ are the unique solution curves through $\bar m_0$ and $\widetilde m_0$ respectively, then $p_{K_c\widetilde K_c}\widetilde m (t)=\bar m (t)$. This means that there is a one-to-one correspondence between the solution curves on  $\widetilde M_c$ and $\bar M_c$, modulo the choice of a point on the fiber over the physical state $\bar m_0$. See also remark \ref{physicalstate}.
\end{remark}

\subsection{Second class subbundles}\label{subsectionV}
The notions of first class and second class constraints and functions and also
second class submanifolds
only depend on the final constraint
submanifold $M_c$ and the ambient symplectic manifold $N$ and do not depend on
the primary constraint $M_0 = M$ or the energy $\mathcal{E} \colon  N \rightarrow
\mathbb{R}$. Accordingly, in this section we will adopt an abstract setting,
where only an ambient symplectic manifold and a submanifold are given.
Among several interesting references we cite \cite{MR0358860}  which has some
points of contact with our work.

The following definition is inspired by the ones given in section \ref{poisson-dirac-gotay-nester}.
\begin{definition}\label{firstfirstclass}
Let $(P,\Omega)$ be a symplectic manifold and $S \subseteq P$ a
given submanifold. Then, by definition,
\begin{align*}
R^{(S,P)} &:= \{f \in \mathcal{F}(P) \mid X_{f}(x) \in T_x S, \,
\mbox{for all}\,\,
x \in S\}\\
I^{(S,P)} &:= \{f \in R^{(S,P)} \mid  f | S = 0\}
\end{align*}
Elements of $R^{(S,P)}$ are called \textbf{\textit{first class
functions.}}
Elements of $I^{(S,P)}$ are called \textbf{\textit{first class
constraints.}}
The submanifold $S$ is called a \textbf{\textit{first class
constraint submanifold}} if for all $f \in \mathcal{F}(P)$ the
condition $f | S = 0$ implies $f \in I^{(S,P)}$, that is,
$I_{(S,P)} \subseteq I^{(S,P)}$, where $I_{(S,P)}$ is the ideal of
the ring $\mathcal{F}(P)$ of all functions vanishing on $S$.
\end{definition}
Obviously, using the notation introduced before in the present
section, $R^{(M_c,N)} = R^{(c)}$ and $I^{(M_c,N)} = I^{(c)}$. All
the properties proven for $R^{(c)}$ and $I^{(c)}$ hold in general
for $R^{(S,P)}$ and $I^{(S,P)}$. For instance, $T_x S$ is
the set of all $X_f (x),\, f \in R^{(S,P)}$. Every function $f \in
R^{(S,P)}$ satisfies $df(x)(X_x) = 0$, for all $X_x \in \ker
(\Omega(x) | T_x S)$ and $\ker (\Omega(x) | T_x S)$ is the set of all
 $X_f (x),\, f \in I^{(S,P)}$.
We shall collect in a single lemma \ref{mainlemma1section7} the basic properties of first and second class constraints and Dirac brackets and their proofs in a compact way and with a uniform language that does not rely on coordinates. 
Most of these results are essentially known since the fundamental works of Dirac, and have been stated and proven in different ways in the literature, except perhaps for the notion of second class subbundle and the corresponding description of second class submanifolds as tangent to the subbundle. The second class subbundles  $V^\Omega \subseteq TP|S$ classify all second class submanifolds 
$S^V$
containing $S$ at a linear level, that is, $V^\Omega$ is tangent to the second class submanifold.
For such a second class submanifold, which is a symplectic submanifold, the Dirac bracket of two functions $F$ and $G$ at points $x\in S$ can be calculated as the canonical bracket of the restrictions of $F$ and $G$. 
This has a global character. 
A careful study of the global existence of a bracket defined on sufficiently small open sets 
$U \subseteq P$ containing $S$ which coincides with the previous one on the second class submanifold will not be considered in this paper. 
However, to write global equations of motion on the final constraint submanifold one only needs to know the vector bundle $V^\Omega$, which carries a natural fiberwise symplectic form.

All these are fundamental properties of second class constraints and constraint submanifolds,
and Theorem \ref{theoremsecondclass} collects the essential part of them; we suggest to take a look at it
before reading Lemma \ref{mainlemma1section7}.
\begin{lemma}\label{mainlemma1section7}
Let $(P,\Omega)$ be a symplectic manifold of dimension $2n$ and $S
\subseteq P$ a given submanifold of codimension $r$. Let $\omega$
be the pullback of $\Omega$ to $S$ and assume that $\ker
\omega(x)$ has constant dimension. Assume that $S$ is defined
regularly by equations $\phi _1 = 0,\dots,\phi _r = 0$
on a neighborhood $U \supseteq S$
and assume that we can choose a
subset $\{\chi_1,\dots,\chi_{2s}\} \subseteq \{\phi_1,\dots,\phi_r\}$
such that $\det (\{\chi_i, \chi_j\}(x)) \neq 0$ for all $x \in U$,
where we assume that $2s= \operatorname{rank} (\{\phi_i,
\phi_j\}(x))$, for all $x \in U$.
We shall often denote
$c^{\chi}_{ij}(x) = \{\chi_i, \chi_j\}(x)$ and $c_{\chi}^{ij}(x)$ the inverse of
$c^{\chi}_{ij}(x)$. 
Moreover, we will assume that the following stronger condition holds, for simplicity. Equations
 $\phi _1 = B_1,\dots,\phi _r= B_r$ and  $\chi _1 = C_1,\dots,\chi _{2s} = C_{2s}$ define submanifolds of $U$  regularly, for small enough $B_1,\dots,B_r$ and $C_1,\dots,C_{2s}$.
Then

\textbf{(a)} $2s = r - \dim \ker \omega = 2n - \dim S - \dim \ker
\omega$. There are $\psi_k \in I^{(S,P)}$, $k = 1,\dots,r - 2s$, which
in particular implies $\{\psi_k, \psi_l\}(x) = 0$, $\{\psi_k, \chi_i\}(x) = 0$,
for $k, l = 1,\dots,r - 2s$, $i = 1,\dots,2s$, such that
$d\psi_1(x),\dots,d\psi_{r - 2s}(x), d\chi_1(x),\dots,d\chi_{2s}(x)$,
are linearly independent for all $x \in S$. Moreover,
$X_{\psi_1}(x),\dots,X_{\psi_{r - 2s}}(x)$ form a basis of $\ker
\omega (x)$, for all $x \in S$ and $d\psi_1(x),\dots,d\psi_{r - 2s}(x),
d\chi_1(x),\dots,d\chi_{2s}(x)$
form a basis of $(T_x S)^\circ$.

\textbf{(b)} The vector subbundle $V^{\chi} \subseteq TP|S$ with
base $S$ and fiber
\[
V^{\chi}_x = \operatorname{span}
\left(X_{\chi_1}(x),\dots,X_{\chi_{2s}}(x)\right) \subseteq T_x P,
\]
satisfies
\begin{align}\label{fla715}
V^{\chi}_x \cap T_x S &= \{0\}
\\ \label{fla716}
V^{\chi}_x \oplus \ker \omega (x) &= (T_x S)^\Omega
\\
\label{fla717} (V^{\chi}_x)^\Omega \cap (\ker \omega (x))^\Omega
&= T_x S,
\end{align}
$x \in S$.

\textbf{(c)} There is a neighborhood $U$ of $S$ such that the
equations $\chi_1 = 0,\dots,\chi_{2s} = 0$, on $U$ define a
symplectic submanifold $S^{\chi}$ such that $S \subseteq S^{\chi}$
and
\begin{align*}
T_x S^{\chi} &=
\left(V^{\chi}_x\right)^\Omega\\
T_x S^{\chi} \oplus V^{\chi}_x &= T_x P,
\end{align*}
for $x \in S^{\chi}$, where we have extended the definition of
$V^{\chi}_x$ for $x \in S^{\chi}$ using the expression
\[
V^{\chi}_x = \operatorname{span}
\left(X_{\chi_1}(x),\dots,X_{\chi_{2s}}(x)\right) \subseteq T_x P,
\]
for $x \in S^{\chi}$. The submanifold $S^\chi$ has the property
$I_{(S, S^\chi)} \subseteq I^{(S, S^\chi)}$,
that is, $S$ is a
first class constraint submanifold of $S^\chi$,
defined regularly by
$\psi_i | S^{\chi} = 0$, $i = 1,\dots,r - 2s$,
and
$\psi_i | S^{\chi} \in I^{(S, S^\chi)}$, $i = 1,\dots,r - 2s$.
Moreover, it has
the only possible dimension, which is $\operatorname{dim} S^\chi
=\operatorname{dim}
S + \operatorname{dim}\ker \omega = 2n - 2s$, for symplectic submanifolds having
that
property.  It is also a minimal object in the set of all symplectic submanifolds $P_1 \subseteq P$, ordered by inclusion, satisfying
$S \subseteq P_1$.

\textbf{(d)} Let $V$ be any vector subbundle of $TP|S$ such that
\begin{equation}\label{720eq}
V \oplus \ker \omega= (TS)^\Omega,
\end{equation}
or equivalently,
\begin{equation}\label{721eq}
V^\flat \oplus (\ker \omega)^\flat= (TS)^\circ.
\end{equation}
Then $\operatorname{dim}V_x = 2s$, for $x \in S$. Let $S^V$ be a
submanifold of $P$ such that $T_x S^V = V_x^\Omega$, for each
$x\in S$. Then $S$ is a submanifold of $S^V$. Such a submanifold
$S^V$ always exists. Moreover, for such a submanifold there is an
open set $U \subseteq P$ containing $S$ such that $S^V \cap
U$ is a symplectic submanifold of $P$.

Let $\bar{x} \in S$ and let $\chi^\prime _1 = 0,\dots,\chi^\prime
_{2s} = 0$ be equations defining $S^V \cap U^\prime$ for some open
neighborhood $U^\prime \subseteq P$  and satisfying that
$d\chi^\prime _1(x),\dots,d\chi^\prime _{2s}(x)$ are linearly
independent for $x \in S^V \cap U^\prime$. Then, $d\chi^\prime
_1(x),\dots,d\chi^\prime _{2s}(x), d\psi _1(x),\dots,d\psi _{r-2s}(x)$
are linearly independent and
$\operatorname{det}(\{\chi^\prime _i, \chi^\prime
_j\}(x)) \neq 0$, for $x\in S^V \cap U^\prime$. All the
properties established in \textbf{(a)}, \textbf{(b)}, \textbf{(c)}
for $\chi_1,\dots,\chi_{2s}$ on $S$ hold in an entirely similar way
for $\chi^\prime _1,\dots,\chi^\prime _{2s}$, on $S \cap U^\prime$.
In particular, $S^V \cap U^\prime = S^{\chi^\prime}$.

\textbf{(e)} Let $\omega^\chi$ be the pullback of $\Omega$ to $S^\chi$
and $\{ , \}_\chi$ the corresponding bracket.
For given $F, G \in \mathcal{F}(P)$ define
$F_\chi : = F - \chi_i c_{\chi}^{ij} \{\chi_j, F\}$ and also
\begin{equation*}%\label{diracbracket}
\{F,G\}_{(\chi)} : = \{F,G\} -
\{F,\chi_i\}c_{\chi}^{ij}\{\chi_j,G\},
\end{equation*}
which is the famous bracket introduced by Dirac, called \textbf{Dirac bracket},
and it is defined  for $x$ in the neighborhood
$U$ where $c^{\chi}_{ij}(x)$ has an inverse $c_{\chi}^{ij}(x)$.
Then, for any $x \in S^\chi$,
\begin{equation*}%\label{constrconst}
\{F_\chi,\chi_i\}(x) = 0
\end{equation*}
for $i = 1,\dots,2s$, and also
\begin{equation*}%\label{constrequalities}
\{F_\chi,G_\chi\}(x) = \{F,G\}_{(\chi)}(x) =
\omega^\chi(x)\left(X_{F|S^\chi}(x), X_{G|S^\chi}(x)\right) =
\{F|S^\chi ,G|S^\chi\}_{\chi}(x).
\end{equation*}
If we denote $X_{(\chi), F}$ the
Hamiltonian vector field associated to the function $F \in
\mathcal{F}(P)$, with respect to the Dirac bracket
$\{ , \}_{(\chi)}$
then the previous equalities are equivalent to
\[
X_{(\chi), F}(x) = X_{F_{\chi}} (x) = X_{F | S^\chi}(x).
\]
The Jacobi identity is satisfied for the Dirac bracket $\{F,G\}_{(\chi)}$ on
$S^\chi$, that is,
\begin{equation*}
\{\{F,G\}_{(\chi)}, H\}_{(\chi)}(x) + \{\{H,F\}_{(\chi)},
G\}_{(\chi)}(x) + \{\{G,H\}_{(\chi)}, F\}_{(\chi)}(x) = 0,
\end{equation*}
for $x \in
S^\chi$.

\textbf{(f)} Let $U$ be an open neighborhood of $S$ such that
$c^{\chi}_{ij}(x)$ is invertible for $x \in U$. For each $C=
(C_1,\dots,C_{2s}) \in \mathbb{R}^{2s}$ let $\chi_i^C = \chi_i -
C_i$ and define $S^{\chi^C}$ by the equations
$\chi^C_i(x) = 0$, $i =
1,\dots,2s$, $x \in U$.
For any $C$ in a  sufficiently small neighborhood of $0$,
$S^{\chi^C}$ is a nonempty symplectic 
submanifold of $P$. Define the matrix
$c^{\chi^C}_{ij}(x) = \{\chi^C_i,\chi^C_j\}(x)$,
and also
$c_{\chi^C}^{ij}(x)$ as being its inverse,
$x \in U$.
Then, the equalities
\begin{equation}\label{equationsss}
c^{\chi}_{ij}(x)= \{\chi_i,\chi_j\}(x) = \{\chi^C_i,\chi^C_j\}(x) =
c^{\chi^C}_{ij}(x),
\end{equation}
and also,
\begin{equation}\label{equationssss}
\{F,G\}_{(\chi)}(x) = \{F,G\}_{(\chi^C)}(x)
\end{equation}
are satisfied for all $x \in U$. All the definitions and
properties proved in \textbf{(e)} for the case $C = 0$ hold in
general for any $C$ in a neighborhood of $0$ small enough to ensure that $S^{\chi^C}$
is nonempty.
In
particular, the equalities
\begin{multline}\label{constrequalitiesC}
\{F_{\chi^C},G_{\chi^C}\}(x) = \{F,G\}_{(\chi^C)}(x) =
\omega^{\chi^C}(x)\left(X_{F|S^{\chi^C}}(x), X_{G|S^{\chi^C}}(x)\right) \\=
\{F|S^{\chi^C} ,G|S^{\chi^C}\}_{\chi^C}(x)
\end{multline}
and
\begin{equation}\label{equationsssss}
X_{(\chi^C), F}(x) = X_{F_{\chi^C}}(x) =
X_{F|S^{\chi^C}}(x)
\end{equation}
hold for $x \in S^{\chi^C}$, and any $C$ in such a neighborhood.
The Dirac bracket $\{F,G\}_{(\chi)}$ satisfies
the Jacobi identity for $F,G \in \mathcal{F}(U)$ and the
symplectic submanifolds $S^{\chi^C}$ are the symplectic leaves of
the Poisson manifold $(\mathcal{F}(U), \{\,,\,\}_{(\chi)})$.
By shrinking, if necessary, the open set $U$  and for $C$ in a sufficiently
small neighborhood of
$0 \in \mathbb{R}^{2s}$, the equations
$\psi_k | S^{\chi^C} = 0$, $k = 1,\dots,r - 2s$, define regularly a first class
constraint submanifold
$S^C \subseteq S^{\chi^C} \subseteq U$, and the functions
$\psi_k | S^{\chi^C} \in R^{(S^C, S^{\chi^C})} \subseteq
\mathcal{F}(S^{\chi^C})$ are first class constraints, that is,
$\psi_k | S^{\chi^C} \in I^{(S^C, S^{\chi^C})}$ $k = 1,\dots,r - 2s$.
We have that
$\operatorname{dim} S^{\chi^C} =\operatorname{dim}
S^C + \operatorname{dim}\ker \omega^C$,
where
$\omega^C$
is the pullback of $\Omega$
to
$S^C$.
One has
$\operatorname{dim} S^{C} = \operatorname{dim} S$
and
$\operatorname{dim} S^{\chi^C} = \operatorname{dim} S^{\chi}$,
therefore
$\operatorname{dim}\ker \omega = \operatorname{dim}\ker \omega^C$.
\end{lemma}
\begin{proof}
\textbf{(a)} Let $x\in S$. We are going to use lemmas and
corollaries \ref{lemma6.1}--\ref{corollary6.7} with $E := T_x P$;
$V := T_x S$; $\gamma_i := d\phi_i(x)$, $i = 1,\dots,r$; $\Omega : =
\Omega(x)$; $\omega : = \omega(x)$; $\beta = 0$.

Elements $\psi =
\lambda^i \phi_i$, $\lambda^i \in \mathcal{F}(P)$ such that $\psi
\in I^{(S,P)}$, which implies $X_{\psi}(x) \in T_x S$ for $x \in S$, must satisfy $\{\psi, \phi_j\}(x)
= 0$, or, equivalently, $\lambda^i(x)
d\phi_i(x)\left(X_{\phi_j}(x)\right) = 0$, for $j = 1,\dots,r$, $x \in
S$. Using Lemma \ref{lemma6.6} we see that $X_{\psi}(x) \in \ker
\omega(x)$. Since one can choose $r - 2s$ linearly independent
solutions, say $\lambda_i = (\lambda_i^1,\dots,\lambda_i^{r})$, $i =
1,\dots,r-2s$, we obtain elements $\psi_i \in I^{(S,P)}$, namely, $\psi_i =
\lambda_i^j \phi_j$, such that $(d\psi_1(x),\dots,d\psi_{r-2s}(x))$
are linearly independent, or, equivalently, taking into account Corollary \ref{corollary6.7},
that
$(X_{\psi_1}(x),\dots,X_{\psi_{r-2s}}(x))$ is a basis of
$\ker\omega(x)$ for $x \in S$. If
$d\psi_1(x)$, \dots, $d\psi_{r-2s}(x)$, $d\chi_1(x)$, \dots, $d\chi_{2s}(x)$, were
not linearly independent, then there would be
a linear combination, say $\bar{\chi} = a^i \chi_i$, with at least one nonzero coefficient,
and some $x \in S$, such that $d\bar{\chi} (x) =
\mu^k d\psi_k (x)$ for some $\mu^k$, $k = 1,\dots,r-2s$. But then,
for any $j = 1,\dots,2s$, $\{\bar{\chi}, \chi_j\}(x) = d\bar{\chi}(x)
X_{\chi_j}(x) = \mu^k d\psi_k(x) X_{\chi_j}(x) = -\mu^k d\chi_j(x)
X_{\psi_k}(x) = 0$, which contradicts the fact that
$\operatorname{det}\left(\{\chi_i, \chi_j\}(x)\right) \neq 0$.
Using this and the fact that $\psi_i = 0$, $\chi_j = 0$, $i = r-2s$, $j = 1,\dots,2s$
define $S$ regularly, we can conclude that
$d\psi_1(x)$, \dots, $d\psi_{r-2s}(x)$, $d\chi_1(x)$, \dots, $d\chi_{2s}(x)$ is a basis of
$(T_x S)^\circ$.

\textbf{(b)} If $\lambda^i X_{\chi_i}(x) \in T_xS$ then
$d\chi_j(x)\lambda^i X_{\chi_i}(x) = 0$, $j = 1,\dots,2s$, which
implies $\lambda^i\{\chi_j, \chi_i\} = 0$, $j = 1,\dots,2s$, then
$\lambda^i = 0$, $i = 1,\dots,2s$,
which proves (\ref{fla715}).
To prove (\ref{fla716}) we apply
the operator $^\flat$ to both sides and obtain the equivalent
equality $\operatorname{span} (d\chi_1(x),\dots,d\chi_{2s}(x))\oplus
\operatorname{span} (d\psi_1(x),\dots,d\psi_{r-2s}(x)) = (T_x
S)^\circ$, which we know is true, as proven in \textbf{(a)}. To
prove (\ref{fla717}) we apply the orthogonal operator $^\Omega$
to both sides of (\ref{fla716}).

\textbf{(c)} Since $d\chi_1(x),\dots,d\chi_{2s}(x)$ are linearly
independent for $x \in S$ they are also linearly independent for
$x$ in a certain neighborhood $U$ of $S$. Then $\chi_1(x) =
0,\dots,\chi_{2s}(x)= 0$ define a submanifold $S^\chi$ of $U$
containing $S$. To see that it is a symplectic submanifold choose
$x \in S^\chi$ and apply Corollary \ref{corollary6.7} with $E : =
T_x P;$ $V : = T_x S^\chi;$ $\beta := 0;$ $\gamma_i := d\chi_i$,
$i = 1,\dots,2s$; $\omega := \Omega(x)| T_x S^{\chi}$. We can
conclude that $\operatorname{dim}(\ker \Omega(x)| T_x S^{\chi}) =
0$. Now let us prove that $T_x S^\chi =
\left(V_x^\chi\right)^\Omega$, namely, $T_x S^\chi =
\operatorname{span}\left(d \chi_1(x),\dots,d
\chi_{2s}(x)\right)^\circ =
(\left(V_x^\chi\right)^\flat)^\circ =
\left(V_x^\chi\right)^\Omega$. From this, using that $S^{\chi}$ is
symplectic one obtains $T_x S^{\chi} \oplus V_x^\chi = T_x P$.
To prove that
$S \subseteq S^\chi$
is a first class constraint submanifold defined by first class constraints
$\psi_i | S^\chi$, $i = 1,\dots,r-2s$, on $S^\chi$, we observe first that it is
immediate that
$\psi_i | S^\chi = 0$, $i = 1,\dots,r-2s$,
define $S$
regularly.
It remains to show that
$X_{\psi_i | S^\chi}(x) \in T_x S$, $i = 1,\dots,r-2s$, $x \in S$,
where
$X_{\psi_i | S^\chi}$
is the Hamiltonian vector field associated to the function
$\psi_i | S^\chi$
with respect to the symplectic form
$\omega^\chi$.
This is equivalent to showing that
\begin{equation*}
X_{\psi_i | S^\chi}(x) \left(\psi_j | S^\chi \right) = 0,
\end{equation*}
for
$x \in S$
or, equivalently,
\begin{equation*}
\omega^{\chi}(x) \left(X_{\psi_i | S^\chi}(x), X_{\psi_j | S^\chi}(x)\right) =
0,
\end{equation*}
for
$x \in S$.
We know that
$\omega^{\chi}$
is the pullback of
$\Omega$
to
$S^{\chi}$
and
$\psi_i | S^{\chi}$
is the pullback of
$\psi_i$,
via the inclusion
$S^{\chi} \subseteq U$,
then we have
\begin{equation*}
\omega^{\chi}(x) \left(X_{\psi_i | S^\chi}(x), X_{\psi_j | S^\chi}(x)\right)
=
\Omega(x) \left(X_{\psi_i}(x), X_{\psi_j}(x)\right) = 0,
\end{equation*}
for
$x \in S$,
since
$\psi_i$
are first class constraints, $i = 1,\dots,r - 2s$.
Finally, using the definitions we can easily see that $\operatorname{dim} S^\chi
= 2n - 2s$ and that
$\operatorname{dim} S = 2n - r$ and
from \textbf{(a)}
we know that
$\operatorname{dim} \operatorname{ker}\omega = r - 2s$. We can conclude that
$2n - 2s = \operatorname{dim} S + \operatorname{dim} \operatorname{ker}\omega$.

\textbf{(d)} We know that for $x\in S$, $\operatorname{dim}T_x S
= 2n - r$, and $\operatorname{dim}\ker\omega (x) = r - 2s$; then
using (\ref{720eq}) we obtain $\operatorname{dim}V_x = 2s$. Also
from (\ref{720eq}) we immediately deduce by applying $^\Omega$ to
both sides,
\begin{equation*}%\label{720eqq}
V^\Omega \cap (\ker \omega)^\Omega = TS,
\end{equation*}
in particular $TS \subseteq V^\Omega$. Let $g$ be a given
Riemannian metric on $P$ and let $W_x$ be the $g$-orthogonal
complement of $T_x S$ in $V_x^\Omega$, in particular, $W_x \oplus
T_x S = V_x^\Omega$, for each $x \in S$. Define
\begin{equation*}
S^V = \{\operatorname{exp} (tw_x) \mid  w_x \in W_x, g(x)(w_x,w_x)
= 1, |t| < \tau(x),x\in S\}.
\end{equation*}
By choosing $\tau(x)$ appropriately one can ensure that $S^V$ is a
submanifold and, moreover, it is easy to see from the
definition of $S^V$ that $T_x S^V = W_x \oplus T_x S =
V_x^\Omega$, for each $x \in S$. We leave for later the proof that
$S^V \cap U$ is a symplectic submanifold of $P$, for an
appropriate choice of the open set $U$, which amounts to choosing
$\tau(x)$ appropriately.

Assume that $d\chi^\prime
_1(x),\dots,d\chi^\prime _{2s}(x)$ are linearly independent for $x
\in S^V \cap U^\prime$. Since $\langle d\chi^\prime _i(x), V_x^\Omega\rangle =
\langle d\chi^\prime _i(x), T_x S^V\rangle = 0$ for $x \in S$ and $i =
1,\dots,2s$, we can deduce that $d\chi^\prime _i(x) \in
\left(V_x^\Omega\right)^\circ$, that is, $d\chi^\prime _i(x) \in
V_x^\flat$. Then using (\ref{721eq}), we see that $d\chi^\prime
_1(x),\dots,d\chi^\prime _{2s}(x), d\psi _1(x),\dots,d\psi _{r-2s}(x)$
are linearly independent and span $V^\flat \oplus (\ker
\omega)^\flat= (TS)^\circ$. If $\operatorname{det}(\{\chi^\prime_i,
\chi^\prime_j\}(x)) = 0$ for some $x \in S$ then $\lambda^i
\{\chi^\prime_i, \chi^\prime_j\}(x) = 0$, where at least some
$\lambda^i \neq 0$, $i = 1,\dots,2s$. Let $\lambda^i\chi^\prime_i =
\varphi$, then $\{\varphi, \chi^\prime_j\}(x)=0$, $j = 1,\dots,2s$. On
the other hand, since $\varphi | S = 0$, then $\{\varphi,
\psi_j\}(x)=0$, $j = 1,\dots,r - 2s$. We can conclude that $\varphi
\in I^{(S, U^\prime)}$ and then $X_{\varphi}(x) \in
\operatorname{ker}\omega(x)$, in particular, $X_{\varphi}(x) =
\mu^j X_{\psi_j}(x)$, which implies $\lambda^id\chi^\prime_i(x) =
d \varphi (x) = \mu^j d\psi_j(x)$, contradicting the linear
independence of $d\chi^\prime_1(x),\dots,d\chi^\prime_{2s}(x),
d\psi_1(x),\dots,d\psi_{r-2s}(x)$.

It follows from which precedes that by replacing $\chi_i$ by
$\chi^\prime_i$, $i = 1,\dots,2s$ and $S$ by $S \cap U^\prime$ all
the properties stated in \textbf{(a)}, \textbf{(b)} and \textbf{(c)} are satisfied. In
particular, $S \cap U^\prime = S^{\chi^\prime}$ and $S \cap
U^\prime$ is symplectic. It is now clear that by covering $S$ with
open subsets like the $U^\prime$ we can define $U$ as being the
union of all such open subsets and one obtains that $S \cap U$ is
a symplectic submanifold.

\textbf{(e)} Let $x \in S^\chi$. Then, since $F_\chi =
F - \chi_i c_{\chi}^{ij}\{\chi_j, F\}$, we obtain
\[
\{F_\chi, \chi_k\}(x) = \{F, \chi_k\}(x) - \{\chi_i, \chi_k\}(x)
c_{\chi}^{ij}(x)\{\chi_j, F\}(x) = \{F, \chi_k\}(x) + \{\chi_k, F\}(x) =
0.
\]
Using this we obtain
\begin{multline*}
\{F_\chi, G_\chi\}(x) = \{F_\chi, G - \chi_k c_{\chi}^{kl}\{\chi_l,
G\}\}(x) = \{F_\chi, G\}(x) \\
= \{F, G\}(x) - \{\chi_i, G\}
c_{\chi}^{ij}\{\chi_j, F\}=
\{F, G\}_{(\chi)} (x).  
\end{multline*}

For any $F \in \mathcal{F}(P)$, $x \in S^\chi$ and $k = 1,\dots,2s$,
we have $X_{F_\chi}(x) \chi_k = \{\chi_k, F_\chi\}(x) = 0$, so
$X_{F_\chi}(x) \in T_x S^\chi$. Therefore, for any $Y_x \in T_x S^\chi$,
\begin{multline*}
\omega^\chi(x)\left(X_{F_\chi}(x), Y_x\right) =
\Omega(x)\left(X_{F_\chi}(x), Y_x\right) = d F_\chi(x) Y_x = d
\left(F_\chi|S^\chi\right)(x) Y_x \\
= d \left(F|S^\chi\right)(x)
Y_x=\omega^\chi(x)\left(X_{F|S^\chi}(x), Y_x\right),
\end{multline*}
which shows that $X_{F_\chi}(x) = X_{F|S^\chi}(x)$,
where both Hamiltonian vector fields are
calculated with the symplectic form
$\omega^\chi$.
Using this, for
any $G \in \mathcal{F}(P)$ and any $x \in S^\chi$, one obtains
\[
\{G_{\chi},F_{\chi}\}(x) = X_{F_\chi}(x) G_{\chi} = X_{F_\chi}(x)
G = X_{F|S^\chi}(x) G | S^\chi = \{G | S^\chi, F|S^\chi\}_\chi(x).
\]
The equality $X_{(\chi), F} (x) = X_{F_{\chi}} (x) = X_{F |
S^\chi}(x)$ is an immediate consequence of the previous ones.
The Jacobi identity for the bracket $\{\, , \,\}_{(\chi)}$ follows
using the previous formulas, namely, for $x \in S^\chi$, one
obtains
\[
\{\{F , G\}_{(\chi)}, H\}_{(\chi)}(x) = \{\{F , G\}_{(\chi)} | S^\chi, H |
S^\chi\}_\chi(x) = \{\{F | S^\chi, G | S^\chi\}_\chi, H| S^\chi\}_\chi(x),
\]
where the bracket in the last term is the canonical bracket on the
symplectic manifold $S^\chi$, for which the Jacobi identity is
well known to be satisfied.

\textbf{(f)} The equalities (\ref{equationsss}) and
(\ref{equationssss}) are proven in a straightforward way. The
equations (\ref{constrequalitiesC}) and (\ref{equationsssss}) follow easily
using a technique
similar to the one used in \textbf{(e)}. Using all this, the proof
of the Jacobi identity for the bracket $\{\, , \,\}_{(\chi)}$ on $U$
goes as follows. Let $x \in U$ and let $C$ be such that $x \in
S^{\chi^C}$. For $F,G,H \in \mathcal{F}(U)$ using \textbf{(e)} we
know that the Jacobi identity holds for $\{\, , \,\}_{(\chi^C)}$ on
$S^{\chi^C}$. But then, according to (\ref{equationssss}) it also
holds for $\{\, , \,\}_{(\chi)}$ for all $x \in S^{\chi^C}$.
Now we will prove that
$S^{\chi^C}$
are the symplectic leaves. Since they are defined by equations
$\chi^C_i = 0$, $i = 1,\dots,2s$ on $U$ we need to prove that
$\{F, \chi^C_i\}_{\chi}(x) = 0$, $x \in S^C$, for all $F \in \mathcal{F}(U)$, $i
= 1,\dots,2s$.
Using (\ref{equationssss}) and (\ref{constrequalitiesC}) we see that
$\{F, \chi^C_i\}_{(\chi)}(x) = \{F, \chi^C_i\}_{(\chi^C)}(x) = \{F| S^{\chi^C},
\chi^C_i| S^{\chi^C}\}_{\chi^C}(x) = 0$.
To finish the proof, observe first that, since
$\chi_i = 0$, $\psi_i = 0$,
$i = 1,\dots,2s, j = 1,\dots,r - 2s$
define regularly the submanifold
$S \subseteq U$,
by shrinking
$U$
if necessary and for all $C$ sufficiently small, we have that
$\chi^{C}_i = 0$, $\psi_i = 0$,
$i = 1,\dots,2s, j = 1,\dots,r - 2s$
define regularly a submanifold
$S^C \subseteq U$
and therefore
$\psi_j | S^{\chi^C} = 0$, $j = 1,\dots,r - 2s$,
define regularly
$S^C$
as a submanifold of
$S^{\chi^C}$.
To prove that it is a first class constraint submanifold and that
$\psi_j | S^{\chi^C}$, $j = 1,\dots,r - 2s$,
are first class constraints, that is
$\psi_j | S^{\chi^C} \in I^{(S^C,S^{\chi^C})}$, $j = 1,\dots,r - 2s$,
we proceed in a similar fashion as we did in \textbf{(c)}, replacing
$\chi$ by $\chi^C$.
The fact that
$\psi_j | S^{\chi^C}$, $j = 1,\dots,r - 2s$,
are first class constraints defining
$S^C$
implies that
$\operatorname{dim} S^{\chi^C} =\operatorname{dim}
S^C + \operatorname{dim}\ker \omega^C$.
{}From the definitions one can deduce that
$\operatorname{dim} S^{C} = \operatorname{dim} S$
and
$\operatorname{dim} S^{\chi^C} = \operatorname{dim} S^{\chi}$,
therefore
$\operatorname{dim}\ker \omega = \operatorname{dim}\ker \omega^C$.
\end{proof}
The following theorem summarizes the essential part of the previous lemma.
\begin{theorem}\label{theoremsecondclass}
Let $(P,\Omega)$ be a symplectic manifold, $S \subseteq P$ and let
$\omega$ be the pullback of $\Omega$ to $S$. Assume that
$\ker\omega$ has constant rank. Let $V$ be a vector subbundle of
$TP | S$ such that $V \oplus \ker \omega= (TS)^\Omega$. 
Then there is a symplectic submanifold $S^V$ containing $S$ of dimension
$\operatorname{dim}S +
\operatorname{dim}\ker \omega$ 
such that the condition $T_x S^V =
V_x^\Omega$, for all $x \in S$ holds.
The vector bundle 
$V^\Omega$ is called the \textbf{second class subbundle} tangent to the second class submanifold $S^V$.
The submanifold $S^V$ is locally defined by certain equations $\chi_1=0$, \dots, $\chi_{2s}=0$ as in Lemma 
\ref{mainlemma1section7}. Under the hypotheses of that lemma, namely, that $2s$ is constant on a neighborhood $U$ of $S$, $S^V$ will be one of the symplectic leaves of the Dirac bracket  $ \{\, , \,\}_{(\chi)}$.
\end{theorem}
\begin{remark}(a) 
Under our strong regularity conditions the symplectic  leaves
of the Dirac bracket give a (local) regular foliation
of a neighborhood of the final constraint submanifold $S$. This implies by the
Weinstein splitting theorem (\cite{MR723816})
that there are local charts where the Dirac bracket is constant.

(b) The tangent second class subbundle $V^\Omega$ in a sense (modulo tangency)  classifies all the possible second class constraint submanifolds containing a given submanifold $S\subseteq P$. It carries enough information to write the Dirac brackets along the final constraint submanifold $S$ and therefore also equations of motion, as we show in section \ref{subsectionequationsof motion}.
\end{remark}

\subsection{Equations of motion}\label{subsectionequationsof motion}
We are going to describe equations of motion in the abstract setting of section \ref{subsectionV}, that is, a symplectic manifold
$(P,\Omega)$  and a submanifold
$S \subseteq P$, defined regularly by equations
$\phi_i = 0, $  $i = 1,\dots,a$. 
For this purpose, we need to introduce in
this abstract setting, by definition, the notions of \textit{primary
constraints}, \textit{primary constraint submanifold} 
and the \textit{energy}.\\

The \textbf{\textit{primary constraint submanifold}} is a given submanifold 
$S^\prime \subseteq P$ containing $S$, and in this context, $S$ will be called the final constraint.
We will assume without loss of generality that 
$S^\prime$
 is defined regularly by the equations
$\phi_i = 0, $  $i = 1,\dots,a^\prime$, 
with 
$a^\prime \leq a$, 
where each 
$\phi_i$,  $i = 1,\dots,a^\prime$
will be called  a \textbf{\textit{primary constraint}} while  each
$\phi_i, $  $i = a^\prime +1,\dots,a$ 
will be called 
a \textbf{\textit{secondary constraint}}, for obvious reasons.
In this abstract setting the \textbf{\textit{energy}} is by definition a given
function 
$\mathcal{E} \in  \mathcal{F}(P)$.\\

The equations of motion can be written in the Gotay-Nester form,
\begin{equation}\label{abstractGN}
\Omega(x)(\dot{x}, \delta x) = d\mathcal{E}(x)(\delta x),
\end{equation}
where $(x,\dot{x}) \in T_x S$, for all $\delta x \in T_x S^\prime$. We require that for each $x\in S$, the space of solutions $(x, \dot x)$ of \eqref{abstractGN} is nonempty and forms an affine bundle.

Now we will transform this equation into an equivalent Poisson equation using the Dirac bracket.

The condition
$\{\mathcal{E}, \psi\}(x) = 0$, for all $x \in S$ and all first class constraints $\psi$
will appear later as a necessary condition for existence of solutions for any given initial condition in $S$, so we will assume it from now on.\\

The \textit{\textbf{total energy}} is defined by
\begin{equation}\label{defnoftotalhamiltonian}
\mathcal{E}_T = \mathcal{E} + \lambda^i \phi_i,\, i = 1,\dots,a^\prime
\end{equation}
where the functions 
$\lambda^i \in C^{\infty}(P)$, $i = 1,\dots,a^\prime$ 
must satisfy, by definition,
$\{\mathcal{E}_T, \phi_j\}(x) = 0$, $j = 1,\dots,a$, $x \in S$ or, equivalently,
\begin{equation*}%\label{conditiononlambda}
 \{\mathcal{E}, \phi_j\}(x) +  \lambda^i  \{\phi_i, \phi_j\}(x) = 0, \, x\in S, \,  i =
1,\dots,a^\prime, \,j= 1,\dots,a,
\end{equation*}
sum over 
$i = 1,\dots,a^\prime$.

\begin{assumption}\label{Lambda}
We assume that the solutions 
$(\lambda^1,\dots,\lambda^{a^\prime})$ form a nonempty affine bundle 
$\Lambda \rightarrow S$. Equivalently, if we call $d_\Lambda$ the dimension of the fibers,
$s':=\operatorname{rank}(\{\phi_i, \phi_j\}(x))=a'-d_\Lambda$ is constant on $S$.  
\end{assumption}
By taking a sufficiently small neighborhood $U$, we can assume that $\Lambda$  is a trivial bundle over $S\cap U$, which we extend arbitrarily as a trivial affine bundle on $U$.  

For each section of $\Lambda$
one has a Hamiltonian $\mathcal{E}_T (x) = \mathcal{E}(x) + \lambda^i (x) \phi_i(x)$, $x
\in U$, and the equation of motion on $S\cap U$,
\begin{equation*}%\label{eqnonfinalconstraint}
 X_{\mathcal{E}_T} = (d\mathcal{E}_T)^\sharp.
\end{equation*}

\paragraph{Equations of motion in terms of the Dirac bracket.} 
Local equations of motion in terms of a Dirac bracket become simpler when the second class constraints are adapted to the primary constraint, as we will explain next. This has been considered by Dirac. The case of general second class constraints is important for the global case, and will not be studied here.

Assume that the primary and final constraints $S^\prime \supseteq S$ are defined by
\begin{align}\label{calSprime}
\phi_i 
&= 0, \, i = 1,\dots,a^\prime, \,\, \mbox{defines}\,\,\, S^\prime\\\label{SprimeC}
\phi_i 
&= 0, \, i = 1,\dots,a^\prime, a^\prime + 1,\dots,a, \,\, \mbox{defines}\,\,\, S,
\end{align}
regularly on a certain neighborhood $U$ of $S^\prime$.
For choosing first class and second class constraints one can proceed as follows.

\begin{assumption}\label{2s_constant}
The rank of the matrix $\Phi(x) = ( \{\phi_i, \phi_j\}(x) )$, $i,j = 1, \dots,a$, which we will call $2s$,  is constant on an open neighborhood $U$ of $S'$. 
\end{assumption}

By shrinking $U$ if necessary, one can choose
$\{\chi_1, \dots, \chi_{2s}\}  \subseteq \{\phi_1, \dots, \phi_a \}$ in such a way that 
$( \{\chi_i, \chi_j\}(x) )$ is a regular submatrix. Without loss of generality, assume that these functions are the last $2s$, that is, $\chi_i=\phi_{a-2s+i}$, $i=1, \dots, 2s$.
For each $\phi_i$, $i=1, \dots, a-2s$, there are $\alpha_i^k (x)$ such that
$\psi_i (x):= \phi_i(x) + \alpha_i^k (x) \chi_k(x)$, $k=1, \dots, 2s$,  defined on $U$ satisfy
\begin{equation*}
\{\psi_i, \phi_l\}(x)= \{\phi_i, \phi_l\}(x) + \alpha_i^k (x) \{\chi_k, \phi_l\}(x) = 0,
\end{equation*}
for all $x \in S$, and all $l = 1, \dots,a$. 
Therefore, these $\psi_i$ are first class constraints and 
$\{\chi_1, \dots, \chi_{2s}\}$ are second class constraints, and for each $x \in U$ the 
\[
d\psi_1(x),\dots,d\psi_{a - 2s}(x), d\chi_1(x),\dots,d\chi_{2s}(x),
\]
are linearly independent, $x \in U$. We will sometimes call $(\psi_1,\dots\psi_{a-2s},\chi_1,\dots\chi_{2s})$ a complete set of first class and second class constraints for $S$.\\

Now we will show how to choose a complete set of first class and second class constraints \textbf{\textit{adapted to the primary constraints}}. 
This, as a definition, means choosing
\[
\psi_1(x),\dots,\psi_{a - 2s}(x), \chi_1(x),\dots,\chi_{2s}(x),
\]
as before in such a way that, for $s^\prime$ equal to the rank of the submatrix 
$\Phi^\prime = ( \{\phi_i, \phi_j\}(x) )$, $i = 1,\dots,a^\prime ,j = 1,\dots,a$, 
\[
\psi_1,\dots,\psi_{a^\prime - s^\prime}, \chi_1,\dots,\chi_{s^\prime}
\]
are primary constraints.

To achieve this, let $\{\chi^\prime_1, \dots, \chi^\prime_{s^\prime}\}  \subseteq \{\phi_1, \dots, \phi_{a^\prime} \}$ 
such that the rows
$\{\chi^\prime_k, \phi_j\}(x)$, $k = 1,\dots,s^\prime$, 
are linearly independent. 
Without loss of generality, assume that these functions are the last $s'$, that is, $\chi'_i=\phi_{a'-s'+i}$, $i=1, \dots, s'$. 
Then, for each $i = 1,\dots,a^\prime - s^\prime$, there are 
$\alpha_i^k (x)$, \,  $k = 1,\dots,s^\prime$ such that
\begin{equation*}
\psi^\prime_i (x):= \phi_i(x) + \alpha_i^k (x) \chi^\prime_k(x)
\end{equation*}
defined on $U$ satisfy
\begin{equation*}
\{\psi^\prime_i , \phi_l\}(x)
= 
\{\phi_i, \phi_l\}(x) + \alpha_i^k (x) \{\chi^\prime_k, \phi_l\}(x) = 0,
\end{equation*}
for all 
$i = 1,\dots,a' - s'$,
$x \in S$, and all $l = 1,\dots,a$. 
Therefore, these $\psi^\prime_i$ form a set of $a^\prime - s^\prime$ first class constraints. Note that $(\psi'_1,\dots\psi'_{a'-s'},\chi'_1,\dots,\chi'_{s'})$ is a complete set of primary constraints, in the sense that their differentials are linearly independent on $U$, and every primary first class constraint $\psi$ can be written as $\psi(x)=b_1(x)\psi'_1(x)+\dots+b_{a'-s'}(x)\psi'_{a'-s'}(x)$.

In order to obtain a complete set of first class and second class constraints we first choose
$\{\chi''_{s'+1}, \dots, \chi''_{2s}\}  \subseteq \{\phi_{a'+1}, \dots, \phi_{a} \}$ 
such that the rows
$\{\chi^\prime_{k'}, \phi_j\}(x)$, $\{\chi''_{k''}, \phi_j\}(x)$ $k' = 1,\dots,s^\prime$, $k''=s'+1,\dots,2s$  
are linearly independent. This can be done because the matrix $( \{\phi_i, \phi_j\}(x) )$, $i,j = 1,\dots,a$, has rank $2s$. As before, assume that $\chi''_i=\phi_{a-2s+i}$, $i=s'+1, \dots, 2s$.

Then, for each $i = a' - s'+1,\dots,a-2s$
there are 
$\alpha_i^{k'} (x)$, $\alpha_i^{k''} (x)$, $k' = 1,\dots,s^\prime$,  $k'' = s^\prime+1,\dots,2s$,
such that
$\psi''_i (x):= \phi_{s'+i}(x) + \alpha_i^{k'} (x) \chi^\prime_{k'}(x) +\alpha_i^{k''} (x) \chi''_{k''}(x)$, $k' = 1,\dots,s^\prime$,  $k'' = s^\prime+1,\dots,2s$,
satisfies
\begin{equation*}
\{\psi''_i , \phi_l\}(x)
= 
\{\phi_{s'+i}, \phi_l\}(x)
 + \alpha_i^{k'} (x) \{\chi^\prime_{k'}, \phi_l\}(x)
 + \alpha_i^{k''} (x) \{\chi''_{k''}, \phi_l\}(x) = 0,
\end{equation*}
for all 
$x \in S$, and all $l = 1,\dots,a$.

It can be shown using the definitions that
\begin{multline*}
(d\psi'_1(x),\dots,d\psi'_{a'-s'}(x),d\psi''_{a'-s'+1}(x),\dots,d\psi''_{a-2s}(x),\\
d\chi'_1(x),\dots,d\chi'_{s'}(x),d\chi''_{s'+1}(x),\dots,d\chi''_{2s}(x))
\end{multline*}
are linearly independent for $x\in S$.
Note that no nontrivial linear combination of the $\psi''_i$ is a primary first class constraint.

Therefore,
\begin{multline*}
(\psi_1,\dots,\psi_{a - 2s}, \chi_1,\dots,\chi_{2s}): = \\
(\psi'_1,\dots,\psi'_{a'-s'},\psi''_{a'-s'+1},\dots,\psi''_{a-2s},\chi'_1,\dots,\chi'_{s'},\chi''_{s'+1},\dots,\chi''_{2s})  
\end{multline*}
form a complete set of first class and second class constraints adapted to the primary constraints.\\

We will now take advantage of these adapted constraints to write the equations of motion in a simpler way. The total energy defined on $U$ is
\[
\mathcal{E}_T=\mathcal{E}+\lambda'^i\psi'_i+\mu'^j\chi'_j,
\]
$i=1,\dots,a'-s'$, $j=1,\dots,s'$. As usual, we impose on $\mathcal{E}_T$ the condition that it is first class, which means that $\{\mathcal{E},\psi_i\}(x)=0$, $i=1, \dots,a-2s$ and
\[
\{\mathcal{E},\chi_i\}(x)+\mu'^j \{\chi'_j,\chi_i\}(x)=0,
\]
$i=1, \dots,2s$, $x\in S$. From this we obtain well-defined $\mu'^j (x)$, $x\in S$, which we extend arbitrarily for $x\in U$. Then the equations of motion are written as
\[
X_{\mathcal{E}+\mu'^j\chi'_j}(x)+\lambda'^iX_{\psi'_i}(x),
\]
$x\in S$, where $\lambda'^i$, $i=1,\dots,a'-s'$, are completely arbitrary real numbers. 
We can also write this equation in the Poisson form
\[
\dot F= \{F,\mathcal{E}+\mu'^j\chi'_j\}+\lambda'^i \{F,\psi'_i\}.
\]

Now we will see a simpler expression for the equations using Dirac brackets rather than the canonical bracket. Since $\mathcal{E}_T$ is first class, then  $X_{\mathcal{E}_T}(x)=X_{(\chi),\mathcal{E}_T}(x)$ for $x\in S$ (see Lemma \ref{mainlemma1section7}(e)). Then the equations of motion can be written as
\begin{equation}\label{eq:F_HT_Dirac_bracket}
\dot F= \{F,\mathcal{E}_T\}_{(\chi)}.
\end{equation}
This equation, obtained in an abstract setting, is similar to the one in page 42 of \cite{Dirac1964} for the case $P=T^*Q$.

We can rewrite the total energy as 
\[
\mathcal{E}_T=\mathcal{E}+\lambda'^i\psi'_i+\mu'^j\chi'_j=\mathcal{E}+\lambda'^i(\phi_i+\alpha_i^k\chi'_k)+\mu'^j\chi'_j,
\]
where $i=1, \dots, a'-s'$ and $j=1, \dots, s'$, and therefore (\ref{eq:F_HT_Dirac_bracket}) can be written as
\[
\dot F= \{F,\mathcal{E}\}_{(\chi)}+\lambda'^i \{F,\phi_i\}_{(\chi)},
\]
or equivalently,
\[
X_{(\chi),\mathcal{E}_T} (x)=X_{(\chi),\mathcal{E}} (x)+\lambda'^iX_{(\chi),\phi_{i}} (x), 
\]
$x\in S$.

We will denote
\begin{equation}\label{abridged_total_energy}
\mathcal{E}_{AT}=  \mathcal{E}+\lambda'^i\phi_{i}  
\end{equation}
$i=1, \dots,a' -s'$, and call it the \textbf{\textit{abridged total energy}}.
It has the property that it provides the same dynamical information as the total energy in terms of the Dirac bracket. Moreover, for each $x \in S$ the map $\mathbb{R}^{a'-s'} \to \Lambda_x$ given by $(\lambda'^1, \dots, \lambda'^{a'-s'}) \mapsto X_{(\chi),\mathcal{E}_{AT}} (x)$ is an isomorphism. In particular, for any solution curve $x (t)$ of $X_{\mathcal{E}_T}$ on $S$, we have
\[
(x (t), \dot x (t))=X_{(\chi),\mathcal{E}_{AT}} (x (t))
\]
for a (uniquely determined) choice of $\lambda'^i (t)$, $i=1, \dots, a'-s'$.

We have proven the following theorem.

\begin{theorem}\label{theorem321}
Let the primary $S ^\prime$ and final constraint $S$ be submanifolds defined regularly by equations 
(\ref{calSprime}), (\ref{SprimeC}) on a neighborhood $U \subseteq P$ of $S ^\prime$, where $(P,\Omega)$ is a symplectic manifold. 
Let $2s$ be the rank of the matrix $\Phi = ( \{\phi_i, \phi_j\}(x) )$, $i,j = 1, \dots,a$, which, by Assumption \ref{2s_constant},  is constant on $U$.
Let
$s^\prime$ be the rank of the matrix 
$\Phi^\prime = ( \{\phi_i, \phi_j\}(x) )$, $i = 1,\dots,a^\prime, j = 1,\dots,a$, which, by Assumption
\ref{Lambda} is constant on $S$.
Choose $s^\prime$ linearly independent rows 
$( \{\chi_i', \phi_j\}(x) )$, $i = 1,\dots,s^\prime, j = 1,\dots,a$
of 
$\Phi^\prime$, say w.l.o.g. $\chi'_i=\phi_{a'-s'+i}$, $i=1, \dots, s'$. Then for each $i=1, \dots, a'-s'$ one can find coefficients
$\alpha_i^k (x)$ 
such that
$\psi^\prime_i (x):= \phi_i(x) + \alpha_i^k (x) \chi^\prime_k(x)$, $k = 1, \dots,s^\prime$, defined on $U$ constitute a maximal set of linearly independent primary first class constraints.
 Then
$(\psi'_1, \dots,\psi'_{a'-s'},\chi'_1, \dots,\chi'_{s'})$ is a complete set of primary constraints and one can extend it to a complete set of first class and second class constraints
\begin{multline*}
(\psi_1,\dots,\psi_{a - 2s}, \chi_1,\dots,\chi_{2s}): =\\
(\psi'_1,\dots,\psi'_{a'-s'},\psi''_{a'-s'+1},\dots,\psi''_{a-2s},\chi'_1,\dots,\chi'_{s'},\chi''_{s'+1},\dots,\chi''_{2s}).
\end{multline*}
Let $\mathcal{E}\colon U \to \mathbb{R}$ be a given energy satisfying 
$\{\mathcal{E},\psi_i\}(x)=0$, $x\in S$, $i=1, \dots, a-2s$.
Consider the following \emph{abridged total energy}
\[
\mathcal{E}_{AT}=\mathcal{E}+\lambda'^i\phi_{i},
\]
$i=1, \dots, a'-s'$, where $\lambda'_i$ are arbitrary real parameters. The corresponding Hamiltonian vector field with respect to the Dirac bracket is
\begin{equation}\label{eq_dirac_evolution_thm}
 X_{(\chi),\mathcal{E}_T} (x)=X_{(\chi),\mathcal{E}} (x)+\lambda'^iX_{(\chi),\phi_{i}} (x), 
\end{equation}
$x\in U$. For $x\in S$, this vector field is tangent to $S$ and gives the equations of motion. Equivalently, we can write the equations of motion as
\begin{equation}\label{eq:poisson_dirac_evolution_thm}
\dot F= \{F,\mathcal{E}\}_{(\chi)}+\lambda'^i \{F,\phi_{i}\}_{(\chi)}.  
\end{equation}
Since $X_{(\chi),\mathcal{E}_T}$ is tangent to $S$,  the evolution of a function $f$ on $S$ is given by \eqref{eq:poisson_dirac_evolution_thm} for any $F$ such that $f=F|S$.  
\end{theorem}
\begin{remark}
Equation (\ref{eq:poisson_dirac_evolution_thm}) represents an elaboration on the equation (\ref{eq:F_HT_Dirac_bracket}). It has interesting advantages which may be useful in practice, because
it only requires
calculating
$(\chi'_1, \dots,\chi'_{s'},\chi''_{s'+1}, \dots,\chi''_{2s}) $, which is a set of second class constraints 
adapted to the primary constraint.
It will be generalized to extend its applicability to examples where the constraints are \textit{foliated}, as we will see in section~\ref{sectionanextensionofetc}.
\end{remark}

\begin{remark}\label{physicalstate}
The geometric meaning of $s'$ is given by the fact that $a'-s'$ is the number of independent primary first class constraints $\psi'_i$, $i=1,\dots, a'-s'$. This number coincides with
$d^{(c)}(x)$ and with the dimension of the leaves of the foliation $\widetilde K_c$ of $M_c$  (see Theorem \ref{lemmmma}). Also, the Hamiltonian vector fields $X_{\psi'_i}, \dots, X_{\psi'_{a'-s'}}$ are tangent to the leaves of $\widetilde K_c$. Using the formula $[X_{f},X_g]=X_{-\{f,g\}}$ it is easy to prove that the iterated Poisson brackets of the $\psi'_i$ give functions whose Hamiltonian vector fields are tangent to the leaves of $\widetilde K_c$. Each one of these functions, let us call it $f$, satisfies that $X_{f}(x)=\lambda^i(x)X_{\psi'_i}(x)$ for all $x\in M_c$, but not necessarily $f$ is a linear combination of the $\psi'_i$,
on a neighborhood  of $M_c$, so it might not be a primary first-class constraint. One can prove that such a function is zero on $M_c$, but since $df(x)=\lambda^i(x)d\psi'_i(x)$ for all $x\in M_c$, it can not be included in a set of constraint functions defining $M_c$ regularly and containing the $\psi'_i$, in other words, it is not an \emph{independent} secondary first class constraint.
There are examples in the literature where the Poisson bracket of two primary first class constraints is an independent secondary first class constraint, for instance, that is the case in \cite[p.\ 2]{MR1191617}; for that example Assumption \ref{K2} (a) does not hold (but \ref{K2} (b) does).

According to \cite{Dirac1964}, pages 23--24, primary first class constraints and their iterated brackets represent transformations that do
not change the physical state of the system, which in this case is implied directly from the requirement of determinism. We have seen that under our regularity conditions, those iterated brackets
do not really add any new such transformations besides the ones given by the primary first class constraints.

In addition, note that $a-2s$ is also the dimension of the kernel of the presymplectic form on $S$
(see Theorem \ref{mainlemma1section7} (a)), which is the number of first class constraints, and also the dimension of the leaves of $K_c$ (see Assumption \ref{K1}).
Since the Hamiltonian vector fields $X_\psi$ corresponding to the first class constraints generate the integrable distribution associated to $K_c$, then the transformations that preserve the leaves of $K_c$ are those which preserve the physical state. So, points of $\bar M_c=M_c/K_c$ would represent exactly the physical states.

\cite[p.\ 25]{Dirac1964} introduces the notion of extended Hamiltonian. In our context, we should define the notion of \textbf{\textit{extended energy}},
\begin{equation}\label{extendedenergy}
\mathcal{E}_E = \mathcal{E}_T  + \lambda''^i \psi''_i,
\end{equation}
$i = 1, \dots,a-2s -(a' -s')$
where the $\psi''^i$ represent the secondary first class constraints. 
Both $\mathcal{E}_E$ and $\mathcal{E}_T$ give the same dynamics on the quotient manifold
$\bar M_c$.
If $(M_c,\omega_c)$ is symplectic, there are only second class constraints among the $\phi_i$, $i=1, \dots, a$, and conversely. In this case, $\mathcal{E}_E=\mathcal{E}_T$. 
\end{remark}

\section{Dirac structures}\label{sectiondiracstructures}
The rest of the paper is devoted to generalizing the Dirac and Gotay-Nester theories of constraints.
In this section we review a few basic facts related to Dirac
structures and Dirac dynamical systems following \cite{MR998124} and \cite{MR2265464, MR2265469}, which are
of direct interest for the present
paper.

We introduce the flat  $D^\flat$ and the orthogonal $^D$ operators with respect to a given Dirac structure, which is important to describe the algorithm, in the next section. In fact, we obtain a technique which imitates the Gotay-Nester technique, by replacing $\omega$ by $D$.
\paragraph{Dirac structures on vector spaces.}
Let $V$ be an $n$-dimensional vector space and let
 $V^*$ be its dual space. Define
 the
\textit{symmetric pairing} $\left\langle \! \left\langle \cdot  \,,\cdot\right\rangle \!
\right\rangle$ on $V \oplus V^*$ by
 \[
 \left\langle \! \left\langle (v_1,\alpha_1),(v_2,\alpha_2)\right\rangle \!
\right\rangle = \left\langle \alpha_1,v_2 \right\rangle + \left\langle
\alpha_2,v_1\right\rangle,\]
 for $(v_1,\alpha_1), (v_2,\alpha_2) \in V \oplus V^*$
where $\left\langle \cdot \, , \cdot \right\rangle$ is the natural pairing
between $V^*$ and $V$.
A \emph{linear Dirac structure} on $V$,
(sometimes called simply a \emph{Dirac structure} on $V$),
is a subspace $D \subseteq V \oplus V^*$
such that $D=D^\perp$, where $D^\perp$ is the orthogonal of $D$
relative to the pairing $\left\langle \! \left\langle \cdot \,,\cdot\right\rangle \!
\right\rangle$. One can easily check the following result.
\begin{lemma}\label{defdiracstr}
A vector subspace $D \subseteq  V \oplus V^*$ is a Dirac structure
on $V$ if and only if it is maximally isotropic with respect to the symmetric
pairing.
A further equivalent condition is given by $\operatorname{dim}D = n$ and
$\left\langle \! \left\langle (v_1,\alpha_1),(v_2,\alpha_2)\right\rangle \!
\right\rangle = 0$
for all
$(v_1,\alpha_1), (v_2, \alpha_2) \in D$.
\end{lemma}
One of the main examples is the Dirac structure naturally defined on a
presymplectic vector space $(V, \omega)$ by
\[
D_{\omega} = \{(v, \alpha) \in V\oplus V^\ast \mid \alpha = \omega^\flat (v)\}.
\]
Here the definition of $\omega^\flat$
is the standard one, namely,
$\omega^{\flat} \colon  V \rightarrow V^\ast$ is defined by
$\omega^{\flat} (v)(w) = \omega (v, w)$, for all $v, w \in V$.
We also recall the definition of orthogonality on $V$ associated to a given
presymplectic form $\omega$. For a given subset
$W \subseteq V$ we define
the $\omega$-orthogonal complement
$W^{\omega}$ by
\[
W^{\omega} = \{v \in V \mid \omega (v, w) = 0 \,\, \mbox{for all}\,\, w \in W\}.
\]
We recall some standard facts in the following lemma, omitting the proof.
\begin{lemma}\label{circperp}
Let $\omega$ be a presymplectic form on a vector space $V$ and let $W \subseteq
V$ be any
vector subspace. Then $W^{\flat} : = \omega ^\flat (W) =
\left(W^{\omega}\right)^{\circ} $, where the right hand side denotes the
annihilator of $W ^\omega$.
\end{lemma}
The following proposition is a direct consequence of propositions 1.1.4 and
1.1.5 in
\cite{MR998124}.
\begin{proposition}
\label{diracpresymplectic}
Given a Dirac structure $D \subseteq V \oplus V^*$ define the subspace $E_D
\subseteq V$ to be the projection of $D$ on $V$. Also, define the 2-form
$\omega_D$ on $E_D$ by
$\omega_D(v,w) = \alpha (w)$, where
$v \oplus \alpha \in D$. (One checks that this definition of  $\omega _D$ is
independent of the choice of $\alpha$).
Then, $\omega_D$ is a skew form on $E_D$.
Conversely, given a vector space $V$, a subspace $E \subseteq V$
and a skew form $\omega$ on $E$, one sees that $D_{\omega} = \{v \oplus \alpha
\mid v \in E,\,\,\alpha (w) = \omega (v, w) \mbox{ for all } w
\in E \}$ is the unique Dirac structure $D$  on $V$ such that $E_D = E$ and
$\omega_D =
\omega$.

$D$ is the Dirac structure associated to a presymplectic form $\omega$ on $V$,
as explained before, if and only if $E_D = V$ and $\omega = \omega_D$.
\end{proposition}

\paragraph{The operators $D^\flat$ and $^D$.}
There is a natural extension, which is important in the present
paper, of the previous definition of $\omega^\flat$ for the case of a
general Dirac structure $D \subseteq V\oplus V^\ast$. For a given
Dirac structure $D \subseteq V\oplus V^*$ and given $X \in E_D$
define the \textit{set} $D^\flat(X)$, sometimes denoted $X^\flat$ for
short, by
\[ D^\flat(X)=\{\alpha \in V^*\mid ( X, \alpha) : = X \oplus \alpha \in D\}.
\]
Note that $D$ is the Dirac structure associated to a presymplectic
form on $V$, that is, $D = D_{\omega}$, if and only if
$E_D = V$
and for each
$X\in V$, the set $D^\flat(X)$ has a single element, more
precisely, $D^\flat(X) = \{\omega^{\flat}(X)\}$. In this
sense, $D^\flat(X)$ generalizes $\omega^\flat(X)$. For a given
subset $W \subseteq E_D$ define $D^\flat(W)$, also denoted
$W^\flat$, by
\[
D^\flat(W)= \bigcup_{X \in W}{D^\flat(X)}.
\]
If $W$ is a subspace, then $D^\flat(W)$ is a subspace of $V^*$.

It is straightforward to check that for all $X, Y \in E_D$,
\[
\{\omega_D (X,Y)\} = D^\flat(X)(Y)=-D^\flat(Y)(X) = -\{\omega_D (Y,X)\}.
\]

The notion of orthogonal complement with respect to $\omega$ can be generalized
as follows. For any subset $W \subseteq E_D$ define $W^D$ by
\[
W^D = \{X \in V\mid D^\flat(Y)(X)=\{0\},\,\,\mbox{for all}\,\, Y \in W\}.
\]
Clearly,
\[
W^D \cap E_D = \{X \in E_D\mid D^\flat(Y)(X)=\{0\},\,\mbox{for
all}\,\, Y \in W\}.
\]
It is easy to check that for any subspace
$W \subseteq E_D$ one has
\[
W^D \cap E_D = W^{\omega_D}.
\]
We recall also that, since $\omega_D \colon  E_D \times E_D \rightarrow \mathbb{R}$
is a presymplectic form on $E_D$, one has, according to Lemma \ref{circperp} and
with a
self-explanatory notation,
\[
\left(W^{\omega_D}\right)^{\circ_{E_D}} = W^{\flat_{\omega_D}}.
\]
The following proposition generalizes Lemma \ref{circperp} and is one of the
ingredients of the
constraint algorithm described in section \ref{sectionconstalgforrdirac}.
\begin{proposition}\label{circperpD}Let $D \subseteq V\oplus V^\ast$ be a given
Dirac
structure and let $W \subseteq E_D$ be a given subspace.
Then $D^\flat (W)\equiv W^\flat = (W^D)^\circ$.
\end{proposition}
\begin{proof} We first show that
$W^\flat \subseteq (W^D)^\circ$.
Let $\alpha \in W^\flat$, say
$\alpha \in D^\flat(X)$ for some $X \in W$.
Since $D^\flat(X)(Y)=\{0\}$ for all
$Y \in W^D$ by definition of $W^D$, it follows that $\alpha \in (W^D)^\circ$.
To prove the converse inclusion, we first observe that it is equivalent to prove
that
$(W^\flat)^\circ \subseteq W^D$.
Let
$X \in (W^\flat)^\circ$,
then
$D^\flat(Y)(X) = 0$,
for all
$Y \in W$, which is an immediate consequence of the definitions. But, also by
definition, this implies that
$X \in W^D$.
\end{proof}

\paragraph{Dirac structures on manifolds.}
We will give the definition and some basic properties of Dirac
manifolds, following \cite{MR998124} and \cite{MR2265464, MR2265469}, using the
notation of the latter.

A \emph{Dirac structure} $D$ on a manifold $M$ is a subbundle of
the Whitney sum $D \subseteq TM \oplus T^\ast M$ such that for
each $x\in M$, $D_x \subseteq T_xM \oplus T_x^\ast M$ is a Dirac
structure on the vector space $T_x M$. A \emph{Dirac manifold} is
a manifold with a Dirac structure on it. From Proposition
\ref{diracpresymplectic} we deduce that a Dirac structure $D$ on
$M$ yields a distribution $E_{Dx} \subseteq T_x M$ whose dimension
is not necessarily constant, carrying a presymplectic form
$\omega_D(x) \colon  E_{Dx} \times E_{Dx} \rightarrow \mathbb{R}$, for
all $x \in M$. We can also deduce the following theorem, whose
detailed proof appears in \cite{MR2265464}.
\begin{theorem}
\label{diracpresymplecticmanifold} Let $M$ be a manifold and let
$\omega$ be a 2-form on $M$. Given a distribution $E$ on $M$,
define the skew-symmetric bilinear form $\omega_E$ on $E$ by
restricting $\omega$ to $E  \times E $. For each $x \in M$ let
\begin{align*}
D_{\omega_{E}x} & = \left\{ (v_x,\alpha_x) \in T_x{M} \oplus T_x^*{M} \mid  v_x
\in E(x) \right.
\\
& \qquad \qquad \textup{and}\,\, \left. \alpha_x(w_x) =
\omega_{E}(x)(v_x,w_x)\,\, \textup{for all} \,\,w_x
\in E(x) \right\}.
 \end{align*}
Then $D_{\omega_{E}} \subseteq TM \oplus T^*M$ is a Dirac structure on $M$. It
is the only
Dirac structure $D$ on $M$ satisfying
$E(x) = E_{Dx}$ and $\omega_E(x) = \omega_{D}(x)$, for all $x \in M$.
\end{theorem}

Using Lemma \ref{circperp} and Proposition \ref{circperpD} we can
easily deduce the following proposition.
\begin{proposition}\label{circperpmanifold}
Let $D$ be a Dirac structure on $M$ and let $W_x$ be a subspace
of $E_{Dx}$ for each $x\in M$; then, with a self-explanatory
notation, the following equalities hold for each $x\in M$:
\begin{align*}
W_x^{\flat_{\omega_D}}
&=
\left(W_x^{\omega_D}\right)^{\circ_{E_D}}\\
W_x^\flat
&=
(W_x^D)^\circ.
\end{align*}
\end{proposition}
A Dirac structure $D$ on $M$ is called \emph{integrable} if the
condition
\[ \left\langle L_{X_1}{\alpha_2},X_3\right\rangle +\left\langle
L_{X_2}{\alpha_3},X_1\right\rangle +\left\langle
L_{X_3}{\alpha_1},X_2\right\rangle =0 \]
is satisfied for all pairs of vector fields and 1-forms $(X_1,\alpha_1)$,
$(X_2,\alpha_2)$, $(X_3,\alpha_3)$ that take values in $D$ and
where $L_X$ denotes the Lie derivative along the vector field $X$
on $M$. This definition encompasses the notion of closedness for
presymplectic forms and Jacobi identity for brackets. The
following fundamental theorem was proven in \cite{MR998124}.
\begin{theorem}\label{presymplfoliation}
Let $D$ be an integrable Dirac structure on a manifold $M$. Then
the distribution $E_D$ is involutive. If, moreover, the hypotheses of the
Stefan-Sussmann theorem (\cite{MR0321133}) are satisfied, for each $x\in M$
there exists a uniquely determined embedded submanifold $S$ of $M$
such that
$x \in S$
and
$T_y S = E_{Dy}$ for all $y \in S$.
In other words, $S$ is an integral submanifold of
$E_D$.
Each integral submanifold $S$
carries a presymplectic form $\omega_{D,S}$ defined by
$\omega_{D,S}(x) = \omega_D (x)$, for each $x\in S$.
\end{theorem}

\subsection{Dirac dynamical systems}\label{subsectiondiracdynamicalsystems}

Lagrangian and Hamiltonian mechanics has been developed in interaction with
differential, symplectic, and Poisson geometry; a few references are
\cite{%
Abraham_Marsden:Foundations_of_Mechanics,
Arnold:Mathematical_methods_of_classical_mechanics,
MR1021489,MR0980716,
Marsden_Ratiu:Introduction_to_mechanics_and_symmetry,
Cendra_Marsden_Ratiu:Lagrangian_Reduction_by_Stages,
MR2847777,MR2894526%
}. 
Dirac dynamical systems in the integrable case represent a synthesis and a generalization of both.

Nonholonomic mechanics represents a generalization of Lagrangian and Hamiltonian mechanics
and is a long-standing branch of mechanics, engineering and
mathematics. Some references and historical accounts on the subject are
\cite{%Bloch_Krishnaprasad_Marsden_Murray:Nonholonomic_mechanical_systems_with_symmetry,
Neimark_Fufaev:Dynamics_of_Nonholonomic_Systems,
MR1323130,
MR1629279,
Cendra_Marsden_Ratiu:Geometric_mechanics_Lagrangian_reduction_and_nonholonomic_systems,
MR1900053,
Cortes:Geometric_control_and_numerical_aspects_of_nonholonomic_systems,
Bloch:Nonholonomic_mechanics_and_control,
MR2067586,
MR2486576,
MR2492630,
MR2511300,
MR2847777,
MR2894526%
}.
Some references that are more closely related to this paper are \cite{%
MR2265464,
MR2265469}.

Dirac dynamical systems  (\ref{11drracds}) in the not necessarily integrable case 
may be viewed as a synthesis and a generalization of nonholonomic mechanics from the Lagrangian and the Hamiltonian points of view.
They
can be written equivalently as a collection of systems of the type
\begin{align}\label{classicalm122}
\operatorname{i}_{\dot{x}}\omega_{D} (x) = d\mathcal{E}(x) | E_D.
\end{align}

A related approach has been studied
in \cite{MR0914056,MR1631383, MR1323130, MR1324157, Blankenstein2004}.
As it was shown in \cite{MR2265464, MR2265469}, one can write nicely the
equations of nonholonomic mechanics, on the Lagrangian side, using the Dirac differential.
We will show how this is related to system (\ref{11drracds}).
It was also shown in \cite{MR2265464, MR2265469}
and references therein how to use Dirac structures in LC circuit theory, on the
Lagrangian side. On the Hamiltonian side, Poisson brackets for LC circuits
were written in \cite{MR1180398,MR1324157}, see also \cite{MaschkeMemo1998}.
In 
\cite{Bloch1997} the Hamiltonian structure for nonlinear LC circuits in the
framework of Dirac structures is investigated and simple and effective formulas
are described.

In \cite{CeMaRaYo2011} a further unification, including reduction, is presented,
which is consistent with mechanics on Lie algebroids
(\cite{%
MR1365779,
MR2232870,
MR2299851,
MR2299850,
MR2394515,
MR2486576%
}).

\section{The constraint algorithm for Dirac
dynamical systems (CAD)}\label{sectionconstalgforrdirac}
Some results in this section, in particular the CAD and examples of nonholonomic mechanics, are proven for not necessarily integrable Dirac structures.

First, we shall briefly consider the case of an integrable Dirac
structure. This includes the case of a constant Dirac structure on a vector
space, that is, one that is invariant under translations, which includes the one used for
LC circuits. Later on we will consider the general case.

\paragraph{The case of an integrable Dirac manifold.}
Assume that $D$ is an integrable Dirac structure on a manifold $M$.
Each solution curve $x(t)$ of the system (\ref{11drracds}) must
satisfy $\left(x(t), \dot{x}(t)\right) \in E_D\left(x(t)\right)$
for all $t$, which implies, according to Theorem
\ref{presymplfoliation} and equation (\ref{classicalm122}), that it must be a
solution curve to the equation
\begin{align}\label{equation15}
\omega_{D,S} (x)(\dot{x}, \, \cdot \,) = i^\ast_S d\mathcal{E}(x)
\end{align}
on a presymplectic leaf $S$, where $i_S\colon S \to M$ is the inclusion,
which can be solved using the Gotay-Nester algorithm.
Such a procedure to solve (\ref{11drracds}) might be useful in those cases where
the
presymplectic leaves $(S,
\omega_{D,S})$
can be found easily. This occurs for instance if $M$ is a vector space and $D$
is a constant
(i.e.\ translation-invariant) Dirac structure,
as we will show next.
However, this procedure has the drawback that in order to find a solution of (\ref{11drracds}) for a given initial condition, one must first find the leaf $S$ containing that initial condition and then solve \eqref{equation15}. For an initial condition on a different leaf, one has to repeat the constraint algorithm for a different corresponding equation \eqref{equation15}. Because of this, even in these simple cases, working directly with the Dirac structure $D$, using the constraint algorithm to be developed in this section, rather than with the associated presymplectic form on a
presymplectic leaf, is not only possible but also convenient, since this leads to obtaining a single equation on 
a final foliated constraint submanifold, as we will see.

\subparagraph{The case of a constant Dirac structure.}
Let $V$ be a vector space and
$\widetilde{D}\subseteq V\oplus V^\ast$ a given linear Dirac structure.
Then we have the presymplectic form
$\widetilde{\omega}_{\widetilde{D}}$ on $E_{\widetilde{D}}$ and the associated linear map
$\widetilde{\omega}^{\flat}_{\widetilde{D}} \colon  E_{\widetilde{D}} \rightarrow
E_{\widetilde{D}}^*$.
We  consider
the Dirac structure $D \subseteq TV \oplus T^\ast V$ on
the \textit{manifold} $V$ defined as $D_x = (x, \widetilde{D})$, where
we have used the natural identification $TV \oplus T^\ast V \equiv V
\times (V\oplus V^\ast)$.
This Dirac structure is integrable and constant, that is, invariant under
translations in a natural sense, as we will show next.
\medskip

For
each $x \in V$ the presymplectic leaf $S_x$ containing $x$ (in the sense of
Theorem \ref{presymplfoliation}) is $x +
E_{\widetilde{D}} \subseteq V$. For each $x \in V$,
$\widetilde{\omega}_{\widetilde{D}}$ induces the constant presymplectic
form $\omega_{D, S} \in \Omega^2\left(x + E_{\widetilde{D}}\right)$
given by
\[
\omega_{D, S}(x + y)\left((x + y, \bar{X}), (x + y,
\bar{Y})\right) = \widetilde{\omega}_{{\widetilde{D}}}(\bar{X}, \bar{Y}),
\]
where $y \in E_{\widetilde{D}}$,
so $x+ y$ represents any point in the symplectic leaf $x + E_{\widetilde{D}}$,
or, equivalently,
\[
\omega_{D,S}^b(x+y)(x+y,\bar{X})(x+y,\bar{Y})=\widetilde{\omega}_{\widetilde{D}}^{b}
(\bar{X})(\bar{Y}).
\]
Consider the system
\begin{align}\label{equation}
(z, \dot{z})\oplus(z, d\mathcal{E}(z)) \in D.
\end{align}

This system, for a given initial condition $z (0) = x $, is
equivalent to the following equation on the presymplectic leaf $S$ that contains
$x$
\begin{align}\label{equation1}
\omega_{D,S} (z)(\dot{z}, \, \cdot \,) = d\mathcal{E}(z)\vert E_{Dz}.
\end{align}
Keeping $x$ fixed and writing $z = x + y $, the system
\eqref{equation} on the presymplectic
leaf $ S = x + E_{\widetilde{D}}$ becomes
\[
\dot{y} \oplus d\mathcal{E}(x + y) \in \widetilde{D}.
\]
Therefore equation \eqref{equation1} is equivalent to
the following equation on the subspace
$E_{\widetilde{D}}$
\begin{align}\label{equation3}
\widetilde{\omega}_{\widetilde{D}}(\dot{y},\, \cdot \,) = d\mathcal{E}(x + y)\vert
E_{\widetilde{D}},
\end{align}
with initial condition $y (0) = 0$. Equation (\ref{equation3}) can
be solved by the algorithm described in \cite{MR506712} and \cite{MR0535369}, and
sketched in section \ref{The Gotay-Nester Algorithm}. Notice that, since
the presymplectic form $\widetilde{\omega}_{\widetilde{D}}$ on the vector
space $E_{\widetilde{D}}$ determines naturally a translation-invariant
Dirac structure on the same space $E_{\widetilde{D}}$ considered as a manifold,
equation (\ref{equation3}) is also a Dirac dynamical system in the same way
(\ref{equation}) is, but with the Dirac structure given by the
presymplectic form $\widetilde{\omega}_{\widetilde{D}}$ considered as a
constant form on the manifold $x + E_{\widetilde{D}}$. Because of this
and also because $\dim E_{\widetilde{D}}\leq \dim V$, equation (\ref{equation3}) is,
in essence,
simpler than system (\ref{equation}).
\paragraph{The case of a general Dirac manifold.}
Now let $D$ be a  Dirac structure on $M$ that needs not be integrable.
In order to explain our algorithm for Dirac manifolds we need the following
auxiliary result, involving
a given subspace $W_x$ of $E_{Dx}$, which is easy to prove
using results from section \ref{sectiondiracstructures}.
\begin{lemma}\label{lemmaequivalence}
For each $x\in M$ we have the following
equivalent conditions, 
where
$W_x$
is a given subspace of
$E_{Dx}$.
\begin{itemize}
\item[(i)] There exists $(x,\dot{x}) \in W_x$ such that  (\ref{11drracds}) is satisfied.
\item[(ii)] There exists $(x,\dot{x}) \in W_x$ such that $\omega_D(x)(\dot{x}, \,) = d\mathcal{E}(x) | E_{Dx}$.
\item[(iii)] $d\mathcal{E}(x) | E_{Dx}
\in
W_x^{\flat_{\omega_D}}$.
\item[(iv)] $\left\langle d\mathcal{E}(x) | E_{Dx},
W_x^{\omega_D}\right\rangle
=\{0\}$.
\item[(v)] There exists $(x,\dot{x}) \in W_x$ such that $d\mathcal{E}(x)
\in
D^\flat(x, \dot{x})$.
\item[(vi)] $d\mathcal{E}(x)\in
D^\flat(W_x)$.
\item[(vii)] $\left\langle d\mathcal{E}(x), W_x^D\right\rangle
=
\{0\}$.
\end{itemize}
\end{lemma}
Now we can describe the constraint algorithm for Dirac dynamical systems,
called \textit{CAD}. Following the same idea of the Gotay-Nester algorithm
described in section \ref{The Gotay-Nester Algorithm} we should construct
a sequence of constraint submanifolds.

To define the first constraint
submanifold $M_1$ we may use either one of the equivalent conditions
of Lemma \ref{lemmaequivalence}, with
$W_x \equiv W_{0x}= E_{Dx}$.
We want
to emphasize the role of the two equivalent
conditions (iv) and (vii), as they represent a formal analogy between the CAD
and the Gotay-Nester
algorithm. Of course, the Gotay-Nester algorithm, by definition, corresponds to
the case
$E_D = TM$.

Define
\begin{equation}\label{constraintsubmf1}
\begin{split}
    M_1 & =  \{x \in M \mid \left\langle d\mathcal{E}(x) | E_{Dx},
    E^{\omega_D}_{Dx}\right\rangle = \{0\}\}\\
 & =  \{x \in M \mid \left\langle d\mathcal{E}(x),
E^D_{Dx}\right\rangle = \{0\}\}.
\end{split}
\end{equation}
Let us assume that $M_1$ is a submanifold. Then we define the
second constraint submanifold $M_2 \subseteq M_1$ by either of the
following equivalent conditions, in agreement with
(iv) and (vii) of Lemma \ref{lemmaequivalence}, with $W_x \equiv W_{1x} = E_{Dx} \cap
T_xM_1$,
\begin{equation*}
\begin{split}
M_2 & =  \{x \in M_1 \mid \left\langle d\mathcal{E}(x)| E_{Dx},
    \left(W_{1x}\right)^{\omega_D}\right\rangle =\{0\}\}\\
& =  \{x \in M_1 \mid \left\langle d\mathcal{E}(x), \left(W_{1x}\right)^D\right\rangle
=\{0\}\}.
\end{split}
\end{equation*}
More generally we define recursively $M_{k+1}$ for $k=1,2,\ldots$,
by either of the conditions
\begin{equation}\label{constraintsubmf3}
\begin{split}
M_{k+1} & =  \{x \in M _k  \mid \left\langle d\mathcal{E}(x) | E_{Dx},
\left(W_{kx}\right)^{\omega_D}\right\rangle =\{0\}\}\\
& =  \{x \in M_k \mid \left\langle d\mathcal{E}(x), \left(W_{kx}\right)^D\right\rangle
=\{0\}\}
\end{split}
\end{equation}
with
$W_{kx} = E_{Dx} \cap
T_xM_k$.
The algorithm stops and the final constraint submanifold $M_c$ is
determined by the condition
$M_{c+1}=M_c$. The solutions $(x, \dot{x})$ are in $W_{cx}=E_{Dx} \cap
T_x{M_c}$.

We have proven that solution curves of  (\ref{11drracds}) are exactly solution curves of
\begin{align}\label{gotaygen1ccc}
(x,\dot{x})\oplus d\mathcal{E} (x) 
&\in 
D_x\\\label{gotaygen1ccc2}
(x, \dot{x})
\in
W_{cx}.
\end{align}

\begin{remark}
In some examples, it is sometimes easier to write the constraint submanifolds using condition (vi) in Lemma \ref{lemmaequivalence}, that is,
\begin{equation*}
M_{k+1} =  \{x \in M _k  \mid d\mathcal{E}(x)\in D^\flat\left(W_{kx}\right)\}
\end{equation*}
with
$W_{kx} = E_{Dx} \cap
T_xM_k$, $k=0,1, \dots$, where $M_0=M$.  
\end{remark}

\begin{remark}\label{ImportantRemark}
 Formula (\ref{constraintsubmf3}) has a special
meaning in the case of an integrable Dirac structure. In fact, let $S$ be an
integral leaf of the distribution $E_D$; then by applying the Gotay-Nester
algorithm, encoded in the recursion formula
(\ref{equationalgorithGN}), to the system
(\ref{equation15}) one obtains a sequence of secondary constraints 
$S = S_0 \supseteq S_1 \supseteq \dots  \supseteq S_c$ given by the recursion formula
\begin{equation}\label{SequationalgorithGN}
S_{k+1} = \{x \in S_k  \mid  \langle i_S^\ast d\mathcal{E}(x), \left(T_x S_k
\right)^{\omega_{D,S}}\rangle = \{0\}\}.
\end{equation}
But it is clear that equation (\ref{SequationalgorithGN}) coincides with
equation (\ref{constraintsubmf3}) since
$T_x S \equiv E_{Dx}$, $T_x S_k \equiv W_{kx}$ and  the presymplectic form 
$\omega_{D,S}$ on $S$ is defined by
$\omega_{D,S}(x)\equiv \omega_D(x)$ for all $x \in S$.
We are assuming regularity conditions that ensure that the Gotay-Nester
algorithm applied for each $S$ stops after a number of steps $c$ which does not
depend on $S$ and which is at the same time the number of steps after which the
CAD stops. 
As a conclusion, the final constraint submanifold $M_c$ of the CAD is
foliated by leaves
$S_c$, where $S$ varies on the set of integral leaves of the distribution $E_D$. We may
say that in the case of an integrable Dirac structure $D$ the CAD is
equivalent to a collection of Gotay-Nester algorithms, one for each leaf $S$ of the
distribution $E_D$. The final equation given by (\ref{gotaygen1ccc})--(\ref{gotaygen1ccc2}) becomes
 \begin{align*}%\label{diracdiracintegrable}
(x, \dot x)\oplus d\mathcal{E}(x)&\in D_x\\
(x, \dot x)&\in T_xS_c,
\end{align*}
which is equivalent to the collection of equations
\begin{align*}%\label{equation15cc}
\omega_{D,S} (x)(\dot{x}, \, \cdot \,) 
&=
i^\ast_S d\mathcal{E}(x)\\
(x, \dot{x})
&\in
T_x S_c.
\end{align*}

In section \ref{sectionanextensionofetc} we will extend the Dirac theory of constraints. For that purpose, we will use an embedding of $M$ in a symplectic manifold $P$ such that the presymplectic leaves $S$ of
are presymplectic submanifolds of $P$. The submanifold $M$ plays the role of a primary \textit{foliated} constraint submanifold. The case in which there is only one leaf gives the Dirac theory.
\end{remark}

\paragraph{Solving the equation.}
The constraint algorithm CAD gives a method to solve the IDE  (\ref{11drracds}) which
generalizes the Gotay-Nester method. Assume that the final constraint
submanifold $M_c$
has been determined, and consider, for each $x \in M_c$, the affine space
\begin{equation*}
S^{(c)}_x : = \{(x, \dot{x}) \in T_xM_c \mid (\mbox{\ref{classicalm122}})
\mbox{ is satisfied}\},
\end{equation*}
which is nonempty if
$M_c$
is nonempty, a condition that will be assumed from now on.
Let
$d^{(c)}(x)$
be the dimension of
$S^{(c)}_x$.

The following theorem is one of the ingredients of our main results, and generalizes the Gotay-Nester algorithm for the case of Dirac dynamical systems \eqref{againdirac} rather than Gotay-Nester systems (\ref{11GN}). Its proof is not difficult,
using the previous lemma, and is left to the reader.

\begin{theorem}\label{maintheorem1}
Let
$M$ be a given manifold, $D$ a given Dirac structure on
$M$ and
$\mathcal{E}$ a given energy function on $M$,
and consider the Dirac dynamical system
\begin{equation}\label{againdirac}
(x,\dot{x})\oplus d\mathcal{E} (x) \in D_x.
\end{equation}
Assume that for each
$k = 1,\dots$,
the subset
$M_k$ of $M$
defined recursively by the formulas
(\ref{constraintsubmf1})--(\ref{constraintsubmf3})
is a submanifold, called the \textbf{\emph{$k$-constraint submanifold}}. 
The decreasing sequence $M_k$ stops, say
$M_c = M_{c + 1}$ (which implies $M_c = M_{c + p}$, for all $p \in
\mathbb{N}$), and call $M_c$ the \textbf{\emph{final constraint submanifold}}.
Then the following hold:

\textit{(a)}
For each $x \in M_c$, there exists
$(x, \dot{x}) \in W_{cx}=E_{Dx} \cap T_x{M_c}$ such that
(\ref{againdirac}) is satisfied.
The Dirac dynamical system (\ref{againdirac}) is equivalent to the equation
\begin{equation}\label{redagaindirac}
\omega_{D}(x)(\dot{x}\, , \,) = d\mathcal{E} (x)| E_{Dx}, \,\,\,(x, \dot{x}) \in
W_{cx},
\end{equation}
that is, both equations have the same solution curves $x(t) \in M_c$.

\textit{(b)}
For each
$x \in M_c$,
$d^{(c)}(x)$
equals the dimension of
$\operatorname{ker} \omega_D(x) \cap T_x M_c$.

\textit{(c)}
If
$d^{(c)}(x)$
is a locally constant function of $x$
on
$M_c$
then
$S^{(c)} = \bigcup_{x \in M_c} S^{(c)}_x$
is an affine bundle with base
$M_c$.
Each section
$X$
of $S^{(c)}$ is a vector field on $M_c$ having the property that
$X(x) \oplus d\mathcal{E} (x) \in D_x$,
for all
$x \in M_c$.
Solution curves to such vector fields are solutions to
the Dirac dynamical system  (\ref{11drracds}). More generally, one can choose
arbitrarily a time-dependent section $X_t$, 
then solution curves of $X_t$ will be also solutions to  (\ref{11drracds})
and those are the only solutions of 
 (\ref{11drracds}).
Solution curves to  (\ref{11drracds}) are unique for any given initial condition
if and only if
$d^{(c)}(x) = 0$, for all
$x \in M_c$.
\end{theorem}

\paragraph{Local representation for $S^{(c)}$.} One can find a local
representation for $S^{(c)}$ by just choosing a local parametrization of $M_c$.
Let
$x = x(z_1,\dots,z_r) \in U$, where $U$ is an open set and
$r$ is the dimension of $M_c$, be such a local parametrization.
Then substitute this expression for $x$ in
 (\ref{11drracds}) to obtain an IDE in $z = (z_1,\dots,z_r)$,
namely
\begin{align}\label{gotaygen122}
\left(x(z),D_{z}x(z) \cdot \dot{z}\right)
\oplus d\mathcal{E} \left(x(z)\right) \in D_{x(z)}.
\end{align}
The local representation of
$S^{(c)}$ is given by the trivial affine bundle
\[
\{(z, \dot{z}) \mid (z, \dot{z}) \mbox{ satisfies (\ref{gotaygen122})}\}.
\]

\paragraph{Some results concerning uniqueness of solution.} Under Assumption \ref{K2} (b), we can prove the following lemma.

\begin{lemma}\label{lemmalemma}
(A)
Existence and uniqueness of a solution curve
$x(t)$
of
(\ref{redagaindirac})
for any initial condition
$x(0) \in M_c$,
and therefore also of the Dirac dynamical system
(\ref{againdirac}),
is equivalent to any of the conditions

(i) $\operatorname{ker}\omega_D (x) \cap W_{cx}
=
\operatorname{ker} \omega_D(x) \cap T_x M_c
=
\{0\}$,
for each
$x \in M_c$,

(ii) $W_{cx}\cap E^{\omega_D}_{Dx} =\{0\}$, for each $x \in M_c$,

(iii) $W_{cx}\cap E^{\Omega}_{Dx} =\{0\}$, for each $x \in M_c$.
\\

(B) If
$\omega_D | W_{cx}$,
or,
equivalently,
$\Omega | W_{cx}$,
is symplectic for each
$x \in M_c$
then there is existence and uniqueness of solution
$x(t)$
of
(\ref{redagaindirac})
for any initial condition
$x(0) \in M_c$,
and therefore also of the Dirac dynamical system
(\ref{againdirac}).
One also has that
$\omega_D | W_{cx}$
is symplectic for each
$x \in M_c$
if and only if any one of the following
conditions is satisfied:

(i) $W_{cx}\cap W_{cx}^{\omega_D} =\{0\}$,

(ii) $W_{cx}\cap W_{cx}^{\Omega} =\{0\}$.
\end{lemma}
\begin{proof}
We first recall the argument from Theorem
\ref{maintheorem1}. Uniqueness of a solution $x(t)$ of equation
(\ref{redagaindirac}) for any given initial condition $x_0 \in
M_c$ (which, as we know from the CAD, must satisfy $x(t) \in
M_c$ and $\dot{x}(t) \in W_{cx(t)}$ for all $t$) holds if and only
if for each $x \in M_c$ there is a uniquely determined $v_x \in
W_{cx}$ such that
\begin{equation*}%\label{equationequation11}
 \omega_D(x) (v_x, \,\,) = d\mathcal{E} | E_{Dx},
\end{equation*}
in other words, if equation
(\ref{redagaindirac})
defines a vector field on
$M_c$.
In fact, if this is the case, by the general theory of
ODE on manifolds we have existence and uniqueness of solution.
We also recall that under Assumption \ref{K2} (b), the equation
(\ref{redagaindirac})
defines a family of vector fields on
$M_c$,
which defines an affine distribution of constant rank
whose space at the point
$x \in M_c$
is the affine space of all solutions
$v_x$
as indicated above.
Using this we can deduce that
uniqueness of solution
$x(t)$
of
(\ref{redagaindirac})
for any initial condition
$x(0) \in M_c$,
is equivalent to the affine distribution defined above having dimension
$0$.

Now we shall prove the equivalence between uniqueness of solution and
\textit{(A)}.
Let
$v_x \in W_{cx}$
satisfying equation
(\ref{redagaindirac}), that is,
\begin{equation*}%\label{equationequation1}
 \omega_D(x) (v_x, \,\,) = d\mathcal{E} | E_{Dx}.
\end{equation*}
Let
$w_x \neq 0$
be such that
$w_x \in \operatorname{ker}\omega_D (x) \cap W_{cx}$.
Then
$v_x + w_x$
also satisfies (\ref{redagaindirac}), which shows that the affine
distribution
described above has dimension greater than $0$.
Using this, the proof of the equivalence between uniqueness of solution and
\textit{(A)(i)} follows easily.
The rest of the proof of the equivalence with \textit{(A)} follows from the fact
that
$\operatorname{ker}\omega_D (x)
=
E^{\omega_D}_{Dx}
=
E^{\Omega}_{Dx} \cap E_{Dx}$
for each $x \in M_c$,
which can be proved directly using the definitions.

Note that the equivalence between \textit{(i), (ii), (iii)}
of
\textit{(A)}
holds for \textit{any} subspace
$W_{cx} \subseteq E_{Dx}$.

Now we shall prove
\textit{(B)}.
First of all, if
$v_x \in W_{cx}$
satisfies
(\ref{redagaindirac}) then,
since
$W_{cx} \subseteq E_{Dx}$,
 it clearly also satisfies
\begin{equation}\label{equationequation111}
(\omega_D(x)| W_{cx})(v_x, \,\,) = d\mathcal{E} | W_{cx}.
\end{equation}
We can conclude that if $\omega_D(x)| W_{cx}$,  which as we know
coincides with $\Omega(x) | W_{cx}$, is symplectic then equation
(\ref{redagaindirac}) defines a vector field $v_x$ on $M_c$,
which is, in fact, given by equation (\ref{equationequation111}).
Then existence and uniqueness of solution $x(t)$ of
(\ref{redagaindirac}) for any initial condition $x(0) \in M_c$ is
guaranteed.
Finally, the equivalence between symplecticity of
$\omega_D(x)| W_{cx} \equiv \Omega(x) | W_{cx}$ and \textit{(B)(i)}
and \textit{(B)(ii)} is easy to prove using basic linear symplectic
geometry.
In fact, is is easy to prove that symplecticity of
$\omega_D(x)| W_{cx}$
is equivalent to
$W^{\omega_D}_{cx} = \{0\}$.
Using Lemma \ref{lemmaperp} we can deduce that $W^{\omega_D}_{cx} = W_{cx} \cap
W^{\Omega}_{cx}$ from which we obtain $W_{cx} \cap W^{\Omega}_{cx} =
W_{cx} \cap W^{\omega_D}_{cx}$. 
\end{proof}

With the method just described, one can deal with many examples of
interest, such as nonholonomic systems and circuits, provided that
one chooses the manifold $M$ and the Dirac structure $D$ properly.
We show in the next section how a
nonholonomic system given by a distribution $\Delta \subseteq TQ$
on the configuration space $Q$ can be described by a Dirac dynamical system
on the Pontryagin bundle $M = TQ \oplus T^\ast Q$ and how one can
apply the constraint algorithm CAD to this example, although we will not
perform a detailed calculation of the sequence of constraint
submanifolds $M_k$. We also show how LC circuits can be treated
with the same formalism as nonholonomic systems. The main point
for doing this is, again, to choose the manifold $M$ as being the
Pontryagin bundle $TQ \oplus T^\ast Q$, where, this time, $Q$ is
the charge space, and a canonically constructed Dirac structure
$\bar{D}_{\Delta}$ on $M$, where, this time, $\Delta$ represents
Kirchhoff's Current Law. We also show how this approach using
$TQ \oplus T^*Q$ is related to the approach used in
\cite{MR2265464,MR2265469}.
\\

\section{Examples}\label{sectionexamples}
In this section we deal with two examples, namely, nonholonomic systems, and LC
circuits, showing that Dirac dynamical systems give a unified treatment for
them.
We will perform the detailed calculation of the constraint submanifolds
$M_k$
for the case of LC circuits only.
\paragraph{Nonholonomic systems.}
Recall that a nonholonomic system is given by a configuration space
$Q$, a distribution $\Delta \subseteq TQ$, called the nonholonomic
constraint, and a Lagrangian $\mathcal{L} \colon  TQ \rightarrow
\mathbb{R}$. Equations of motion are given by Lagrange-d'Alembert's
principle.

 Inspired by the Hamilton-Poincar\'{e} principle given in
\cite{%
Cendra_Marsden_Pekarsky_Ratiu:Variational_principles--Hamilton-Poincare_equations}, we can write a convenient equivalent form of
the Lagrange-d'Alembert principle as
\[
\delta \int^{t_1}_{t_0}(p\dot{q} - \mathcal{E} (q,v,p))dt = 0,
\]
where $\mathcal{E}  \colon  TQ\oplus T^\ast Q \rightarrow \mathbb{R}$ is
defined by $\mathcal{E} (q,v,p) = pv - \mathcal{L}(q,v)$, and with
the restriction on variations $\delta q \in \Delta$, $\delta q (t_i)
= 0$ for $i = 0,1$, along with the kinematic restriction $v \in
\Delta$. The resulting equations are
\begin{align}\label{nonholoeq}
\dot{p} - \frac{\partial \mathcal{L}}{\partial q} &\in \Delta^\circ
\\
\label{nonholoeq1}
\dot{q}
&=
v\\
\label{nonholoeq2} p - \frac{\partial \mathcal{L}}{\partial v} &=
0\\
\label{nonholoeq3}
v
&\in \Delta.
\end{align}
We are going to show that equations (\ref{nonholoeq})--(\ref{nonholoeq3}) can be
written in the form  (\ref{11drracds}). For this
purpose we must construct an appropriate Dirac structure
associated to the nonholonomic constraint. Inspired by several
results in \cite{MR2265464} and by the
Hamilton-Poincar\'{e}'s point of view we define a Dirac structure
$\bar{D}_{\Delta} \subseteq TM\oplus T^\ast M$ on $M = TQ\oplus
T^\ast Q$ associated to a given distribution $\Delta \subseteq TQ$
on a manifold $Q$ by the local expression
\[
\bar{D}_{\Delta}(q,v,p)=\{(q,v,p,\dot{q},\dot{v},\dot{p},\alpha,\gamma,
\beta)\mid \dot{q} \in \Delta(q),\,\alpha+\dot{p} \in
\Delta^{\circ}(q),\,\beta=\dot{q},\,\gamma=0\}.
\]
\textbf{Note.} We shall accept both equivalent notations
\[
(q,v,p,\dot{q},\dot{v},\dot{p},\alpha,\gamma,
\beta)
\equiv
(q,v,p,\dot{q},\dot{v},\dot{p})\oplus(q,v,p,\alpha,\gamma,\beta),
\]
for an element of
$TM\oplus T^\ast M$.

By checking that $\dim \bar{D}_{\Delta}(q,v,p)=3\dim Q$,
that is $\dim \bar{D}_{\Delta}(q,v,p)=\dim M$,
and that
\[
\left\langle \! \left\langle
(q,v,p,\dot{q}_1,\dot{v}_1,\dot{p}_1)\oplus(q,v,p,\alpha_1,\gamma_1,\beta_1),
(q,v,p,\dot{q}_2,\dot{v}_2,\dot{p}_2)\oplus(q,v,p,\alpha_2,\gamma_2,\beta_2)
\right\rangle \! \right\rangle =0
\]
for all
\begin{align*}
& (q,v,p,\dot{q}_1,\dot{v}_1,\dot{p}_1)\oplus (q,v,p,\alpha_1,\gamma_1,\beta_1)
\in \bar{D}_{\Delta}(q,v,p), \\
& (q,v,p,\dot{q}_2,\dot{v}_2,\dot{p}_2)\oplus (q,v,p,\alpha_2,\gamma_2,\beta_2)
\in \bar{D}_{\Delta}(q,v,p),
\end{align*}
we can conclude using Lemma \ref{defdiracstr} that
$\bar{D}_{\Delta}$ is a Dirac structure on $M$. We should now
prove that $\bar{D}_{\Delta}$ is well defined globally, in other
words, that it does not depend on the choice of a local chart. Let
$\bar{\tau} \colon  TQ\oplus T^\ast Q \rightarrow Q$ and $\bar{\pi} \colon 
TQ\oplus T^\ast Q \rightarrow T^\ast Q$ be the natural maps that
in local coordinates are given by $\bar{\tau} (q,v,p) = q$ and
$\bar{\pi} (q,v,p) = (q,p)$.
For a given distribution $\Delta \subseteq TQ$ consider the
distribution $\bar{\Delta} = (T\bar{\tau})^{-1}(\Delta)$ and also
the 2-form $\bar{\omega} = \bar{\pi}^\ast \omega$, on the manifold
$TQ\oplus T^\ast Q$, where $\omega$ is the canonical 2-form on
$T^\ast Q$. We have the local expressions $\bar{\Delta} =
\{(q,v,p, \dot{q}, \dot{v}, \dot{p}) \,|\, \dot{q} \in \Delta\}$ and
$\bar{\omega}(q,v,p) = dq \wedge dp$. Now we can apply Theorem
\ref{diracpresymplecticmanifold} replacing $M$ by $TQ\oplus T^\ast
Q$, $E$ by $\bar{\Delta}$ and $\omega$ by $\bar{\omega}$ and then
we can easily check that the Dirac structure $D_{\omega_E}$
coincides with $\bar{D}_{\Delta}$. In other words, by using
Theorem \ref{diracpresymplecticmanifold} we have obtained a
coordinate-independent description of $\bar{D}_{\Delta}$ given in terms of
$\bar{\tau}$,
$\bar{\pi}$ and $\bar{\omega}$,
which proves in particular that it is well defined globally.
\medskip

It is straightforward to check that the condition
\begin{align}\label{nonholodirac}
(x, \dot{x})\oplus d\mathcal{E} (x) \in \bar{D}_{\Delta},
\end{align}
where $x = (q,v,p)$,
is equivalent to
\begin{align}\label{nnhhooll1}
\dot{p} - \frac{\partial \mathcal{L}}{\partial q} &\in
\Delta^\circ\\
\label{nnhhooll12}
\dot{q}
&=
v\\
\label{nnhhooll3}
p
&=
\frac{\partial \mathcal{L}}{\partial v}\\
\label{nnhhooll4} \dot{q} &\in \Delta,
\end{align}
which is clearly equivalent to equations
(\ref{nonholoeq})--(\ref{nonholoeq3}).
\medskip

Now that we have written the equations  of motion as a Dirac dynamical system,
we can proceed to apply the CAD\@. First, we
can easily prove the following formulas using the definitions and
Proposition \ref{circperpD}:
\begin{align}
E_{\bar{D}_{\Delta}}
&=
\{(q,v,p, \dot{q}, \dot{v}, \dot{p}) \mid  \dot{q} \in \Delta\}, \label{E_D1}\\
E_{\bar{D}_{\Delta}}^\flat
&=
\{(q,v,p, \alpha, \gamma, \beta) \mid  \gamma = 0, \beta \in \Delta\},
\label{E_D2}\\
E^{\bar{D}_{\Delta}}_{\bar{D}_{\Delta}}
&= \{(q,v,p,\dot{q},
\dot{v}, \dot{p}) \mid  \dot{q} = 0, \dot{p} \in \Delta^\circ\}.\label{E_D3}
\end{align}
Then we have
\begin{equation}\label{constraintnonholo1}
\begin{split}
M_1
&=
\left\{ (q,v,p) \; \left|  \; \left\langle
d \mathcal{E} (q,v,p), E_{\bar{D}_{\Delta}}^{\bar{D}_{\Delta}}(q,v,p)
\right\rangle  = \{0\} \right\} \right. \\
&= \left\{ (q,v,p) \; \left|  \; 
d \mathcal{E} (q,v,p) \in E_{\bar{D}_{\Delta}}^\flat(q,v,p)
 \right\} \right. \\
&= \left\{ (q,v,p) \; \left| \; p - \frac{\partial L}{\partial v} = 0,\,\, v \in
\Delta \right\} \right..
\end{split}
\end{equation}

We could continue applying the algorithm as explained in general in
section \ref{sectionconstalgforrdirac}, and we would obtain
specific formulas for $M_2, M_3$, etc.

\subparagraph{Relationship with implicit Lagrangian systems.} In
\cite{MR2265464} the
description of a nonholonomic system as an
\textit{implicit Lagrangian system} was introduced and equations of motion were
shown to be a \textit{partial vector field} written in terms of
the \emph{Dirac differential}. There is a close and simple relationship between
this approach and
the one of the present paper, which we shall explain next.
\medskip

First, we should recall the notion of an implicit Lagrangian system.
Let $\mathcal{L}\colon TQ \rightarrow \mathbb{R}$ be a given Lagrangian;
then $d\mathcal{L}$ is a $1$-form on $TQ$, $d\mathcal{L}\colon TQ
\rightarrow T^*TQ$, which is locally expressed by
\[
d\mathcal{L}=\left(q,v,\frac{\partial{\mathcal{L}}}{\partial{q}},\frac{\partial{\mathcal{L}}}{\partial{v}}\right).
\]
Define a differential operator $\mathfrak{D}$, called the {\it Dirac
differential}\/ of $\mathcal{L}$, by
\[\mathfrak{D}\mathcal{L}=\gamma_Q \circ d\mathcal{L},\]
where $\gamma_Q$ is the diffeomorphism defined by
\[\gamma_Q=\Omega^\flat \circ (\kappa_Q)^{-1}\colon T^*TQ  \rightarrow
T^*T^*Q.\]
Here
$\Omega$
is the canonical $2$-form on $T^\ast Q$ and
$\kappa_Q \colon  TT^\ast Q \rightarrow T^\ast TQ$ is the canonical isomorphism which
is
given in a local chart by
$\kappa_Q (q,p,\delta q, \delta p) = (q, \delta q, \delta p, p)$.
We have the local expression
\[
\mathfrak{D}\mathcal{L} = \left(q,
\frac{\partial{\mathcal{L}}}{\partial{v}},
-\frac{\partial{\mathcal{L}}}{\partial{q}}, v\right).
\]

We must now recall the definition and properties of the Dirac structure
$D_{\Delta}$ on the manifold $T^\ast Q$,
where
$\Delta$ is a given distribution on
$Q$,
studied in propositions 5.1 and 5.2 of
\cite{MR2265464}.
Let $\pi \colon  T^\ast Q \rightarrow Q$ be the canonical projection. Then $\Delta_{T^\ast Q}$ is defined by
\[
\Delta_{T^\ast Q} = \left(T\pi \right)^{-1}(\Delta).
\]
The Dirac structure $D_{\Delta}$ is the one given by Theorem
\ref{diracpresymplecticmanifold}, with $M = T^\ast Q$, $\omega = \Omega$ the
canonical symplectic form on $T^\ast Q$ and $E = \Delta_{T^\ast Q}$, and it is described by
\begin{align}\label{eqmahi}
\nonumber
 D_{\Delta}(q,p) = \{v\oplus \alpha \in T_{(q,p)}T^\ast Q
\oplus T_{(q,p)}^\ast T^\ast Q \mid
v \in \Delta_{T^\ast Q}(q,p) \,\, \\
\mbox{and}\,\, \alpha (w) = \Omega (q,p) (v, w)\,\, \mbox{for all}\,\, w \in
\Delta_{T^\ast Q}(q,p)\}.
\end{align}
Now we will define the notion of an implicit Lagrangian system.
\begin{definition} Let $\mathcal{L}\colon TQ
\rightarrow \mathbb{R}$ be a given Lagrangian (possibly degenerate)
and let $\Delta \subseteq TQ$ be a given regular {\it constraint
distribution} on a configuration manifold $Q$. Denote by
$D_{\Delta}$ the induced Dirac structure on $T^*Q$ that is given by
the equation (\ref{eqmahi}) and write $\mathfrak{D}\colon TQ \rightarrow
T^*T^*Q$ for the Dirac differential of $\mathcal{L}$. Let
$\mathbb{F}L(\Delta) \subseteq T^*Q$ be the image of $\Delta$ under
the Legendre transformation. An {\it implicit Lagrangian system} is
a triple $(\mathcal{L}, \Delta, X)$ where $X$ represents a vector
field defined at points of $\mathbb{F}L(\Delta) \subseteq T^*Q$,
together with the condition
\[
X \oplus \mathfrak{D}\mathcal{L} \in D_{\Delta}.
\]
In other words, for each point $(q,v) \in \Delta$ let
$(q,p)=\mathbb{F}L(q,v)$
and then, by definition,
$X(q,v, p) \in T _{( q, p ) } \left( T^{\ast} Q \right)$ must satisfy
\begin{align}\label{diracdiffeq}
X(q,v, p)\oplus\mathfrak{D}\mathcal{L}(q,v) \in D_{\Delta}(q,p).
\end{align}
\end{definition}
A {\it solution curve} of an implicit Lagrangian system
$(\mathcal{L},\Delta,X)$ is a curve $(q(t),v(t)) \in \Delta$, $t_1
\leq t \leq t_2$, which is an integral curve of $X$ where
$(q(t),p(t))=\mathbb{F}L(q(t),v(t))$.
The following proposition is essentially Proposition 6.3 of
\cite{MR2265464}.
\begin{proposition}\label{IDEMAHI}
The condition $X \oplus \mathfrak{D}\mathcal{L} \in D_{\Delta}$
defining an implicit Lagrangian system is given locally by the
equalities
\begin{align}\label{nonholomahi}
p
&=
\frac{\partial \mathcal{L}}{\partial v}\\
\label{nonholomahiA}
\dot{q}
&\in
\Delta\\
\label{nonholomahiB}
v
&=
\dot{q}\\
\label{nonholomahiC} \dot{p} - \frac{\partial \mathcal{L}}{\partial
q} &\in \Delta^\circ .
\end{align}
\end{proposition}
It is clear that equations (\ref{nonholomahi})--(\ref{nonholomahiC}) are
equivalent to
(\ref{nonholoeq})--(\ref{nonholoeq3}) and also to
(\ref{nnhhooll1})--(\ref{nnhhooll4}). This leads immediately to a precise link
between the
approach  to nonholonomic systems given in \cite{MR2265464} and the one in the
present paper, which is given in the
next proposition whose proof in local coordinates is easy and will be omitted.
Recall that $M_1$ is the first constraint manifold given by
(\ref{constraintnonholo1}).
\begin{proposition}\label{link}
Let $\mathcal{L} \colon  TQ \rightarrow \mathbb{R}$ be a given Lagrangian,
$\Delta \subseteq TQ$ a given distribution and let $\mathcal{E}  \colon 
TQ \oplus T^\ast Q \rightarrow \mathbb{R}$ be given by $\mathcal{E}
(q,v,p) = pv - \mathcal{L}(q,v)$. Then the following assertions
hold.
\begin{enumerate}
\item[{\rm (i)}]
 $\mathbb{F}L(\Delta) = \bar{\pi}(M_1)$.
\item[{\rm (ii)}]
Let $(q,v,p,\dot{q}, \dot{v}, \dot{p}) = \bar{X}(q,v,p)$
be a solution to (\ref{nonholodirac}); then in particular
one has
$(q,v,p) \in M_1$,
and let
$X(q,v, p) = (q, p, \dot{q}, \dot{p})$, that is, $X(q,v,p) = T \bar{\pi}
\bar{X}(q,v,p)$.
Then
$X(q,v, p)$ is a solution to (\ref{diracdiffeq}).
\item[{\rm (iii)}] A curve $\left(q(t),v(t),p(t)\right)$ is a solution curve of
(\ref{nonholodirac}) if and only if
$\left(q(t),v(t)\right)$ is a solution curve of the implicit Lagrangian system
given by (\ref{diracdiffeq}).
\end{enumerate}
\end{proposition}
\paragraph{LC circuits.}
This is a case of a constant Dirac structure of the type explained
in
section \ref{sectionconstalgforrdirac}.
Our approach is closely related and equivalent to the one described in
\cite{MR2265464}.
There are many relevant references where the structure of Kirchhoff's laws has
been studied from different points of view; some of them emphasize the geometry
behind the equations,
see \cite{MR0322904,
MR0414236,
MR1324157,
MR2115691}.

We first briefly recall the description of LC circuits given in
\cite{MR2265464}.
Let $E$ be a vector space representing \textit{the charge space}; then $TE$ is
the
\textit{current space}  and $V = T^\ast E$ is the \textit{flux linkage space.}
There
is a constant distribution $\Delta \subset
TE$ (that is, $\Delta$ is invariant under translations), which represents the
\textit{Kirchhoff's Current Law} (KCL). The subbundle
$\Delta^\circ \subseteq T^\ast E$ represents the \textit{Kirchhoff's Voltage
Law}
(KVL). One has a Dirac structure on the cotangent
bundle $T^*E$,  $D_{\Delta} \subseteq TV \oplus T^{*}V$ given by
\[
D_{\Delta}=\{(q,p,\dot{q},\dot{p})\oplus (q,p,\alpha_q,\alpha_p) \in TV
\oplus T^{*}V\mid \dot{q} \in
\Delta,\,\alpha_p=\dot{q},\,\alpha_q+\dot{p} \in \Delta^{\circ}\},
\]
 which clearly does not depend on $(q,p)$, in other words it is a constant Dirac
structure,
 so for each base point
$(q,p)$ we have
\[
D_{\Delta}(q,p)=(q,p,\widetilde{D}),
\]
where
\[
\widetilde{D}=\{(\dot{q},\dot{p},\alpha,\beta) \in V \oplus V^{*}\mid \dot{q} \in
\Delta,\,\beta=\dot{q},\,\alpha+\dot{p} \in \Delta^{o}\}.
\]
The dynamics of the system is given by a Lagrangian $\mathcal{L} \colon 
TE \rightarrow \mathbb{R}$. This Lagrangian is given by a quadratic
form on $E\times E$ representing the difference of the energies in
the inductors and the capacitors, say,
\[
\mathcal{L}(q,v) = \frac{1}{2}\sum_{i = 1}^{n} L_i v_i^2 -
\frac{1}{2}\sum_{i = 1}^{n}\frac{1}{C_i} q_i^2,
\]
where $n$ is the number of branches of the circuit. Of course, some of the terms in the previous sum may be zero, corresponding to
the absence of
an inductor or a capacitor in the corresponding branch.\\

One approach would be to use the time evolution of the circuit in terms of the Dirac differential, given by equation
(\ref{diracdiffeq}). This shows that Dirac structures provide a
unified treatment for nonholonomic systems and LC circuits. In
other words the equation is
\begin{align*}%\label{cequation}
(\mathbb{F}\mathcal{L}(q,v),\dot{q},\dot{p})\oplus
(q,p,\mathfrak{D}{\mathcal{L}}(q,v)) \in D_{\Delta}(q,p),
\end{align*}
with $(q,p)=\mathbb{F}\mathcal{L}(q,v)$
and $(q,v) \in \Delta$, where  $\mathfrak{D}{\mathcal{L}}$
represents the Dirac differential of $\mathcal{L}$.\\

However, we will follow the philosophy of this paper and choose to work on $TE\oplus T^*E$ rather than working with the Dirac differential.

\subparagraph{Applying the CAD to LC circuits.}
It should be now clear how we can apply the constraint algorithm CAD developed
in the present paper to deal with LC circuits exactly as we did
with the case of nonholonomic systems.
More precisely, it should now
be clear that an LC circuit can be described by the methods
described in the first part of the paragraph
\textit{Nonholonomic Systems} (beginning of section \ref{sectionexamples}), by taking $Q = E$, $M = TE
\oplus T^\ast E$ and defining $\bar{D}_{\Delta}$ and $\mathcal{E} $ as
indicated in that paragraph.
As we have already said, in the case of circuits the Dirac structure
$\bar{D}_{\Delta}$ is constant and therefore integrable.
However, we prefer
not to work with the system restricted to a presymplectic leaf, as explained
before,
but to apply directly the algorithm for a general Dirac structure, as we will
explain next.\\

Define the linear maps
$
\varphi\colon E \rightarrow E^*
$
and
$
\psi\colon E
\rightarrow E^*
$
by
\begin{align*}
\varphi(v)&=
\frac{\partial{\mathcal{L}}}{\partial{v}}=\left(L_1 v_1,\dots,L_n v_n\right),\\
\psi(q)&= \frac{\partial{\mathcal{L}}}{\partial{q}}=-\left(q_1/C_1,\dots,q_n/C_n\right).
\end{align*}
If there is no capacitor in a circuit branch, the corresponding component of $\psi(q)$ will be zero (infinite capacitance), while zero capacitance is ruled out since it would represent an electrically open circuit branch. The physically relevant case corresponds to $0\leq L_i<\infty$, $0<C_i\leq\infty$, for $i=1, \dots, n$. However, the constraint algorithm applies to the general case where capacitances and inductances can be negative. 
\medskip

The evolution equations (\ref{nnhhooll1})--(\ref{nnhhooll4}) for a general
nonholonomic
system become
\begin{align}\label{Cnnhhooll1}
\dot{p} - \psi(q) &\in
\Delta^\circ\\
\label{Cnnhhooll12}
\dot{q}
&=
v\\
\label{Cnnhhooll3}
p
&=
\varphi(v)\\
\label{Cnnhhooll4}
\dot{q}
&\in
\Delta.
\end{align}
Now we will apply
the CAD\@. First, note that
\[
d\mathcal{E}(q,v,p)=\psi(q) dq + (p-\varphi(v)) dv + v dp.
\]

The first constraint
submanifold is calculated as follows, taking into account the general
expression (\ref{E_D2}),
\begin{align*}%\label{muno}
 M_1& = \{ (q,v,p) \in M \; | \;  d\mathcal{E}(q,v,p) \in 
E^\flat_{\bar{D}_{\Delta}}(q,v,p) \}
\nonumber \\
& =\{(q,v,p)\mid p=\varphi(v),\,v \in \Delta\}.
\end{align*}
Now we simply continue applying the algorithm. Let
\[
W_1=TM_1 \cap
E_{\bar{D}_{\Delta}}=\{(q,v,p,\dot{q},\dot{v},\dot{p})\mid
(q,v,p) \in M_1, \,\, \dot{q}
\in \Delta,\,\dot{p}=\varphi(\dot{v}),\,\dot{v} \in \Delta\}.
\]
Then $\bar{D}_{\Delta}^\flat(W_1)$, denoted simply as $W_1^{\flat}$, is
\[W_1^{\flat}=\{(q,v,p,\alpha,\gamma,\beta)\mid (q,v,p) \in M_1, \,\, \alpha \in
\varphi(\Delta)+\Delta^{\circ},\,\gamma=0,\,\beta \in \Delta\},
\]
and therefore
\begin{align}
\label{mdos}
M_2 & =\{(q,v,p) \in
M_1 \mid  d\mathcal{E}(q,v,p) \in W_1^\flat(q,v,p)\}
\nonumber \\
& =\{(q,v,p) \mid q \in
\psi^{-1}(\varphi(\Delta)+\Delta^{\circ}),\,p=\varphi(v),\,v \in
\Delta\}.
\end{align}
In the same way we can calculate $M_3$. In fact,
\begin{multline*}
W_2=TM_2 \cap
E_{\bar{D}_{\Delta}}
=
\{(q,v,p,\dot{q},\dot{v},\dot{p})\mid (q,v,p) \in M_2, \,\, \dot{q}
\in \Delta,\dot{q} \in
\psi^{-1}(\varphi(\Delta)+\Delta^{\circ}),\\\dot{p}=\varphi(\dot{v}),\,\dot{v}
\in \Delta\},
\end{multline*}
\[W_2^{\flat}=\{(q,v,p,\alpha,\gamma,\beta)\mid (q,v,p) \in M_2, \,\, \alpha \in
\varphi(\Delta)+\Delta^{\circ},\,\gamma=0,\,\beta \in
\Delta \cap \psi^{-1}(\varphi(\Delta)+\Delta^{\circ})\},
\]
then
\begin{align*}%\label{mtres}
M_3
&= \{ (q,v,p) \in
M_2 \; | \;   d\mathcal{E}(q,v,p) \in W_2^\flat(q,v,p) \}  \nonumber \\
&=\{(q,v,p) \mid q \in
\psi^{-1}(\varphi(\Delta)+\Delta^{\circ}),\,p=\varphi(v),\,v \in
\Delta \cap \psi^{-1}(\varphi(\Delta)+\Delta^{\circ})\}  \nonumber \\
&=\{(q,v,p) \mid q \in
\psi^{-1}(\varphi(\Delta)+\Delta^{\circ}),\,p=\varphi(v),\,v \in
\Delta_1\},
\end{align*}
where we have called
\[\Delta_1=\Delta \cap
\psi^{-1}(\varphi(\Delta)+\Delta^{\circ})\subseteq \Delta.\]
Now we shall calculate $M_4$. We will show later that $M_4=M_3$ for all LC circuits with either a positive inductance or a non-infinite positive capacitance on every branch (Lemma \ref{lemmaforDeltaoneequalDeltatwo} and Theorem \ref{thm_positive_nu}).
First we compute
\begin{align*}
W_3
&= TM_3 \cap
E_{\bar{D}_{\Delta}}\\
&=\{(q,v,p,\dot{q},\dot{v},\dot{p})\mid (q,v,p) \in M_3, \,\, \dot{q}
\in \Delta, \dot{q} \in
\psi^{-1}(\varphi(\Delta)+\Delta^{\circ}),\,\dot{p}
=
\varphi(\dot{v}),\,\dot{v} \in \Delta_1\}\\
&=\{(q,v,p,\dot{q},\dot{v},\dot{p})\mid (q,v,p) \in M_3, \,\, \dot{q} \in
\Delta_1,\,\dot{p}= \varphi(\dot{v}),\,\dot{v} \in \Delta_1\}.
\end{align*}
Thus, one finds that
\begin{equation*}
W_3^{\flat}  =\{(q,v,p,\alpha,\gamma,\beta)\mid (q,v,p) \in M_3, \,\, \alpha
\in
\varphi(\Delta_1)+\Delta^{\circ},\,\gamma=0,\,\beta \in
\Delta_1\},
\end{equation*}
and hence
\begin{align*}%\label{mcuatro}
M_4 & = \{ (q,v,p) \in
M_3 \mid d\mathcal{E}(q,v,p) \in W_3^\flat(q,v,p)
\}
\nonumber \\
&  =\{(q,v,p) \in M_3\mid q \in
\psi^{-1}(\varphi(\Delta_1)+\Delta^{\circ}),\,p=\varphi(v),\,v \in
\Delta_1\} \nonumber \\
& =\{(q,v,p) \mid q \in
\psi^{-1}(\varphi(\Delta_1)+\Delta^{\circ}),\,p=\varphi(v),\,v \in
\Delta_1\}.
\end{align*}

Similarly we calculate $M_5$ as follows.
\begin{align*}
W_4 & =TM_4 \cap
E_{\bar{D}_{\Delta}} \\
& =\{(q,v,p,\dot{q},\dot{v},\dot{p})\mid (q,v,p) \in M_4, \,\, \dot{q}
\in \Delta, \dot{q} \in
\psi^{-1}(\varphi(\Delta_1)+\Delta^{\circ}),\,\dot{p} =
\varphi(\dot{v}),\,\dot{v} \in \Delta_1\}\\
& =\{(q,v,p,\dot{q},\dot{v},\dot{p})\mid (q,v,p) \in M_4, \,\, \dot{q} \in
\Delta \cap
\psi^{-1}(\varphi(\Delta_1)+\Delta^{\circ}),\,\dot{p}=
\varphi(\dot{v}),\,\dot{v} \in \Delta_1\}
\end{align*}
Let $\Delta_2=\Delta \cap
\psi^{-1}(\varphi(\Delta_1)+\Delta^{\circ})$; then $\Delta_2
\subseteq \Delta_1$ and we can rewrite $W_4$ as
\[
W_4=\{(q,v,p,\dot{q},\dot{v},\dot{p})\mid (q,v,p) \in M_4, \,\,\dot{q} \in
\Delta_2,\,\dot{p}=\varphi(\dot{v}),\,\dot{v} \in \Delta_1\}.
\]
Then we have
\begin{align*}
 W_4^{\flat} & =\{(q,v,p,\alpha,\gamma,\beta)\mid (q,v,p) \in M_4, \,\,  \alpha
\in
\varphi(\Delta_1)+\Delta^{\circ},\,\gamma=0,\,\beta \in
\Delta_2\}
\end{align*}
\begin{align*}%\label{mcinco}
M_5 & = \{ (q,v,p) \in
M_4 \; | \;  d\mathcal{E}(q,v,p)\in W_4^\flat(q,v,p) \} \nonumber \\
&  =\{(q,v,p) \in M_4\mid q \in
\psi^{-1}(\varphi(\Delta_1)+\Delta^{\circ}),\,p=\varphi(v),\,v \in
\Delta_2\} \nonumber \\
& =\{(q,v,p) \mid q \in
\psi^{-1}(\varphi(\Delta_1)+\Delta^{\circ}),\,p=\varphi(v),\,v \in
\Delta_2\}.
\end{align*}
{}From this, one sees how to recursively  define $M_k$.
For all $k \geq 1$ define
\begin{equation*}%\label{recursiondelta}
\Delta_k=\Delta \cap
\psi^{-1}(\varphi(\Delta_{k-1})+\Delta^{\circ}),
\end{equation*}
where $\Delta_0=\Delta$ by definition.
We have the following expressions for the constraint submanifolds
$M_k$.
\begin{align*}
%\label{recursionM1}
M_1&=\{(q,v,p) \mid p=\varphi(v), v \in \Delta\}\\
%\label{recursionMimpar}
 M_{2k-1} &=
\{(q,v,p) \mid \,q \in
\psi^{-1}(\varphi(\Delta_{k-2})+\Delta^{\circ}),\,p=\varphi (v),\,v
\in \Delta_{k-1}\},\,k \geq 2\\
%\label{recursionMpar}
M_{2k} &= \{(q,v,p) \mid \,q \in
\psi^{-1}(\varphi(\Delta_{k-1})+\Delta^{\circ}),\,p=\varphi (v),\,v \in
\Delta_{k-1}\},\,k \geq 1.
\end{align*}
In order to solve the system we may apply the method to solve the general
equation (\ref{gotaygen122}) explained before.
It is clear that the parametrization
$x = x(z_1,\dots,z_r)$, where $r$ is the dimension of the final constraint
$M_c$,
can be chosen to be a linear map and the IDE (\ref{gotaygen122}) will be a
linear system.\\

This system is of course still an IDE. One of the purposes of this paper is to write Hamilton's equations of motion, using an extension of the Dirac procedure extended to primary foliated constraint submanifolds, developed in section \ref{sectionanextensionofetc}. For this purpose, we will need to consider $M = TE\oplus T^\ast E$ as being the primary foliated constraint embedded naturally in the symplectic manifold $T^\ast TE$.

\subparagraph{Physical interpretation of the constraint equations for LC circuits.} The equations defining the constraint submanifolds have an interesting interpretation in circuit theory terms. From now on we will assume the physically meaningful situation where $0\leq L_i < \infty$, $0<C_i\leq \infty$, $i=1, \dots, n$. We will show that in this case, the algorithm stops either at $M_1$ or at $M_3$.

A circuit has an underlying directed graph, since each branch has a direction in which the current flow will be regarded as positive. A \emph{loop} is a closed sequence of different adjacent branches. It also has a direction, which does not have to be compatible with the directions of the branches involved. The set of loop currents is a set of generators for the KCL subspace $\Delta\subseteq E$.

We shall call a circuit branch \emph{inductive} (resp.\ \emph{capacitive}) if there is an inductor (resp.\ capacitor) present on that branch. A loop will be called
\emph{inductive} (resp.\ \emph{capacitive}) if it has an inductor (resp.\ capacitor) on at least one of its branches.
A branch or loop that is capacitive but not inductive will be called \emph{purely capacitive}. A branch or loop that is neither inductive nor capacitive will be called \emph{empty}, since no capacitors or inductors are present. 
Thus, a non-inductive branch or loop can be either empty or purely capacitive,
and a purely capacitive loop must have at least one capacitor, no inductors, and possibly some empty branches.

The first one of the equations defining $M_1$ is $p=\varphi(v)$, which is $p_i=L_iv_i$, $i=1, \dots, n$. The quantity $L_i \dot v_i$ is the voltage corresponding to an inductor through which the current is $v_i$. Then $p_i$ can be interpreted as a time integral of the voltage on branch $i$ due to the inductor in that branch (the \emph{flux linkage} of the inductor). The second equation, $v\in \Delta$, is just KCL for the currents $v$.

The next constraint submanifold $M_2$ incorporates the equation $\psi(q) \in
\varphi(\Delta)+\Delta^{\circ}$, which represents the KVL equations for the purely capacitive loops, as explained in the next theorem.
\begin{theorem}\label{thm_KVLPC}
The subspace $\varphi(\Delta)+\Delta^{\circ}\subseteq E^*$ represents Kirchhoff's Voltage Law for non-inductive loops (KVLNI for short). That is, its elements are precisely those branch voltage assignments on the circuit that satisfy the subset of KVL equations corresponding to the non-inductive loops. The equation $\psi(q) \in \varphi(\Delta)+\Delta^{\circ}$ represents the condition that on every purely capacitive loop, the branch voltages $q_i/C_i$ of the corresponding capacitors satisfy KVL. We call these the KVLPC equations (KVL for purely capacitive loops), and they form a subset of the KVLNI equations.
\end{theorem}
\begin{proof}
Let $\{\bar e_1, \dots, \bar e_n\}$ be the basis of $E$ where $\bar e_i$ is associated to branch $i$ of the circuit, and let $\{\underline e^1, \dots, \underline e^n\}$ be its dual basis.
Recall that the branch voltage assignments on the circuit are elements of $E^*$, and $\Delta^\circ$ represents Kirchhoff's Voltage Law in the sense that each element of $\Delta^\circ$ is a branch voltage assignment satisfying KVL. The KVL equations are linear equations on the branch voltages and are therefore represented by elements of $E$. In fact, they are precisely the elements of $(\Delta^\circ)^\circ=\Delta$.

Consider a maximal set of independent, non-inductive oriented loops on the circuit, labeled $1, \dots, m$. Each loop gives rise to a linear equation on the corresponding branch voltages according to KVL. These KVL equations are represented by $\eta_1, \dots, \eta_m\in E$, and are a subset of the full KVL equations for the circuit. That is, these linear equations hold on any branch voltage assignment that is compatible with KVL, and therefore $\eta_1, \dots, \eta_m\in \Delta$. Note that each $\eta_j=\eta_j^i \bar e_i$ does not involve any inductive branches; that is, if $L_i\neq 0$ then $\eta_j^i=\eta_j(\underline e^i)=0$. Also, the KVLNI equations $\eta_1, \dots, \eta_m$ can be seen as the KVL equations for the circuit that is obtained by removing the inductive branches from the given circuit. Let us call this \emph{the non-inductive subcircuit}, whose charge space $E_\text{NI}$ is a subspace of $E$ in a natural way. Also, $E_\text{NI}^*\subseteq E^*$ is generated by $\{\underline e^i\}$ where $i$ ranges over all non-inductive branches.

Let us denote $\Delta_\text{NI}\subseteq E_\text{NI}$ the KCL distribution for the non-inductive subcircuit. Note that $\Delta_\text{NI}\subseteq \Delta$, since any KCL-compatible current assignment on the non-inductive subcircuit corresponds to a KCL-compatible current assignment on the original circuit for which no current flows through the inductive branches. Reciprocally, elements of $\Delta$ that are zero on the inductive branches can be regarded as elements of $\Delta_\text{NI}$. Note that $\eta_1, \dots, \eta_m$ are loop currents for the non-inductive subcircuit, and they form a basis of $\Delta_\text{NI}$.

For any $v=v^i\bar e_i\in E$, $\varphi(v)=\sum L_iv^i\underline e^i$. Then for each $j=1, \dots,m$, $\eta_j(\varphi(v))=\sum L_iv^i\eta_j(\underline e^i)=0$. This is true in particular for $v\in\Delta$, so $\eta_j\in (\varphi(\Delta))^\circ$. Therefore $\eta_j\in(\varphi(\Delta))^\circ \cap \Delta= (\varphi(\Delta)+\Delta^\circ)^\circ$. In other words, each $\eta_1, \dots, \eta_m$, seen as a linear equation on $E^*$, holds on $\varphi(\Delta)+\Delta^\circ\subseteq E^*$.

Let us see that $\varphi(\Delta)+\Delta^\circ$ is precisely the vector subspace of $E^*$ defined by the equations $\eta_1, \dots, \eta_m$. Let $\rho=\rho^i\bar e_i\in(\varphi(\Delta)+\Delta^\circ)^\circ=(\varphi(\Delta))^\circ \cap \Delta$ be another equation that holds on  $\varphi(\Delta)+\Delta^\circ$. Then $0=\rho(\varphi(\rho))=\sum L_i(\rho^i)^2$. Since all inductances are nonnegative, this implies that the components of $\rho$ corresponding to the inductive branches must be zero. Also, $\rho\in\Delta$, therefore $\rho\in\Delta_\text{NI}$. This means that $\rho$ is a linear combination of $\eta_1, \dots, \eta_m$. 

The equation $\psi(q) \in \varphi(\Delta)+\Delta^{\circ}$ means that $\eta_j(\psi(q))=0$ for $j=1, \dots, m$. Only those $\eta_j$ corresponding to purely capacitive loops give a condition on $q$. The remaining $\eta_k$ correspond to empty loops, so  $\eta_k(\psi(q))=0$ amounts to the equation $0=0$.
\end{proof}

The third constraint submanifold $M_3$ is obtained from $M_2$ by incorporating the equation $\psi(v) \in
\varphi(\Delta)+\Delta^{\circ}$, which means that the currents on the branches of the purely capacitive loops satisfy the same KVLPC equations as the charges. It is clear that it must hold when we consider the dynamics, which includes the equation $\dot q=v$.

Suppose that there exists at least one purely capacitive loop so that the algorithm does not stop at $M_1$. Then it is possible to show that $M_3\subsetneq M_2$. In order to do that, consider a purely capacitive loop, which by definition must involve at least one capacitor, and recall that all the capacitances are positive. The corresponding KVLPC equation is represented by $\eta\in E$, whose components are $\pm 1$ or $0$. As in the proof of the previous theorem, $\eta\in(\Delta^\circ)^\circ=\Delta$, so it can be interpreted as a nonzero loop current, that is, a current that is $\pm 1$ on all the branches of that loop (depending on the relative orientation of the branches with respect to the loop) and zero on the remaining branches. Note that $\psi(\eta)\in E^*$ does not satisfy KVLPC, since in particular $\eta(\psi(\eta))=\sum (\eta^i)^2/C_i=\sum 1/C_i> 0$, where the last sum is over the branches in the chosen loop. Therefore $\Delta \not\subseteq \psi^{-1}(\varphi(\Delta)+\Delta^{\circ})$ and  $M_3\subsetneq M_2$.\\

Now we will prove the interesting fact that since inductances are greater than or equal to zero and capacitances are positive numbers or $+ \infty$ the algorithm always stops at $M_3$ (or earlier). 
We shall start by proving the following lemma.
\begin{lemma}\label{lemmaforDeltaoneequalDeltatwo}
\textbf{(a)}
$\delta_1 \in \psi^{-1}\left( \varphi (\Delta)  + \Delta ^\circ\right)$ is equivalent to the condition that there exists $\delta_0 \in \Delta$ such that $\psi (\delta_1) - \varphi (\delta_0) \in  \Delta ^\circ$.\\

\textbf{(b)}
The condition $\Delta_1 =  \Delta_2$ is equivalent to the following condition:\\

\textbf{Condition} $(*)$

For any $\delta \in \Delta$ there exist $\delta_0, \delta_1 \in  \Delta$
such that the following hold:
\begin{align*}
%\label{lemmaDeltasubonea}
 \varphi \left( \delta \right) - \varphi \left( \delta_1 \right) 
&\in  \Delta ^\circ\\
%\label{lemmaDeltasuboneb}
 \psi \left( \delta_1 \right) - \varphi \left( \delta_0 \right) 
&\in  \Delta ^\circ.
\end{align*}

Equivalently, we can say that for each
 $\delta \in \Delta$ there exists  $\delta_1 \in  \Delta_1$ such that  
$\varphi \left( \delta \right) - \varphi \left( \delta_1 \right) 
\in  \Delta ^\circ$.\\

\textbf{(c)}
Condition $(*)$ implies that if
$M_4 = M_5$ then $M_3 = M_4$.
\end{lemma}

\begin{proof}
The proof of \textbf{(a)} is immediate taking into account the definitions. To prove \textbf{(b)}
take any given $\delta_1 \in  \Delta_1$, then it satisfies the condition
$\psi \left( \delta_1 \right) - \varphi \left( \delta_0 \right) \in \Delta^\circ$, where $\delta_0 \in \Delta$. 
By condition $(*)$ there exists $\delta^\prime_1 \in  \Delta_1$ such that $\varphi \left( \delta_0 \right) - \varphi \left( \delta^\prime_1 \right) \in  \Delta^\circ$, and we can conclude that $\psi \left( \delta_1 \right) - \varphi \left( \delta^\prime_1 \right) 
\in  \Delta ^\circ$, which shows that $\delta_1 \in \Delta_2$. 
The converse can be easily verified.

Now we shall prove \textbf{(c)}. Since the inclusion $M_4 \subseteq M_3$ is immediate by construction we will only prove the converse. Let $(q,v,p) \in M_3$ then, in particular, there exists $\delta \in \Delta$ such that
$\psi(q) - \varphi(\delta) \in \Delta^\circ$. By condition $(*)$ there exists $\delta_1 \in \Delta_1$ such that
$\varphi(\delta_1) - \varphi(\delta) \in \Delta^\circ,$ which implies that $\psi(q) - \varphi(\delta_1) \in \Delta^\circ$,
that is, $q \in \psi^{-1}\left( \varphi(\Delta_1) + \Delta^\circ)\right)$, from which we can deduce that
$M_3 \subseteq M_4$.
\end{proof}
The following theorem gives a sufficient condition under which condition $(*)$ in the previous lemma holds, which considers the standard physically meaningful case. Also, we assume for simplicity that there are no empty branches. We should mention that it can be proven that even with empty branches, the algorithm stops at $M_1$ or $M_3$.

\begin{theorem}\label{thm_positive_nu}
Condition $(*)$ in Lemma \ref{lemmaforDeltaoneequalDeltatwo} holds if in every branch there is either a positive inductance or a non-infinite positive capacitance.
\end{theorem}
\begin{proof}
Condition $(*)$ can be reformulated as follows:

For any $\delta \in \Delta$ there exists $\delta_0, \delta_1 \in  \Delta$
such that the following equalities are satisfied
\begin{align}
\label{lemmaDeltasuboneaprime}
\langle \varphi \left( \delta_1 \right), \delta^\prime_1 \rangle
&= \langle\varphi \left( \delta \right), \delta^\prime_1\rangle \\
\label{lemmaDeltasubonebprime}
\langle (\psi + \varphi) \left( \delta_1 \right), \delta^\prime_0 \rangle
-
\langle\varphi \left( \delta_0 \right), \delta^\prime_0\rangle 
&=
 \langle\varphi \left( \delta \right), \delta^\prime_0\rangle,
\end{align}%
for all 
$\delta^\prime_1,  \delta^\prime_0 \in \Delta$. This formulation has the advantage that
$\nu := \psi + \varphi$
is nonsingular.

The basis $\bar{e}_i$, $i = 1, \dots,n$ of the charge space $E$ defines naturally an Euclidean metric by the condition
$\langle\bar{e}_i, \bar{e}_j\rangle = \delta_{ij}$ and an identification $E \equiv E^*$ by the condition
$\bar{e}_i = \underline{e}^i$, $i = 1, \dots,n$,
in particular obtains $\Delta^\circ \equiv \Delta^\perp$. The linear map $\nu$ is self-adjoint and positive definite while $\varphi$ is self-adjoint and positive semi-definite. System 
(\ref{lemmaDeltasuboneaprime})--(\ref{lemmaDeltasubonebprime}) can be written in the form
\begin{align*}
%\label{lemmaDeltasuboneaprimex}
P_{\Delta} \circ \varphi | \Delta (\delta_1)
&= P_{\Delta} \circ \varphi | \Delta (\delta)\\
%\label{lemmaDeltasubonebprimex}
P_{\Delta} \circ \nu | \Delta (\delta_1)
-
P_{\Delta} \circ \varphi | \Delta (\delta_0)
&=
P_{\Delta} \circ \varphi | \Delta (\delta_0).
\end{align*}%
where $P_{\Delta}  \colon E \rightarrow \Delta$ is the orthogonal projection on 
$\Delta$.

One can choose an orthonormal basis of $\Delta$ and then the system of equations 
(\ref{lemmaDeltasuboneaprime})--(\ref{lemmaDeltasubonebprime}) can be written in matrix form. Moreover we have an orthogonal decomposition $\Delta = \Delta^{(1)} + \Delta^{(2)}$, 
where $\Delta_1$ and $\Delta_2$ are the image and the kernel of  
$P_{\Delta} \circ \varphi | \Delta$. Each vector $\delta \in \Delta$ is decomposed as
$\delta = \delta^{(1)} + \delta^{(2)}$.
We can choose the basis in such a way that the block decomposition of the matrix representing
$P_{\Delta}  \circ\varphi | \Delta$ has the form
\[
\left[\begin{array}{cc}
\varphi^{(1,1)}&\varphi^{(1,2)}\\
\varphi^{(2,1)}&\varphi^{(2,2)}
\end{array} \right],
\]
where
$\varphi^{(1,2)} = \varphi^{(2,1)} = \varphi^{(2,2)} = 0$ and $\varphi^{(1,1)}$ is a diagonal matrix with positive eigenvalues. The map $P_{\Delta} \circ \nu | \Delta$ also has a block decomposition, and since it is self-adjoint and positive definite, the blocks 
$\nu^{(1,1)}$ and $\nu^{(2,2)}$ are symmetric and positive definite matrices. We obtain the following system of equations, which is equivalent to (\ref{lemmaDeltasuboneaprime})--(\ref{lemmaDeltasubonebprime}),
\begin{align}
\label{positiveinductancesand capacities1}
\left[\begin{array}{cc}
\varphi^{(1,1)}&\varphi^{(1,2)}\\
\varphi^{(2,1)}&\varphi^{(2,2)}
\end{array} \right] 
\left[
\begin{array}{c} \delta_1^{(1)}\\
\delta_1^{(2)}\end{array}\right]
&=
\left[\begin{array}{cc}
\varphi^{(1,1)}&\varphi^{(1,2)}\\
\varphi^{(2,1)}&\varphi^{(2,2)}
\end{array} \right] 
\left[\begin{array}{c} \delta^{(1)}\\
\delta^{(2)}\end{array}\right]\\
\label{positiveinductancesand capacities2}
\left[\begin{array}{cc}
\nu^{(1,1)}&\nu^{(1,2)}\\
\nu^{(2,1)}&\nu^{(2,2)}
\end{array} \right] 
\left[
\begin{array}{c} \delta_1^{(1)}\\
\delta_1^{(2)}\end{array}\right]
-
\left[\begin{array}{cc}
\varphi^{(1,1)}&\varphi^{(1,2)}\\
\varphi^{(2,1)}&\varphi^{(2,2)}
\end{array} \right] 
\left[\begin{array}{c} \delta_0^{(1)}\\
\delta_0^{(2)}\end{array}\right]
&=
\left[\begin{array}{cc}
\varphi^{(1,1)}&\varphi^{(1,2)}\\
\varphi^{(2,1)}&\varphi^{(2,2)}
\end{array} \right] 
\left[\begin{array}{c} \delta^{(1)}\\
\delta^{(2)}\end{array}\right]
\end{align}
For given $\delta$, equation (\ref{positiveinductancesand capacities1}) fixes 
$\delta_1^{(1)}$ and imposes no condition on $\delta_1^{(2)}$. Using the fact that
$\varphi^{(1,1)}$
and
$\nu^{(2,2)}$
are invertible we can find $\delta_0^{(1)}$ and $\delta_1^{(2)}$ to satisfy 
(\ref{positiveinductancesand capacities2}).
\end{proof}

As a conclusion, if (and only if) there are no purely capacitive loops in the circuit, then $M_2=M_1$ and the final constraint submanifold is $M_1$, defined by the conditions that $p_i=L_iv_i$ (where some $L_i$ might be zero) and $v$ satisfies KCL. Otherwise, the final constraint submanifold is $M_3$, defined by the conditions that $p_i=L_iv_i$, $v$ satisfies KCL, and $q_i/C_i$ and $v_i/C_i$ satisfy KVLPC. 
Recall that $q_i/C_i$ is the voltage of the capacitor in branch $i$, and the absence of a capacitor on branch $i$ means $C_i=\infty$. The rest of the KVL equations, which involve the branch voltages $\dot p_i=L_i \dot v_i$ on the inductors, will appear when considering the dynamics. In this sense, the KVLPC equations can be regarded as ``static'' KVL equations.

\begin{remark}
The hypothesis that no capacitances or inductances are negative, besides being physically meaningful, is crucial to ensure that the algorithm stops at $M_3$ (or $M_1$). For example, for a circuit with one inductor $L$, one capacitor $C$ and a negative capacitor $-C$, all of them in parallel, it stops at $M_5$. 
\end{remark}

\paragraph{Geometry of the final constraint submanifold.} We know that the solution curves will have tangent vectors in $W_c=TM_c\cap E_{\bar D_\Delta}$, where $c=1$ or $c=3$ depending on the case. 

\begin{theorem} Assume the physically meaningful situation where $0 \leq L_i<\infty$, $0<C_i\leq\infty$, $i=1, \dots,n$. The distribution $W_c\subseteq TM_c$ is constant and therefore integrable, and its integral leaves are preserved by the flow. Denote the leaf through $0$ by $\widetilde W_c$, so these leaves can be written as $x+\widetilde W_c$, where $x=(q,v,p)$. Each integral leaf $x+\widetilde W_c$ is a symplectic manifold with the pullback of the presymplectic form $\omega=dq\wedge dp$, if and only if there are no empty loops.
\end{theorem}

\begin{proof} As mentioned before, if there are no purely capacitive loops, then $c=1$, otherwise $c=3$. For the first case, using \eqref{E_D1}, we have
\begin{align*}
W_1&=\{(q,v,p, \dot q, \dot v, \dot p)\in TE\oplus T^*E | (q,v,p, \dot q, \dot v, \dot p)\in TM_1, \dot q\in \Delta\}\\
&=\{(q,v,p, \dot q, \dot v, \dot p)\in TE\oplus T^*E | p=\varphi(v), v\in\Delta, \dot p=\varphi( \dot v), \dot v\in\Delta, \dot q\in \Delta\}
\end{align*}
Then $W_1\subseteq TM_1$ is a constant and therefore integrable distribution whose rank is $2\dim\Delta$. Write $\omega_{x+\widetilde W_1}$ for the pullback of the presymplectic form $\omega=dq\wedge dp$ to each integral leaf.
Let us compute $\ker\omega_{x+\widetilde W_1}$.

Consider $(q,v,p, \dot q, \dot v, \varphi( \dot v)), (q,v,p, \dot q', \dot v', \varphi( \dot v'))\in W_1$, where $\dot q, \dot v,\dot q', \dot v'\in \Delta$ are arbitrary. Assume $(q,v,p, \dot q, \dot v, \varphi( \dot v))\in \ker\omega_{x+\widetilde W_1}$, that is,
\[
0=\omega((q,v,p, \dot q, \dot v, \varphi( \dot v)), (q,v,p, \dot q', \dot v', \varphi( \dot v')))=\dot q(\varphi( \dot v'))-\dot q'(\varphi( \dot v))
\]
for any $\dot q', \dot v'\in \Delta$. This means that $\dot q, \dot v\in (\varphi(\Delta))^\circ$, but since they also belong to $\Delta$ then $\dot q, \dot v\in (\varphi(\Delta))^\circ\cap \Delta=(\varphi(\Delta)+ \Delta^\circ)^\circ$. That is, $\ker\omega_{x+\widetilde W_1}$ is defined by $\dot q, \dot v\in(\varphi(\Delta)+ \Delta^\circ)^\circ$, $\dot p=\varphi( \dot v)$. The subspace $(\varphi(\Delta)+ \Delta^\circ)^\circ$ is generated by the loop currents on the non-inductive loops (Theorem \ref{thm_KVLPC}). Since there are no purely capacitive loops, then it is generated by the loop currents on the empty loops. As a conclusion, $x+\widetilde W_1$ is symplectic if and only if there are no empty loops on the circuit.
 
If there is at least one purely capacitive loop, then the final constraint submanifold is $M_3$. Then
\begin{align*}
W_3&=\{(q,v,p, \dot q, \dot v, \dot p)\in TE\oplus T^*E | (q,v,p, \dot q, \dot v, \dot p)\in TM_3, \dot q\in \Delta\}\\
&=\{(q,v,p, \dot q, \dot v, \dot p)\in TE\oplus T^*E | p=\varphi(v), v\in\Delta, q\in\psi^{-1}(\varphi(\Delta)+\Delta^\circ),\\ 
&\quad\quad v\in\psi^{-1}(\varphi(\Delta)+\Delta^\circ) ,
\dot p=\varphi( \dot v), \dot v\in\Delta, \dot q\in\psi^{-1}(\varphi(\Delta)+\Delta^\circ),\\
&\quad\quad \dot v\in\psi^{-1}(\varphi(\Delta)+\Delta^\circ),\dot q\in \Delta\}.
\end{align*}
Again, $W_3\subseteq TM_3$ is a constant and therefore integrable distribution whose rank is $2\dim (\Delta\cap\psi^{-1}(\varphi(\Delta)+\Delta^\circ))=2\dim \Delta_1$. Write $\omega_{x+\widetilde W_3}$ for the pullback of the presymplectic form $\omega=dq\wedge dp$ to each integral leaf, and
let us compute $\ker\omega_{x+\widetilde W_3}$.

Consider $(q,v,p, \dot q, \dot v, \varphi( \dot v)), (q,v,p, \dot q', \dot v', \varphi( \dot v'))\in W_3$, where $\dot q, \dot v,\dot q', \dot v'\in \Delta_1$ are arbitrary. Assume $(q,v,p, \dot q, \dot v, \varphi( \dot v))\in \ker\omega_{x+\widetilde W_3}$, that is,
\[
0=\omega((q,v,p, \dot q, \dot v, \varphi( \dot v)), (q,v,p, \dot q', \dot v', \varphi( \dot v')))=\dot q(\varphi( \dot v'))-\dot q'(\varphi( \dot v))
\]
for any $\dot q', \dot v'\in \Delta_1$. This means that $\dot q, \dot v\in (\varphi(\Delta_1))^\circ$. If there is an empty loop, represented by $\eta\neq 0$, then $\eta\in \Delta$, $\psi(\eta)=0$ and $\eta\in \Delta_1$. Also, $\varphi(\eta)=0$, so $(q,v,p,\eta,\eta,0)$ is a nonzero element of $\ker\omega_{x+\widetilde W_3}(q,v,p)$ and $x+\widetilde W_3$ is not symplectic. 

If there are no empty loops, we apply the following argument to $\dot q$, and the same conclusion holds for $\dot v$. Since $\dot q\in (\varphi(\Delta_1))^\circ\cap \Delta_1$, then $\dot q(\varphi( \dot q))=0$, which means that the components of $ \dot q$ corresponding to the inductive branches are zero, as we reasoned in the proof of Theorem \ref{thm_KVLPC}. Also, $\dot q\in \Delta_1\subset \Delta$, so $\dot q\in\Delta_\text{NI}$, that is, it can be regarded as a branch current assignment on the non-inductive subcircuit, satisfying KCL. In addition, $\dot q_i/C_i$ satisfy KVLNI, which means that $\dot q_i/C_i$ satisfy KVL for the non-inductive subcircuit. Now we will apply the fact that the dynamics of a circuit with no inductors and no empty loops consists only of equilibrium points. Indeed, KVL for such a circuit is a homogeneous linear system on the charges $q_i$ of the capacitors, so it has a solution $\bar q=(\bar q_1, \dots, \bar q_n)$. Setting $q(t)=\bar q$ and $v(t)= \dot q(t)=0$, we have that $v$ satisfies KCL trivially. If there was a nonzero loop current $\eta$, then it would produce a change $\dot q=\eta$ in the charges of the capacitor on that loop. Then they would no longer satisfy KVL for that loop, since $d/dt(\eta(\psi(q)))=\eta(\psi(\eta))\neq 0$. Then the currents must therefore be zero, that is, $\dot q=0$. Of course, this is not true if there are empty loops, which can have arbitrary loop currents. By the same reasoning, $\dot v=0$, so $\ker\omega_{x+\widetilde W_3}=0$ and $x+\widetilde W_3$ is symplectic.
\end{proof}

Let us now write equations of motion in Hamiltonian form, by working on a symplectic leaf. This means that the initial conditions must belong to that particular leaf. This is related to the comments on equations (\ref{equation15}) and (\ref{SequationalgorithGN})
made before. An equation of motion in Poisson form that is valid for all initial conditions in $M_c$ will be given in Theorem \ref{thm321EXTENDED} by the equations
\eqref{eq:poisson_dirac_evolution_thmfoliated} and
\eqref{eq:vector_field_foliated}.

First, let us represent the KCL equations by $\kappa_1, \dots, \kappa_a\in E^*$, in the sense that $\kappa_i(v)=0$ represents the KCL equation corresponding to node $i$. The number $a$ is the number of independent nodes (usually the number of nodes minus one). As before, denote the KVLPC equations by $\eta_1, \dots, \eta_b\in E$, where $b$ is the number of independent purely capacitive loops.  
Seeing each $\kappa_i$ as a row vector and each $\eta_j$ as a column vector, write
\[
K=\left[\begin{array}{c}
\kappa_1\\
\vdots\\
\kappa_a
\end{array} \right],\quad\quad\quad
\Gamma=\left[\begin{array}{ccc}
\eta_1&
\dots&
\eta_b
\end{array} \right].
\]

The space $TE\oplus T^*E$ with the presymplectic form $\omega$ is embedded in $T^*TE$ with the canonical symplectic form (see Appendix \ref{Pontryagin_embedding}), where the variables are $(q,v,p,\nu)$. The integral leaf $x_0+\widetilde W_c$, where $c=1$ or $c=3$ as before, and $x_0=(q_0,v_0,p_0)$, is defined as a subspace of $T^*TE$ regularly by the equations
\begin{align*}
\kappa_i(q)-\kappa_i(q_0)&=0,\qquad i=1, \dots,a\\
\eta_i(\psi(q))&=0,\qquad i=1, \dots,b\\
p_i-L_iv_i&=0,\qquad i=1, \dots,n\\
\kappa_i(v)&=0,\qquad i=1, \dots,a\\
\eta_i(\psi(v))&=0,\qquad i=1, \dots,b\\
\nu_i&=0,\qquad i=1, \dots,n
\end{align*}
and the matrix of Poisson brackets of these constraints is, in block form,
\[\Sigma=\left[\begin{array}{cccccc}
0&0&K&0&0&0\\
0&0&\Gamma^T\psi&0&0&0\\
-K^T&-\psi\Gamma&0&0&0&-\varphi\\
0&0&0&0&0&K\\
0&0&0&0&0&\Gamma^T\psi\\
0&0&\varphi&-K^T&-\psi\Gamma&0
\end{array} \right],
\]
where $\varphi$ and $\psi$ stand for the diagonal matrices with $L_i$ and $1/C_i$ along their diagonals, respectively.

\paragraph{A concrete example.}
We shall illustrate our method with the simple LC circuit studied in
\cite{MR2265464} which is shown in Figure \ref{circuitfigure}. It is a $4$-port
LC circuit where the configuration space is
$E=\mathbb{R}^4$.

\begin{figure}[ht]
\begin{center}
\includegraphics[scale=1]{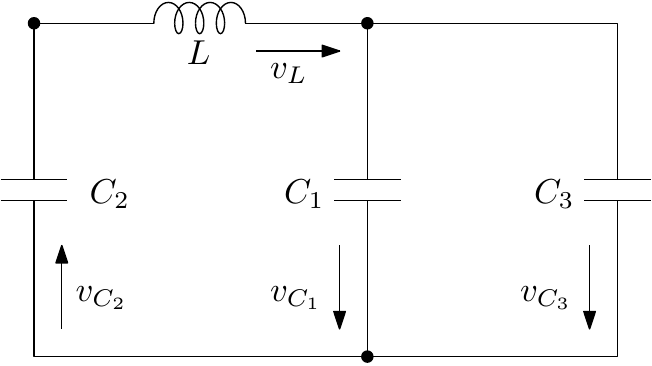}
\end{center}
\caption{\footnotesize
The constraint algorithm is applied to this 4-port LC circuit.}
\label{circuitfigure}
\end{figure}

We shall use the notation $q=(q_L,q_{C_1},q_{C_2},q_{C_3}) \in E$,
$v=(v_L,v_{C_1},v_{C_2},v_{C_3}) \in T_q{E}$,
 and
$p=(p_L,p_{C_1},p_{C_2},p_{C_3}) \in T_q^*{E}$. The Lagrangian of
the LC circuit is $\mathcal{L}\colon TE  \rightarrow \mathbb{R}$,
\[\mathcal{L}(q,v)=\frac{1}{2}L(v_L)^2-\frac{1}{2}\frac{(q_{C_1})^2}{C_1}-
\frac{1}{2}\frac{(q_{C_2})^2}{C_2}-\frac{1}{2}\frac{(q_{C_3})^2}{C_3}.\]
We will assume the physically meaningful case where $0<L$ and $0<C_i<\infty$, $i=1,2,3$. Then, by Theorem \ref{thm_positive_nu}, the CAD algorithm will stop at most at $M_3$. Also, Theorem \ref{thm_KVLPC} gives a direct description of the specific equations that define $M_1$, $M_2$ and $M_3$. This circuit has a purely capacitive loop, so $M_1\supsetneq M_2\supsetneq M_3$.
 However, the algorithm could also be applied for the case where negative values for $L$ or $C_i$ are allowed. For this particular circuit, the algorithm stops at $M_3$ regardless of the signs of the inductance and capacitances.

The KCL constraints $\Delta \subseteq TE$ for the current $v$ are
\begin{align}
\label{KCLconstcoord1}
-v_L + v_{C_2} & = 0
\\
\label{KCLconstcoord2}
 -v_{C_1}+ v_{C_2}-v_{C_3} & = 0
\end{align}
Therefore the constraint KCL space is defined, for each $q \in E$,
by
\[ \Delta(q) =\{v \in T_q{E}\mid \langle p^a,v \rangle =0,\,
a=1,\,2 \},\]
where
\[
(p_k^a) =\left(\begin{array}{cccc} -1 & 0 & 1 & 0 \\
0 & -1 & 1 & -1 \end{array} \right).
\]
On the other hand, the annihilator $\Delta^{\circ}$ of $\Delta$ is
the constraint KVL space, defined, for each $q \in E$, by
\[ \Delta^{\circ}(q) =\{p \in T_q^{*}{E}\mid \langle p,v \rangle =0,
\mbox{ for all $v \in \Delta(q)$} \}.\]
Note that $\{p^a,\,a=1,2\}$ is a basis of $\Delta^{\circ}$. Taking into account that the Dirac structure $\bar{D}_{\Delta}
\subseteq TM \oplus T^*M$ on $M=TE \oplus T^*E$ associated to the
space $\Delta(q)$ is, for each $(q,v,p) \in M$, given by
\begin{align*}
& \bar{D}_{\Delta}(q,v,p)  \\
& \quad =\{(q,v,p,\dot{q},\dot{v},\dot{p},\alpha,\gamma,\beta)
\in TM \oplus T^*M\mid \dot{q} \in \Delta(q),\,\alpha + \dot{p}
\in \Delta^{\circ}(q),\,\beta=\dot{q},\,\gamma=0\},
\end{align*}
and the corresponding energy $\mathcal{E} \colon M \rightarrow
\mathbb{R}$, is given by $\mathcal{E} (q,v,p)=pv-\mathcal{L}(q,v)$,
we can easily verify that the Dirac dynamical system
\[(q,v,p,\dot{q},\dot{v},\dot{p})\oplus d \mathcal{E}(q,v,p)  \in
\bar{D}_{\Delta}(q,v,p)\]
is equivalent to the IDE system that is given in coordinates by
(\ref{Cnnhhooll1})--(\ref{Cnnhhooll4}), which gives
\begin{align}
\label{IDEcircuitI}
\dot{q}_L = v_L,\,\, \dot{q}_{C_1} =
v_{C_1},\,\,  \dot{q}_{C_2} = v_{C_2},\,\,  \dot{q}_{C_3} =
v_{C_3}\\
\label{IDEcircuitII} \dot{p}_L+\frac{q_{C_1}}{C_1}+\dot{p}_{C_1}+
\frac{q_{C_2}}{C_2}+\dot{p}_{C_2}=0\\
\label{IDEcircuitIII} \dot{p}_L+\frac{q_{C_2}}{C_2}+\dot{p}_{C_2}+
\frac{q_{C_3}}{C_3}+\dot{p}_{C_3}=0\\
\label{IDEcircuitIV}
p_L=Lv_L,\,\,p_{C_1}=p_{C_2}=p_{C_3}=0 \\
\label{IDEcircuitV} v_L=v_{C_2},\,\,v_{C_1}=v_{C_2}-v_{C_3}.
\end{align}
We now apply the constraint algorithm CAD for Dirac dynamical systems.
We calculate the expressions of
\[
\varphi(v)=\frac{\partial{\mathcal{L}}}{\partial{v}}\,\,\,
\mbox{and}\,\,\, \psi(q)=\frac{\partial{\mathcal{L}}}{\partial{q}}\]
to get
\begin{align*}
%\label{varphi}
\varphi(v_L,v_{C_1},v_{C_2},v_{C_3})&=(Lv_{L},0,0,0)\\
%\label{psi}
\psi(q_L,q_{C_1},q_{C_2},q_{C_3})&= (0,- \frac{q_{C_1}}{C_1},-\frac{q_{C_2}}{C_2},- \frac{q_{C_3}}{C_3}).
\end{align*}
From now on we will consider the constant distribution $\Delta$ as a subspace of $E$. As we have seen before $M_1=\{(q,v,p)\mid v \in \Delta,\,
p-\varphi(v)=0\}$, so using (\ref{KCLconstcoord1})--(\ref{KCLconstcoord2}) and the expression for $\varphi$ we get
\begin{align*}
%\label{munoej} \nonumber
M_1&
=\{(q_L,q_{C_1},q_{C_2},q_{C_3},v_L,v_{C_1},v_{C_2},v_{C_3},p_L,p_{C_1},p_{C_2},
p_{C_3})
\mid \\
 & \quad \quad p_L=Lv_L,\,\,p_{C_1}=p_{C_2}=p_{C_3}=0,
v_L=v_{C_2},\,\,v_{C_1}=v_{C_2}-v_{C_3}\}.
\end{align*}
We could calculate $M_2$ using the expression
(\ref{mdos}). However, according to Theorem \ref{thm_KVLPC}, $M_2$ adds one more constraint, corresponding to the KVL for the purely capacitive loop that involves capacitors $C_1$ and $C_3$. Then
\begin{align*}
%\label{mdosej}
M_2 & = \left\{ 
(q_L,q_{C_1},q_{C_2},q_{C_3},v_L,v_{C_1},v_{C_2},v_{C_3},p_L,p_{C_1},p_{C_2},p_{
C_3})
\; \left|  \;  p_L=Lv_L, \phantom{\frac{q_{C}}{C}} \right. \right. \nonumber \\
& \qquad  \left. p_{C_1}=p_{C_2}=p_{C_3}=0,
 v_L=v_{C_2},\,v_{C_1}=v_{C_2}-v_{C_3},\,\frac{q_{C_1}}{C_1}=\frac{q_{C_3}}{C_3}
\right\} .
\end{align*}
As explained right after the proof of Theorem \ref{thm_KVLPC}, the final constraint submanifold $M_3$ adds the constraint $v_{C_1}/C_1=v_{C_3}/C_3$:
\begin{align}
\label{mtresej}
 \nonumber
M_3 & =\Big\{
(q_L,q_{C_1},q_{C_2},q_{C_3},v_L,v_{C_1},v_{C_2},v_{C_3},p_L,p_{C_1},p_{C_2},p_{
C_3})
 \,\Big|
\nonumber \\
& \qquad  p_L=Lv_L,\,p_{C_1}=p_{C_2}=p_{C_3}=0,\,
v_L=v_{C_2},\,v_{C_1}=v_{C_2}-v_{C_3},\,
\nonumber \\
& \qquad \qquad \frac{q_{C_1}}{C_1}=\frac{q_{C_3}}{C_3},\,
\frac{v_{C_1}}{C_1}=\frac{v_{C_3}}{C_3} \Big\}.
\end{align}

To solve the system we can simply parametrize $M_3$, which has
dimension $4$, for instance by taking some appropriate $4$ of the
$12$ variables
\[
(q_L,q_{C_1},q_{C_2},q_{C_3},v_L,v_{C_1},v_{C_2},v_{C_3},p_L,p_{C_1},p_{C_2},p_{
C_3})
\]
as being independent parameters and then replace in equations
(\ref{IDEcircuitI})--(\ref{IDEcircuitIII}) to obtain an ODE
equivalent to equations of motion.
For instance, if we choose
$q_L$, $q_{C_1}$, $q_{C_2}$, $p_L$, as independent variables we
get the ODE
\begin{align*}
%
%\label{eqtnofmotion1}
 \dot{q}_L & = \frac{p_L}{L} \\
%\label{eqtnofmotion2}
 %
 \dot{q}_{C_1} & = \frac{C_1}{L(C_1+C_3)}p_L \\
%\label{eqtnofmotion3}
 %
 \dot{q}_{C_2} & = \frac{p_L}{L} \\
%\label{eqtnofmotion4}
 %
 \dot{p}_L & = -\frac{q_{C_1}}{C_1}-\frac{q_{C_2}}{C_2}.
 \end{align*}

We know that solutions should be tangent to
$W_3 = TM_3 \cap E_{\widetilde{D}_\Delta}$
where, according to (\ref{E_D1}),
$E_{\widetilde{D}_\Delta}$
is defined by the conditions
$\dot{q} \in \Delta$,
that is,
\begin{equation}\label{Deltaequation}
\dot{q}_L=\dot{q}_{C_2},\,\dot{q}_{C_1}=\dot{q}_{C_2}-\dot{q}_{C_3}.
\end{equation}
Therefore
$W_3$ is defined by
(\ref{Deltaequation})
together with the equations defining $TM_3$
which are obtained by differentiating with respect to time the equations
(\ref{mtresej})
defining $M_3$, that is,
\begin{equation*}%\label{tm3equation}
\dot{p}_L=L\dot{v}_L,\,\dot{p}_{C_1}=\dot{p}_{C_2}=\dot{p}_{C_3}=0,\,
\dot{v}_L=\dot{v}_{C_2},\,\dot{v}_{C_1}=\dot{v}_{C_2}-\dot{v}_{C_3},\,
\frac{\dot{q}_{C_1}}{C_1}=\frac{\dot{q}_{C_3}}{C_3},\,
\frac{\dot{v}_{C_1}}{C_1}=\frac{\dot{v}_{C_3}}{C_3}.
\end{equation*}
Then
$W_3$ is a constant and therefore integrable distribution of rank 2. Its integral leaves
are
preserved by the flow. One can check that the pullback of the presymplectic form
$\omega = dq \wedge dp$ to each leaf is nondegenerate, so they are symplectic
manifolds. Denote $\widetilde W_3=W_3(0)$, so these leaves can be written as $x+\widetilde W_3$.

For a given
$x_0 \in M_3$,
say
\[
x_0
=
(q_{L0},q_{C_10},q_{C_20},q_{C_30},v_{L0},
v_{C_10},v_{C_20},v_{C_30},p_{L0},p_{C_10},p_{C_20},p_{C_30})
\]
we have that elements of the symplectic leaf $x_0 + \widetilde{W}_3$ are
characterized by the conditions defining $M_3$ plus the condition
obtained by integrating with respect to time the conditions
(\ref{KCLconstcoord1}) and (\ref{KCLconstcoord2}), applied to
$(q_{L} - q_{L0}, q_{C_1} - q_{C_10}, q_{C_2} - q_{C_20},q_{C_3} -
q_{C_30})$, that is
\begin{align*}
-(q_{L} - q_{L0}) + (q_{C_2} - q_{C_20}) & =  0
\\
 -(q_{C_1} - q_{C_10})+ (q_{C_2} - q_{C_20})
-
(q_{C_3} - q_{C_30}) & =  0.
\end{align*}

For simplicity, we will assume from now on that the conditions
\begin{align*}
q_{L0} - q_{C_20} & =  0
\\
q_{C_10} - q_{C_20}
+ q_{C_30} & =  0,
\end{align*}
are satisfied. While this is not the most general case, it is enough to
illustrate the procedure.\\

We can conclude that
$x_0 + \widetilde{W}_3$
has dimension
$2$
and in fact the projection
$\bar{\pi}(x_0 + \widetilde{W}_3) \subseteq T^\ast E$
is simply the subspace of
$T^\ast E$
defined by the conditions
$q_{C_1} = q_{C_2} = q_{C_3} = 0$
and
$p_{C_1} = p_{C_2} = p_{C_3} = 0$.
Therefore one can use the variables
$(q_L, p_L)$
to parametrize
$\bar{\pi}(x_0 + \widetilde{W}_3)$,
and, in fact, we obtain
\begin{align}
 \label{parametrizationofsymplleaf1}
v_L &= \frac{p_L}{L}&
q_{C_1} &=\frac{C_1}{C_1+C_3}q_L\\
\label{parametrizationofsymplleaf2}
v_{C_1} &= \frac{C_1}{C_1 + C_3}\frac{1}{L}p_L&
q_{C_2} &=q_L\\
\label{parametrizationofsymplleaf3}
v_{C_2} &= \frac{p_L}{L}&
q_{C_3} &=\frac{C_3}{C_1+C_3}q_L\\
\label{parametrizationofsymplleaf4}
v_{C_3} &= \frac{C_3}{C_1 + C_3}\frac{1}{L}p_L&
p_{C_1} &=
p_{C_2} =
p_{C_3} = 0\\
\label{parametrizationofsymplleaf5}
& &
\nu_L &=
\nu_{C_1} =
\nu_{C_2} =
\nu_{C_3} =0.
\end{align}
Let $x_0 + \widetilde{W}_3$ be defined by
$\epsilon_j=0$, $j=1,\ldots,14$, where we have chosen
\begin{align*}
\epsilon_1 & = p_L-Lv_L&
\epsilon_8 & = \frac{v_{C_1}}{C_1}-\frac{v_{C_3}}{C_3}\\
\epsilon_2 & = p_{C_1} &
\epsilon_9 &= q_L - q_{C_2} \\
\epsilon_3 & = p_{C_2} &
\epsilon_{10} &=
q_{C_1}-q_{C_2}+q_{C_3} \\
\epsilon_4 & = p_{C_3} &
\epsilon_{11} &= \nu_L \\
\epsilon_5 & = v_L-v_{C_2} &
\epsilon_{12} &= \nu_{C_1} \\
\epsilon_6 & = v_{C_1}-v_{C_2}+v_{C_3} &
\epsilon_{13} &= \nu_{C_2} \\
\epsilon_7 & = \frac{q_{C_1}}{C_1}-\frac{q_{C_3}}{C_3}&
\epsilon_{14} &= \nu_{C_3}.
\end{align*}

We can calculate the matrix $(\Sigma_{ij})=(\{\epsilon_i,\epsilon_j\}(q,v,p,\nu))$ as 
\[
(\Sigma_{ij})=\left[\begin{array}{cccccccccccccc} 0 & 0 & 0 & 0 & 0 & 0 &
0 & 0 & -1 & 0 & -L & 0 & 0 & 0 \\ 0 & 0 & 0 & 0 & 0 & 0 &
-\frac{1}{C_1} & 0 & 0 & -1 & 0 & 0 & 0 & 0 \\ 0 & 0 & 0 & 0 & 0 &
0 & 0 & 0 & 1 & 1 & 0 & 0 & 0 & 0\\ 0 & 0 & 0 & 0 & 0 & 0 &
\frac{1}{C_3} & 0 & 0 & -1 & 0 & 0 & 0 & 0 \\ 0 & 0 & 0 & 0 & 0 &
0 & 0 & 0 & 0 & 0 & 1 & 0 & -1 & 0 \\ 0 & 0 & 0 & 0 & 0 & 0 & 0 &
0 & 0 & 0 & 0 & 1 & -1 & 1 \\ 0 & \frac{1}{C_1} & 0 &
-\frac{1}{C_3} & 0 & 0 & 0 & 0 & 0 & 0 & 0 & 0 & 0 & 0 \\ 0 & 0 &
0 & 0 & 0 & 0 & 0 & 0 & 0 & 0 & \frac{1}{C_1} & 0 &
-\frac{1}{C_3} & 0\\ 1 & 0 & -1 & 0 & 0 & 0 & 0 & 0 & 0 & 0 & 0 & 0 &
0 & 0 \\ 0 & 1 & -1 & 1 & 0 & 0 & 0 & 0 & 0 & 0 & 0 & 0 & 0 & 0 \\
L & 0 & 0 & 0 & -1 & 0 & 0 & 0 & 0 & 0 & 0 & 0 & 0 & 0 \\0 & 0 & 0
& 0 & 0 & -1 & 0 & -\frac{1}{C_1} & 0 & 0 & 0 & 0 & 0 & 0 \\ 0 & 0
& 0 & 0 & 1 & 1 & 0 & 0 & 0 & 0 & 0 & 0 & 0 & 0 \\0 & 0 & 0 & 0 &
0 & -1 & 0 & \frac{1}{C_3} & 0 & 0 & 0 & 0 & 0 & 0 \end{array}
\right] ,\]
which is invertible, with inverse $\Sigma^{ij}$, then $\Sigma^{ik}\Sigma_{kj}=\delta_j^i$. This means that $x_0+\widetilde W_3$ is a second class submanifold, or equivalently, that all the constraints are second class. In particular, $x_0+\widetilde W_3$ is a symplectic submanifold, and we will call the symplectic form $\omega_c$. Then we have two ways of writing equations of motion on this submanifold.

On one hand, we can write the equations of motion on the symplectic manifold
$x_0 + \widetilde{W}_3$
in the form \eqref{eq:Ham_form_sympl_omegac},
\begin{equation*}
X_{\mathcal{E}|x_0 + \widetilde{W}_3}
= (\omega_c)^\sharp\left(d(\mathcal{E} | x_0 + \widetilde{W}_3)\right).
\end{equation*}
For this we are going to use the previous parametrization of $x_0 +
\widetilde{W}_3$ with coordinates $q_L$, $p_L$, so $\omega_c= d q_L \wedge d p_L$. Since the Lagrangian is given
by
\[\mathcal{L}(q,v)=\frac{1}{2}L(v_L)^2-\frac{1}{2}\frac{(q_{C_1})^2}{C_1}-
\frac{1}{2}\frac{(q_{C_2})^2}{C_2}-\frac{1}{2}\frac{(q_{C_3})^2}{C_3},\]
in terms of the parametrization the energy $\mathcal{E} =
p v - \mathcal{L} (q, v)$ restricted to $x_0 + \widetilde{W}_3$ is
\[
\left(\mathcal{E} | x_0 + \widetilde{W}_3 \right) (q_L, p_L)
=
\frac{1}{2}\frac{p_L^2}{L}
+
\frac{1}{2}q^2_L \left(\frac{1}{C_1 + C_3}
+
\frac{1}{C_2}
\right).
\]
Then the vector field
$X_{\mathcal{E}|x_0 + \widetilde{W}_3}$
is given in coordinates
$q_L$, $p_L$
by
\begin{align}\label{simplest1}
\dot{q}_L
&= \frac{\partial H}{\partial{p_L}}=\frac{p_L}{L} \\ \label{simplest2}
\dot{p}_L
& = -\frac{\partial
H}{\partial{q_L}}=-q_L\left(\frac{1}{C_1+C_3}
+
\frac{1}{C_2}
\right),
\end{align}
with
$H = \mathcal{E}|x_0 + \widetilde{W}_3$, for short.

We remark that the previous equations can easily be
deduced by elementary rules of circuit theory. Namely, one can use the formulas for capacitors in parallel and series and replace the three capacitors by a single one. The resulting circuit is very easy to solve, and the corresponding system is equivalent
to equations
(\ref{simplest1}) and (\ref{simplest2}). However, we must remark that this simplified system no longer accounts for the currents and voltages on the original capacitors.\\

On the other hand, we can find the evolution of all the variables $(q,v,p)$
using the extended energy (\ref{extendedenergy}) rather than the total energy, since $x_0 + \widetilde{W}_3$ is symplectic and using the last paragraph of remark \ref{physicalstate}.
Namely, we will calculate the vector field $X=d\mathcal{E}^{\sharp}+\lambda^j_{(3)}
X_{\epsilon_j}$ associated to the extended energy. 

We have $\lambda^i_{(3)}=\Sigma^{ij}\{\mathcal{E},\epsilon_j\}$,
$j=1,\ldots,14$, where
the column vector $[\{\mathcal{E},\epsilon_j\}]$ is
\[[\{\mathcal{E},\epsilon_j\}]=\left[0,\,-\frac{q_{C_1}}{C_1},\,-\frac{q_{C_2}}{
C_2},\,-\frac{q_{C_3}}{C_3},
\,0,\,0,\,0,\,0,\,0,\,0,\,0,\,0,\,0,\,0\right]^{T}.\]
For instance, we can calculate $\dot{p}_L$ and $\dot{v}_L$ as
\begin{align*}
\dot{p}_L & = X(p_L)  = (d\mathcal{E})^\sharp(p_L)+\lambda^j_{(3)}
X_{\epsilon_j}(p_L) = \{p_L,\mathcal{E}\} +
\lambda^j_{(3)}\{p_L,\epsilon_j\} = 0 + \lambda^9_{(3)} (-1) \\
&= -\frac{q_{C_1}}{C_1}-\frac{q_{C_2}}{C_2}\\
\dot{v}_L & = X(v_L) 
 = (d\mathcal{E})^\sharp(v_L)+\lambda^j_{(3)}
 X_{\epsilon_j}(v_L)
 = \{v_L,\mathcal{E}\} + \lambda^j_{(3)}\{v_L,\epsilon_j\} 
= 0 + \lambda^{11}_{(3)}\\
 & = \frac{1}{L}\left(-\frac{q_{C_1}}{C_1}-\frac{q_{C_2}}{C_2}\right)
\end{align*}
The complete coordinate expression of $X$ is
\begin{align*}
X
&=
(\dot{q}_L,\dot{q}_{C_1},\dot{q}_{C_2},\dot{q}_{C_3},\dot{v}_L,\dot{v}_{C_1},
\dot{v}_{C_2},\dot{v}_{C_3},\dot{p}_L,\dot{p}_{C_1},\dot{p}_{C_2},\dot{p}_{C_3},
\dot{\nu}_L,\dot{\nu}_{C_1},
\dot{\nu}_{C_2},\dot{\nu}_{C_3})\\
& =
\left(v_L,v_{C_1},v_{C_2},v_{C_3},
\frac{1}{L}\left(-\frac{q_{C_1}}{C_1}-\frac{q_{C_2}}{C_2}\right),
\frac{C_1}{C_1 + C_3}\frac{1}{L}
\left(-\frac{q_{C_1}}{C_1}-\frac{q_{C_2}}{C_2}\right),\right.\\
&
\left.\frac{1}{L}\left(-\frac{q_{C_1}}{C_1}-\frac{q_{C_2}}{C_2}\right),
\frac{C_3}{C_1 + C_3}\frac{1}{L}\left(-\frac{q_{C_1}}{C_1}-\frac{q_{C_2}}{C_2}\right),
-\frac{q_{C_1}}{C_1}-\frac{q_{C_2}}{C_2},0,0,0,0,0,0,0\right).
\end{align*}
Note that this is not a Hamiltonian vector field. However, if we use Theorem \ref{theorem321}, we can write the equations of motion in terms of the Dirac bracket and the abridged total energy corresponding to the chosen leaf $x_0 + \widetilde{W}_3$. In section \ref{sectionanextensionofetc} we will return to LC circuits and show how to write equations of motion on $M_c$ without specifying any leaf \emph{a priori}, by using an abridged total energy and the Dirac bracket associated to adapted constraints.\\

We remark that it is not always the case that there is uniqueness of solution. In fact, several branches of the circuit may be present with no capacitors or inductors, as a limit case. In the next section we show how to write a Poisson bracket description of the solution.

\section{An extension of the Dirac theory of constraints}\label{sectionanextensionofetc}

In sections \ref{section7} and \ref{mainresultsofdiracandgotaynester}, under precise regularity conditions, we studied
several geometric and algebraic objects related to the Dirac theory of
constraints
like primary and final constraints, first class and second class constraints, Dirac brackets, first class and
second class constraint submanifolds. 

The notions of first class and second class depend only on the final constraint submanifold $M_c$ and they do not depend on the primary constraint submanifold or the Hamiltonian, which are needed only to write equations of motion.
Then, as we have already said at the beginning of section \ref{subsectionV}, in order to study those notions one can start with
an abstract situation, given by a  symplectic manifold 
$(P,\Omega)$ and an arbitrary submanifold $S \subseteq P$, to be thought of as the final constraint submanifold.

In this section, motivated by
integrable nonholonomic mechanics and LC circuits theory, we extend those studies by replacing
$S$ by a submanifold
$\mathbf{S} \subseteq P$ regularly foliated by submanifolds $S$.
This means that 
$\mathbf{S}  = \cup_{S \in \mathfrak{S}} S$ and the map
$\mathbf{S} \rightarrow \mathfrak{S}$ given by $x \mapsto S$ iff $x \in S$ is a submersion.
To study the dynamics, one should consider a primary \textit{regularly foliated} constraint submanifold $\mathbf{S}^\prime$ whose leaves $S^\prime$ satisfy the condition $S = S^\prime \cap \mathbf{S}$,
and an energy  $\mathcal{E}\colon  P \rightarrow \mathbb{R}$.

The meaning of 
$\mathbf{S}$
can be interpreted
in connection with the constraint algorithm CAD for a Dirac dynamical system  (\ref{11drracds})
on the manifold $M$ in the case in which the Dirac structure $D$ is integrable, which gives a foliation of $M$.
We should compare $\mathbf{S}$ with the foliated submanifold $M_c$ appearing
in remark \ref{ImportantRemark}, while $\mathbf{S}^\prime$ should be compared with $M$ as a submanifold of $P$.

The results in section \ref{subsectionequationsof motion}
can be extended for a certain kind of Dirac dynamical systems. Namely, let $\mathbf{S}\subseteq \mathbf{S}'\subseteq P$ as before. Consider each leaf $S'_C$ endowed with the presymplectic form $\omega_C$ obtained as the pull-back of $\Omega$. Assume that the distribution $\ker \omega_C$ is regular and its dimension does not depend on $C$. Define the Dirac structure $D$ on $\mathbf{S}'$ by declaring that its presymplectic leaves are $(S'_C,\omega_C)$. We will consider the Dirac dynamical system \eqref{11drracds} for the case $M=\mathbf{S}'$, that is,
\begin{equation}\label{diracdirac2}
(x, \dot x)\oplus d\mathcal{E}(x)\in D_x.
\end{equation}
The results in section \ref{subsectionequationsof motion} would correspond to the case in which 
the foliation of $\mathbf{S}$ consists of a single leaf.

\subsection{Dirac brackets adapted to foliated constraint submanifolds}

\paragraph{Description of primary and final foliated constraint submanifolds.}

Let $(P, \Omega)$ be a symplectic manifold.
We are going to describe the \textit{primary foliated constraint submanifold} $\mathbf{S}^\prime$ and the \textit{final foliated constraint submanifold} $\mathbf{S}$ of $P$. We will assume that there is a one to one correspondence between the foliations such that
for each leaf $S^\prime \subseteq \mathbf{S}^\prime$ the corresponding leaf of $\mathbf{S}$ is 
$S=S^\prime \cap \mathbf{S}$. A motivation for this assumption comes from mechanics, where an energy on $P$ is given. In that context, each leaf $S\subseteq\mathbf{S}$ is the final constraint submanifold corresponding to the primary constraint submanifold $S'\subseteq\mathbf{S}'$ after applying a constraint algorithm.
These leaves will be parametrized by a certain vector $C$, and we will denote them by $S'_C$ and $S_C$ respectively. 
We are going to work on a suitable neighborhood $\mathcal{U}$ of $\mathbf{S}'$ in $P$, where all these submanifolds can be defined regularly by equations, as we will describe next.

For convenience, we will work with a more general situation where $\mathbf{S}'$ is a member of a \emph{family} of primary foliated constraint submanifolds parameterized by a certain vector $C_{(1,a^\prime)}$ and $\mathbf{S}$ is a member of a corresponding family of final foliated constraint submanifolds parameterized by a certain vector $C_{(1,a)}$. For $C_{(1,a^\prime)}=0$ and $C_{(1,a)}=0$ we obtain $\mathbf{S}'$ and $\mathbf{S}$, respectively.

Let
\begin{equation}\label{defofC}
C = (C_1,\dots,C_{a^\prime}, C_{a^\prime + 1},\dots,C_a, C_{a + 1},\dots,C_b) 
\equiv 
(C_{(1,a^\prime)}, C_{(a^\prime + 1, a)},C_{(a+1, b)})
\end{equation}
be a generic point of
$\mathbb{R}^b \equiv \mathbb{R}^{a^\prime}\times \mathbb{R}^{a - a^\prime}\times \mathbb{R}^{b - a}$, where $a'\leq a\leq b$. 
Denote 
$C_{(1,a)} = (C_1,\dots,C_a)$. 
Then we assume the following definitions by equations (regularly on $\mathcal{U}$) of the submanifolds $\mathbf{S}^{C_{(1,a^\prime)}}$
and
$\mathbf{S}^{C_{(1,a)}}$
and their foliations by leaves 
$S_C^\prime$ and $S_C$ respectively, which satisfy
$S_C = S_C^\prime \cap \mathbf{S}^{C_{(1,a)}}$, namely,
\begin{align}\label{calSprimeFo} 
\phi_i 
&= C_i, \, i = 1,\dots,a^\prime, \,\, \mbox{defines}\,\,\,\mathbf{S}^{C_{(1,a^\prime)}}\\\label{SprimeCFo}
\phi_i 
&= C_i, \, i = 1,\dots,a^\prime, a^\prime + 1,\dots,a, \,\, \mbox{defines}\,\,\,\mathbf{S}^{C_{(1,a)}}\\\label{SC}
\phi_i 
&= C_i, \, i = 1,\dots,a^\prime, a+1,\dots,b, \,\,\mbox{defines}\,\,\,  S_C^\prime\\\label{calS}
\phi_i 
&= C_i, \, i = 1,\dots,a, a+1,\dots,b, \,\, \mbox{defines}\,\,\, S_C
\end{align}
Since for each primary constraint submanifold $S'_C$ there is only one final constraint submanifold $S_C$, then $C_{(a'+1,a)}$ must be a function of $C_{(1,a')}$, and without loss of generality we will assume that this function maps 0 to 0.
Then for fixed $C_{(1,a^\prime)}$, $C_{(a^\prime+1, a)}$ is also fixed and one has, for each choice of 
$C_{(a+1,b)} = (C_{a+1},\dots,C_b)$, one leaf 
$S_C^\prime$ of $\mathbf{S}^{C_{(1,a^\prime)}}$ and the corresponding leaf 
$S_C = S_C^\prime \cap \mathbf{S}^{C_{(1,a)}}$ of $\mathbf{S}^{C_{(1,a)}}$.

\begin{assumption}\label{assumptionFOLIATION}
{}From now we assume that our foliated submanifolds $\mathbf{S}^\prime$ and $\mathbf{S}$ 
are $\mathbf{S}^{C_{(1,a^\prime)}}$ and $\mathbf{S}^{C_{(1,a)}}$
for the choice
$C_{(1,a^\prime)} = 0$ and $C_{(1,a)} = 0$, respectively, while
the
leaves
$S^\prime$ and $S$
are the leaves
$S_C^\prime$ and $S_C$
respectively, for the choice 
$C = \left( 0, 0, C_{(a+1,b)}\right)$, and $C_{(a+1,b)}$ varying arbitrarily, 
that is, $C\in\{0\}\times\{0\}\times\mathbb{R}^{b-a}$ (see Figure~\ref{fig:foliation}).
\end{assumption}
\begin{figure}[ht]
\centering
\includegraphics{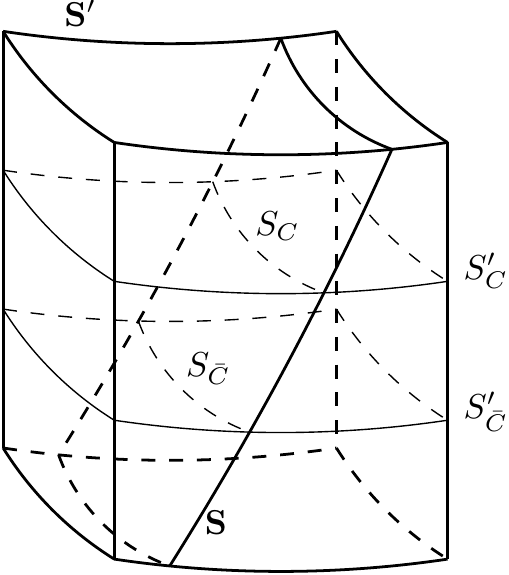}
  \caption{\small$\mathbf{S'}$ is foliated by the submanifolds $\{S'_C\}$, while $\mathbf{S}$ is foliated by $\{S_C\}$. Here we depict the leaves for two values $C$ and $\bar C$.}
  \label{fig:foliation}
\end{figure}

\begin{remark} The case considered by Dirac, studied in section \ref{mainresultsofdiracandgotaynester}, corresponds to $a= b$. It is in this sense that our theory extends Dirac's. In this case, the primary constraint  $\mathbf{S^\prime}$ and the final constraint $\mathbf{S}$ have only one leaf. More precisely, this particular case has been studied in Lemma \ref{mainlemma1section7}, with a different notation.  Our extension takes care of those Dirac dynamical systems coming for instance from integrable nonholonomic systems or LC circuit theory.
\end{remark}

\paragraph{Equations of motion as a collection of Hamilton's equations.}

Let us adopt the point of view that the Dirac dynamical system (\ref{diracdirac2}) is a collection of Gotay-Nester systems, 
one for each leaf, namely
\begin{equation}\label{wizardequationofmotion}
\Omega(x)(\dot{x}, \delta x) = d\mathcal{E}(x)(\delta x),
\end{equation}
where $(x,\dot{x}) \in T_x S_C$, for all $\delta x \in T_x S_C^\prime$.

Then one can apply the theory developed in section \ref{subsectionequationsof motion} leaf by leaf.
For each $C=(0,0,C_{(a+1,b)})$, the total Hamiltonian defined in \eqref{defnoftotalhamiltonian} should be replaced by the total energy depending on $C$, by definition,
\begin{equation*}%\label{defnoftotalenergy}
\mathcal{E}_{C,T} = \mathcal{E} + \lambda^i (\phi_i - C_i),\,\,\,
\mbox{sum over}\,\,\,
i = 1,\dots,a^\prime, a+1,\dots,b
\end{equation*}
and, for each such $C$, the 
$\lambda^i \in C^{\infty}(\mathcal{U})$, $i = 1,\dots,a^\prime, a+1,\dots,b$ 
must satisfy the conditions
$\{\mathcal{E}_{C,T}, \phi_j - C_j\}(x) = 0$, $j = 1,\dots,b$, for each $x \in S_C$ 
or, equivalently,
\begin{equation*}%\label{conditiononlambda}
 \{\mathcal{E}, \phi_j\}(x) +  \lambda^i  \{\phi_i, \phi_j\}(x) = 0, j= 1,\dots,b,
\end{equation*}
for each $x\in S_C$.
This shows in particular the interesting fact that the set of solutions, which is an affine space, depends on 
$x$ but not on $C$.
We assume, as part of our regularity conditions, that the solutions 
$(\lambda^1,\dots,\lambda^{a^\prime},\lambda^{a + 1},\dots,\lambda^b)$ form a nonempty affine bundle 
$\Lambda \rightarrow \mathbf{S}$.
Since we are working locally, by shrinking $\mathcal{U}$ conveniently we can assume that this bundle is trivial. Furthermore, we extend this bundle arbitrarily to a trivial affine bundle over $\mathcal{U}$, which we continue to denote $\Lambda \to \mathcal{U}$. 
Note that for each $C$, the total energy can be seen as a function $\mathcal{E}_{C,T}\colon \Lambda \to \mathbb{R}$.  

For each section $\widetilde\lambda$ of $\Lambda$ one can define
$\mathcal{E}_{C,\widetilde\lambda,T}=\widetilde\lambda^*\mathcal{E}_{C,T}\colon \mathcal{U} \to \mathbb{R}$, that is, $\mathcal{E}_{C,\widetilde\lambda,T}(x) = \mathcal{E}(x) + \widetilde\lambda^i (x)( \phi_i(x) - C_i)$. One has an equation of motion on $\mathcal{U}$,
\begin{equation}\label{C-equationofmotion}
 X_{\mathcal{E}_{C,\widetilde\lambda,T}}(x) = (d \mathcal{E}_{C,\widetilde\lambda,T})^\sharp(x),
\end{equation}
which should  be interpreted properly, as follows.
In the definition of $\mathcal{E}_{C,T}$, $x$ and $C$ are independent variables. 
For a given
$x_0 \in \mathcal{U}$, let 
$C_0 = (\phi_1(x_0),\dots,\phi_b(x_0))$; then $\mathcal{E}_{C_0,\widetilde\lambda,T}(x)$ is a function of $x$ and
\begin{equation*}%\label{C0-equationofmotion}
 X_{\mathcal{E}_{C_0,\widetilde\lambda,T}}(x_0) = (d \mathcal{E}_{C_0,\widetilde\lambda,T})^\sharp(x_0)
 = 
 d\mathcal{E}^\sharp \left(x_0\right)  + \widetilde\lambda^i(x_0) d\phi_i^\sharp \left(x_0\right).
\end{equation*}
This means that equation (\ref{C-equationofmotion}) should not be naively interpreted as Hamilton's equation, and therefore the algorithms devised by Dirac and Gotay-Nester cannot be generalized in a direct way for the foliated case. In fact, even though the variable $C$ is a function of $x$, it should be considered a constant in order to calculate the vector field at the point $x$. 
Note that, by  construction, the condition
$ X_{\mathcal{E}_{C_0,\widetilde\lambda,T}}(x_0) \in T_{x_0} S_{C_0}$
is satisfied, which means that the leaves $S_{C_0}$ of $\mathbf{S}$
are preserved by the motion, as expected.

\begin{remark} For any given solution $x(t)$ of equations of motion (\ref{wizardequationofmotion}) there exist uniquely determined 
$\lambda^i_t$, $i = 1,\dots,a^\prime, a+1,\dots,b$ such that
\[
\dot{x}(t) = d\mathcal{E}^\sharp \left(x(t)\right)  + \lambda^i_t d\phi_i^\sharp \left(x(t)\right).
\]
On the other hand, for any given time-dependent section $\widetilde\lambda^i_t(x) $, $i = 1,\dots,a^\prime, a+1,\dots,b$ of $\Lambda$ one has a
time-dependent vector field 
\begin{equation}\label{nonHamiltonian-timedependent}
d\mathcal{E}^\sharp \left(x\right)  + \widetilde\lambda^i_t(x) d\phi_i^\sharp \left(x\right)
\end{equation}
whose integral curves are solutions to the equations of motion. One can show that under the strong regularity conditions that we assume in this paper, all the solutions of the equations of motion can be represented in this way, at least locally.
\end{remark}

Summarizing, the vector fields \eqref{nonHamiltonian-timedependent} obtained from \eqref{C-equationofmotion} are not in general Hamiltonian vector fields with respect to the canonical bracket, as we have indicated before.
Now we will see how to recover the Hamiltonian character of the equation of motion. We are going to write it as Hamilton's equation in terms of a certain  Poisson bracket whose symplectic leaves are isomorphic to $\mathcal{U}$. 

\paragraph{Equations of motion in terms of a Poisson bracket.}
Let us consider an extended Poisson manifold, namely the product Poisson manifold 
$(\mathcal{U}\times\{0\}\times\{0\}\times \mathbb{R}^{b-a}, \{\, , \,\}^{(b)})$ with the canonical  Poisson bracket $(\Omega)^{-1}$ on $\mathcal{U}$ and the $0$ Poisson bracket on $\{0\}\times\{0\}\times\mathbb{R}^{b-a}$.  
Let $\mathcal{E}_T\colon \Lambda\times \{0\}\times\{0\}\times\mathbb{R}^{b-a} \to \mathbb{R}$ be defined by 
$\mathcal{E}_T (\lambda,C) := \mathcal{E}_{C,T}(\lambda)$. Given a section $\widetilde\lambda$ of $\Lambda$ one can define an energy function $\mathcal{E}_{\widetilde\lambda,T}\colon \mathcal{U} \times\{0\}\times\{0\}\times \mathbb{R}^{b-a}\to \mathbb{R}$ as $\mathcal{E}_{\widetilde\lambda,T}(x,C):=\mathcal{E}_{C,T}(\widetilde\lambda(x))=\mathcal{E}_{C,\widetilde\lambda,T}(x)$. Consider its corresponding Hamiltonian vector field 
\[X_{\mathcal{E}_{\widetilde\lambda,T}}(x,C)=(X_{\mathcal{E}_{C,\widetilde\lambda,T}}(x),0),\]
on $\mathcal{U}\times \{0\}\times\{0\}\times \mathbb{R}^{b-a}$ with respect to the Poisson bracket $\{\, , \,\}^{(b)}$,
where $X_{\mathcal{E}_{C,\widetilde\lambda,T}}(x)$ is defined in \eqref{C-equationofmotion}. We remark that even though the first component cannot be interpreted in general as a Hamiltonian vector field on $\mathcal{U}$, as we have observed before, $X_{\mathcal{E}_{\widetilde\lambda,T}}$ is Hamiltonian on the extended Poisson manifold.

Consider the function
\begin{equation*}
\mathcal{C}_{(0,0)}(x) = \left(0,0,\phi_{a+1}(x) ,\dots,\phi_b(x) \right)\in  \mathbb{R}^{a^\prime}\times \mathbb{R}^{a - a^\prime}\times \mathbb{R}^{b - a}
\end{equation*}
defined on $\mathcal{U}$. Clearly, $\operatorname{graph}(\mathcal{C}_{(0,0)})$ is a submanifold of
$\mathcal{U} \times \{0\}\times\{0\}\times\mathbb{R}^{b-a}$  diffeomorphic to $\mathcal{U}$.

Since $X_{\mathcal{E}_{C,\widetilde\lambda,T}}$ preserves $S_C$ for each $C$, it is straightforward to prove that $X_{\mathcal{E}_{\widetilde\lambda,T}}$ preserves $\operatorname{graph}(\mathcal{C}_{(0,0)}|S_C)$, which for each $C$ is a copy of the level set $\mathcal{C}_{(0,0)}^{-1}(C)$ embedded in $\mathcal{U}\times \{0\}\times\{0\}\times\mathbb{R}^{b-a}$, and therefore $X_{\mathcal{E}_{\widetilde\lambda,T}}$ preserves $\operatorname{graph}(\mathcal{C}_{(0,0)}|\mathbf{S})$,  which is a copy of $\mathbf{S}$.

Consider the evolution equation for a quantity $F$ on 
$\mathcal{U}\times \{0\}\times\{0\}\times\mathbb{R}^{b-a}$ given by
\begin{equation*}%\label{PC-equationofmotion}
 \dot{F} (x,C) = \{F, \mathcal{E}_{\widetilde\lambda,T}\}^{(b)}(x,C).
\end{equation*}
Note that this is an equation in Poisson form on 
$\mathcal{U}\times \{0\}\times\{0\}\times\mathbb{R}^{b-a}$, not on $\mathbf{S}\times \{0\}\times\{0\}\times\mathbb{R}^{b-a}$. 
Note that for $a=b$ and $P=T^*Q$, then $\mathbf{S}=S$ and  we recover the situation in Dirac's theory, where Hamilton's equations of motion are written in a neighborhood of the final constraint submanifold in $T^\ast Q$ rather than on the constraint submanifold itself.

\paragraph{Local equations of motion in Poisson form with respect to the Dirac bracket.}

As we have remarked above, the equations of motion in terms of the canonical bracket for the foliated case are not naive extensions of the Dirac or Gotay-Nester procedures. A similar situation occurs with the equation $\dot g \approx [g,H_T]^*$ in \cite{Dirac1964}, page 42. As we will show next, we can extend this equation to the foliated case by using adapted constraints and the abridged total energy.

In the definition of $S^\prime_C$ and $S_C$ take $C= 0$. Then we can apply the procedure of section \ref{subsectionequationsof motion} with
$S = S_0$ and $S^\prime = S^\prime_0$, and obtain second class constraints
$\chi_1, \dots,\chi_{2s}$ among 
$\phi_i$, $i = 1, \dots,b$
adapted to $S^\prime_0$ and also we can choose some primary constraints
among $\phi_i, i = 1, \dots,a^\prime, a + 1, \dots,b$, say w.l.o.g.\ $\phi_k$, $k=1, \dots,a'-s'_A,a+1, \dots ,b-s'_B$. Observe that here $s'=s'_A+s'_B$ has the same meaning as in section \ref{subsectionequationsof motion}, this time considering the constraints $\phi_1=0, \dots,\phi_b=0$. Then the equations of motion on $S_0$  can be written in the form
(\ref{eq:poisson_dirac_evolution_thm}).

Now, for any small enough $C$, which defines $S_C$, one can readily see that $\chi_1 - B_1, \dots,\chi_{2s}- B_{2s}$ are second class constraints adapted to $S'_C$ where $B_1, \dots, B_{2s}$ are the components of $C$ corresponding to the choice $\chi_1 , \dots,\chi_{2s}$.
This implies immediately that the Dirac bracket for
$S_C$ does not depend on $C$. 
On the other hand, each
$\phi_k - C_k$, $k=1, \dots,a'-s'_A,a+1, \dots ,b-s'_B$, is a primary constraint for 
$S^\prime_C$.

Using the previous facts and the fact that the $\lambda^{\prime k}$ are arbitrary (even time-dependent) parameters, as it happens with the $\lambda'^i$ in equation (\ref{eq:poisson_dirac_evolution_thm}), one can conclude that the equation of motion on 
each $S_C$
is given by 
\begin{align}
\dot F
&= \{F,\mathcal{E}\}_{(\chi)}+\lambda'^k \{F,\phi_{k}\}_{(\chi)}\nonumber\\
&= \{F,\mathcal{E}+\lambda'^k\phi_{k}\}_{(\chi)},  \label{eq:poisson_dirac_evolution_thmfoliated}
\end{align}
sum over $k=1, \dots,a'-s'_A,a+1, \dots ,b-s'_B$.

This means in particular that each $S_C$ is preserved by the motion, or that the Hamiltonian vector fields
\begin{equation}\label{eq:vector_field_foliated}
X_{(\chi), \mathcal{E} + \lambda'^k\phi_{k}},
\end{equation}
defined on $U$,
are tangent to the $S_C$.

Note that $\mathcal{E} + \lambda'^k\phi_{k}$ is the abridged total energy $\mathcal{E}_{AT}$ for $S_0$ and $S'_0$ as defined in (\ref{abridged_total_energy}). 

In conclusion, using adapted constraints, the abridged total energy and the Dirac bracket yields a very simple procedure for writing the equations of motion, because one only needs to write the abridged total energy for one leaf, say $C=0$, and  \eqref{eq:poisson_dirac_evolution_thmfoliated} gives the equation of motion for all nearby leaves.

We have proven the following theorem, which extends Theorem \ref{theorem321}.

\begin{theorem}\label{thm321EXTENDED}
Let  $(P, \Omega)$ be a symplectic manifold and let 
$\mathbf{S}^\prime  \supseteq \mathbf{S}$ 
be given primary and final foliated constraint submanifolds, with leaves $S'_C$ and 
$S_C=S^\prime_C  \cap  \mathbf{S}$, respectively, as described by
(\ref{defofC})--(\ref{calS}). 
Assume that all the hypotheses of Theorem  \ref{theorem321} are satisfied for the submanifolds
$S'_0 \supseteq S_0$. In addition, suppose that the number $s'$ appearing in Assumption \ref{Lambda} is the same for all $S_C$ for $C$ close enough to $0$. From the validity of Assumption \ref{2s_constant} for $S'_0$, it is immediate to see that the number $2s$ is also the same for all $S'_C$ for $C$ close enough to $0$.
Choose second class constraints $(\chi'_{1}, \dots, \chi'_{s'},\chi''_{s'+1},\dots,\chi''_{2s})$ adapted to $S'_0$
and also $\phi_{k}$, $k=1, \dots,a'-s'_A,a+1, \dots ,b-s'_B$, as in Theorem  \ref{theorem321}. 

Let $\mathcal{E}\colon P \to \mathbb{R}$ be an energy function, and consider the Dirac dynamical system \eqref{diracdirac2}. Define the abridged total energy $\mathcal{E}_{AT}=\mathcal{E} + \lambda'^k\phi_{k}$ for $S_0$ and $S'_0$ as in Theorem  \ref{theorem321}. Then, each $x_0\in \mathbf{S}$ has an open neighborhood $U$ such that the vector field $X_{(\chi),\mathcal{E}_{AT}}$ (equation \eqref{eq:vector_field_foliated}) defined on $U$, when restricted to the final foliated constraint submanifold $\textbf{S}\cap U$,
represents equations of motion of the Dirac dynamical system
on $\textbf{S}\cap U$. Moreover, the evolution of the system preserves the leaves $S_C\cap U$.
In addition, equation
\eqref{eq:poisson_dirac_evolution_thmfoliated} gives equations of motion on each $S_C\cap U$. Since $X_{(\chi),\mathcal{E}_{AT}}$ is tangent to $S_C\cap U$, then the evolution of a function $f$ on $S_C\cap U$ is given by \eqref{eq:poisson_dirac_evolution_thmfoliated} for any $F$ such that $f=F|S_C\cap U$.  Also, each $S_C\cap U$ is contained in a unique symplectic leaf of the Dirac bracket, defined by  $\chi'_{1}=B_1, \dots, \chi'_{s'}=B_{s'},\chi''_{s'+1}=B_{s'+1},\dots,\chi''_{2s}=B_{2s}$.
\end{theorem}

\begin{remark}
\cite[page 42]{Dirac1964}, 
wrote the equation
\[
 \dot{g} \approx \{g, H_T\}^*
\]
which represents the dynamics on the final constraint submanifold, in terms of the total Hamiltonian and the Dirac bracket. This equation does not gives the correct dynamics for the case of foliated constraint submanifolds. 
As we have seen, for the foliated case, one has, instead a similar equations of motion (\ref{eq:poisson_dirac_evolution_thmfoliated}) in terms of the Dirac bracket and the abridged total energy $\mathcal{E}_{AT}$.
\end{remark}
\paragraph{The concrete example of an LC circuit of section \ref{sectionexamples} revisited.}
Consider the symplectic leaf $S$ of $M_3$ defined parametrically by equations
(\ref{parametrizationofsymplleaf1})--(\ref{parametrizationofsymplleaf5}). The parameter space carries the canonical symplectic form $dq_L \wedge dp_L$, which coincides with the pullback of the canonical symplectic form on $T^* TQ$ to the parameter space. Since the matrix
$\Sigma$ is invertible, all constraints are second class constraints. 
The number $b-s'_B$ in equation (\ref{eq:poisson_dirac_evolution_thmfoliated}) is $0$, then
the abridged total energy is simply the energy $\mathcal{E}$.

Then it is easy to calculate the Dirac bracket as the canonical  bracket on the parameter space 
of the
functions (\ref{parametrizationofsymplleaf1})--(\ref{parametrizationofsymplleaf5}). Let
$y = (q,v,p,\nu)$, then the matrix of Dirac brackets of these variables can be written in block form as
\[
(\{y^i,y^j\}_\chi)=
\left[\begin{array}{ccc}
0_{4\times 4} & A & 0_{4\times 7}\\
-A^T & 0_{5\times 5} & 0_{5\times 7}\\
0_{7\times 4} & 0_{7\times 5} & 0_{7\times 7}
\end{array}\right]
\]
where
\[
A=\left[\begin{array}{ccccc}
%Q_L
 \frac{1}{L} & \frac{C_1}{(C_1 + C_3)L} &  \frac{1}{L} &  \frac{C_3}{(C_1 + C_ 3)L}
& 1
\\[6pt]
%Q_C1
 \frac{C_1}{(C_1 + C_3)L} &  \frac{C_1^2}{(C_1 + C_3)^2 L} &  \frac{C_1}{(C_1 + C_3)L} & 
 \frac{C_1 C_3}{(C_1 + C_3)^2 L}
&  \frac{C_1}{(C_1 + C_3)}
\\[6pt]
%Q_C2
 \frac{1}{L} &  \frac{C_1}{(C_1 + C_3)L} &  \frac{1}{L} &  \frac{C_3}{(C_1 + C_3)L}
& 1 
\\[6pt]
%Q_C3
 \frac{C_3}{(C_1 + C_3)L} &  \frac{C_1 C_3}{(C_1 + C_3)^2 L} &  \frac{C_3}{(C_1 + C_3)L} & 
 \frac{ C_3^2}{(C_1 + C_3)^2 L}
& \frac{C_3}{(C_1 + C_3)}
\end{array}\right].
\]

The evolution equations for the vector 
\[
(y^1, \dots,y^{16})
=
(q_L, q_{C_1}, q_{C_2}, q_{C_3},
 v_L, v_{C_1}, v_{C_2}, v_{C_3}, 
p_L, p_{C_1}, p_{C_2}, p_{C_3},
\nu_L, \nu_{C_1}, \nu_{C_2}, \nu_{C_3})
\]
with initial condition belonging to $M_3$, that is, satisfying
(\ref{mtresej}):
$p_L=Lv_L,\,p_{C_1}=p_{C_2}=p_{C_3}=0,\,
v_L=v_{C_2},\,v_{C_1}=v_{C_2}-v_{C_3},\,
\frac{q_{C_1}}{C_1}=\frac{q_{C_3}}{C_3},\,
\frac{v_{C_1}}{C_1}=\frac{v_{C_3}}{C_3}$,
is given by the product of the matrix 
$(\{y^i, y^j \}_{\chi})$
by the vector
\begin{multline*}
\left(\frac{\partial \mathcal{E}}{\partial y^1}, \dots,\frac{\partial \mathcal{E}}{\partial y^{16}}\right)
=\\
\left(0, - \frac{q_{C_1}}{C_1}, - \frac{q_{C_2}}{C_2}, - \frac{q_{C_3}}{C_3}, 
p_L - v_L, p_{C_1}, p_{C_2}, p_{C_3},
 v_L, v_{C_1}, v_{C_2}, v_{C_3},
0,0,0,0\right).
\end{multline*}

\section{Conclusions and future work}\label{conclusions}
We have shown that Dirac's work on constrained systems can be extended for cases where the primary constraint has a given foliation. This extends the applicability of the theory for
Dirac dynamical systems, like LC circuits where the constraints may come from singularities of the Lagrangian and, besides, from Kirchhoff's Current Law. Throughout the paper we combine ideas of an algebraic character from Dirac's theory with some more geometrically inspired ideas from Gotay-Nester work. 
In particular, we use Dirac brackets adapted to the foliation as well as a Constraint Algorithm for Dirac dynamical systems (CAD), which is an extension of the Gotay-Nester algorithm. The results were proven locally and under regularity conditions. It is our purpose to study in the future the globalization of the results of the paper as well as the singular cases, using IDE techniques, with applications.

\section{Acknowledgments}\label{acknowledgments}
This paper has been inspired by Jerry Marsden's invaluable participation in part of it.
We thank the following institutions for providing us with means to work on this paper:
Universidad Nacional del Sur (projects PGI 24/L075 and PGI 24/ZL06);
Universidad Nacional de la Plata; 
Agencia Nacional de Promoci\'on Cient\'ifica y Tecnol\'ogica, Argentina (projects PICT 2006-2219 and PICT 2010-2746);
CONICET, Argentina (projects PIP 2010--2012 11220090101018 and X553 Ecuaciones Di\-fe\-ren\-cia\-les Impl\'icitas);
European Community, FP7 (project IRSES ``GEOMECH'' 246981).

\appendix

\section{Appendix}\label{Pontryagin_embedding}
\begin{lemma} There is a canonical inclusion $\varphi\colon TQ\oplus T^*Q \to T^*TQ$. In addition, consider the canonical two-forms $\omega_{T^*Q}$ and $\omega_{T^*TQ}$ on $T^*Q$ and $T^*TQ$ respectively, the canonical projection $\operatorname{pr}_{T^*Q}\colon TQ\oplus T^*Q \to T^*Q$, and define the presymplectic two-form $\omega=\operatorname{pr}_{T^*Q}^*\omega_{T^*Q}$ on $TQ\oplus T^*Q$. Then the inclusion preserves the corresponding two-forms, that is, $\omega=\varphi^*\omega_{T^*TQ}$.
\end{lemma}
\begin{proof}
If $\tau_Q\colon TQ \to Q$ and $\tau_{TQ}\colon TTQ \to TQ$ are the tangent projections, we can consider the dual tangent rhombic
\[
\xymatrix{
&TTQ\ar[dl]_{\tau_{TQ}}\ar[dr]^{T\tau_Q}&\\
TQ\ar[dr]_{\tau_Q}&&TQ\ar[dl]^{\tau_Q}\\
&Q&}
\]
Define $\varphi\colon TQ\oplus T^*Q \to T^*TQ$ by $\varphi(v_q\oplus \alpha_q)\in T^*_{v_q}TQ$, 
\[
\varphi(v_q\oplus \alpha_q)\cdot w_{v_q}=\alpha_q\cdot T\tau_Q(w_{v_q}),
\]
for $w_{v_q}\in T_{v_q}TQ$. Here $v_q\oplus \alpha_q$ denotes an element in the Pontryagin bundle over the point $q\in Q$. Note that the following diagram commutes.
\[
\xymatrix{
TQ\oplus T^*Q \ar[r]^\varphi\ar[d]_{\operatorname{pr}_{TQ}}&T^*TQ\ar[d]^{\pi_{TQ}}\\
TQ\ar[r]^{\operatorname{Id}_{TQ}}&TQ
}
\]
Let us see that $\varphi$ is an injective vector bundle map from the bundle $\operatorname{pr}_{TQ}\colon TQ\oplus T^*Q \to TQ$ to the cotangent bundle $\pi_{TQ}\colon T^*TQ \to TQ$, over the identity of $TQ$. The last part of this assertion follows from the commutative diagram above.

First, if $\varphi(v_q\oplus \alpha_q)=\varphi(v'_{q'}\oplus \alpha'_{q'})$ then both sides are in the same fiber $T^*_{v_q}TQ=T^*_{v'_{q'}}TQ$, so $v_q=v'_{q'}$. Also, for all $w_{v_q}\in T^*_{v_q}TQ$ we have
\[
\varphi(v_q\oplus \alpha_q)\cdot w_{v_q}=\varphi(v_q\oplus \alpha'_q)\cdot w_{v_q}
\]
or
\[
\alpha_q\cdot T\tau_Q(w_{v_q})=\alpha'_q\cdot T\tau_Q(w_{v_q}).
\]
Since $\tau_Q\colon TQ \to Q$ is a submersion, we have $\alpha_q=\alpha'_q$ and $\varphi$ is injective.

Second, $\varphi$ is linear on each fiber, since
\[
\varphi(v_q\oplus (\alpha_q+\lambda\alpha'_q))\cdot w_{v_q}=(\alpha_q+\lambda\alpha'_q)\cdot T\tau_Q(w_{v_q})=\varphi(v_q\oplus \alpha_q)\cdot w_{v_q}+\lambda \varphi(v_q\oplus \alpha'_q)\cdot w_{v_q}.
\]

For the second part of the lemma, let us recall the definition of the canonical one-form on $\theta_{T^*Q}\in \Omega^1(T^*Q)$. For $\alpha_q\in T^*Q$, $\theta_{T^*Q}(\alpha_q)$ is an element of $T^*_{\alpha_q}T^*Q$ such that for $w_{\alpha_q}\in T_{\alpha_q}T^*Q$,
\[
\theta_{T^*Q}(\alpha_q) \cdot w_{\alpha_q}=\alpha_q(T\pi_Q(w_{\alpha_q})),
\]
where $\pi_Q\colon T^*Q \to Q$ is the cotangent bundle projection.

With a similar notation, the canonical one-form $\theta_{T^*TQ}\in \Omega^1(T^*TQ)$ is given by
\[
\theta_{T^*TQ}(\alpha_{v_q}) \cdot w_{\alpha_{v_q}}=\alpha_{v_q}(T\pi_{TQ}(w_{\alpha_{v_q}})).
\]
Pulling back these forms to the Pontryagin bundle by $\varphi$ and the projection $\operatorname{pr}_{T^*Q}\colon TQ\oplus T^*Q \to T^*Q$, we obtain the same one-form. Indeed, for $w\in T_{v_q\oplus \alpha_q}(TQ\oplus T^*Q)$, we get on one hand
\begin{equation*}
(\operatorname{pr}^*_{T^*Q}\theta_{T^*Q})(v_q\oplus \alpha_q) \cdot w=\theta_{T^*Q}(\alpha_q) \cdot T\operatorname{pr}_{T^*Q}(w)=\alpha_q \cdot T(\pi_Q\circ \operatorname{pr}_{T^*Q})(w),
\end{equation*}
and on the other hand
\begin{equation*}
\begin{split}
(\varphi^*\theta_{T^*TQ})(v_q\oplus \alpha_q) \cdot w=\theta_{T^*TQ}(\varphi(v_q\oplus \alpha_q)) \cdot T\varphi(w)=\\
\varphi(v_q\oplus \alpha_q) \cdot T\pi_{TQ}(T\varphi(w))=\alpha_q \cdot T(\tau_Q\circ\pi_{TQ}\circ\varphi)(w).
\end{split}
\end{equation*}
However, the following diagram commutes
\[
\xymatrix{
TQ\oplus T^*Q\ar[rr]^\varphi\ar[d]_{\operatorname{pr}_{T^*Q}}\ar[drr]_{\operatorname{pr}_{TQ}}& &T^*TQ\ar[d]^{\pi_{TQ}}\\
T^*Q\ar[dr]_{\pi_Q}&&TQ\ar[dl]^{\tau_Q}\\
&Q&
}
\]
so $\pi_Q\circ \operatorname{pr}_{T^*Q}=\tau_Q\circ\pi_{TQ}\circ\varphi$ and therefore $\operatorname{pr}^*_{T^*Q}\theta_{T^*Q}=\varphi^*\theta_{T^*TQ}$. Taking (minus) the differential of this identity, we obtain $\omega=\varphi^*\omega_{T^*TQ}$.
\end{proof}

In local coordinates, if we denote the elements of $TQ\oplus T^*Q$ and $T^*TQ$ by $(q,v,p)$ and $(q,v,p,\nu)$ respectively, then $\varphi(q,v,p)=(q,v,p,0)$.

\newcommand\oneletter[1]{#1}\newcommand\Yu{Yu}\newcommand{\ct}{t}\def\cprime{$'$}
  \def\polhk#1{\setbox0=\hbox{#1}{\ooalign{\hidewidth
  \lower1.5ex\hbox{`}\hidewidth\crcr\unhbox0}}} \def\cprime{$'$}

\end{document}